\documentclass[12pt,reqno]{amsart}

\usepackage{epsfig}
\usepackage{amsmath, amsthm, amsfonts}
\usepackage{graphicx}

\numberwithin{equation}{section}

\newcommand{\R}{{\mathbb R}}
\newcommand{\C}{{\mathbb C}}

\newcommand{\Z}{{\mathbb Z}}

\renewcommand{\Re}{{\operatorname{Re\,}}}
\renewcommand{\Im}{{\operatorname{Im\,}}}

\newcommand{\Ai}{{\operatorname{Ai}}}
\newcommand{\sn}{{\operatorname{sn}}}
\newcommand{\cn}{{\operatorname{cn}}}
\newcommand{\dn}{{\operatorname{dn}}}

\newcommand{\I}{{\bold I}}

\newcommand{\sign}{{\operatorname{sgn}\,}}
\newcommand{\Res}{{\operatorname{Res}\,}}
\newcommand{\al}{\alpha}
\newcommand{\be}{\beta}
\newcommand{\ga}{\gamma}

\newcommand{\ep}{\varepsilon}

\newcommand{\De}{\Delta}

\newcommand{\sg}{\sigma}
\newcommand{\Sg}{\Sigma}
\newcommand{\om}{\omega}
\newcommand{\Om}{\Omega}
\renewcommand{\th}{\vartheta}
\newcommand{\z}{\zeta}

\newtheorem{theo}{{\sc \bf Theorem}}[section]

\newtheorem{lem}[theo]{{\sc \bf Lemma}}
\newtheorem{prop}[theo]{{\sc \bf Proposition}}

\newenvironment{rem}{\medskip\noindent{\it Remark:\/} }{\medskip}

\begin{document}

\title[Exact solution of the six-vertex model]
{Exact solution of the six-vertex model with domain wall boundary conditions.
Antiferroelectric phase}

\author{Pavel Bleher}
\address{Department of Mathematical Sciences,
Indiana University-Purdue University Indianapolis,
402 N. Blackford St., Indianapolis, IN 46202, U.S.A.}
\email{bleher@math.iupui.edu}

\author{Karl Liechty}
\address{Department of Mathematical Sciences,
Indiana University-Purdue University Indianapolis,
402 N. Blackford St., Indianapolis, IN 46202, U.S.A.}
\email{kliechty@math.iupui.edu}

\thanks{The first author is supported in part
by the National Science Foundation (NSF) Grant DMS-0652005.}

\thanks{Both authors would like to thank the Centre de Recherches Math\'{e}matiques at l'Universit\'{e} de Montr\'{e}al, where much of this work was performed, for its hospitality during their visit in the Fall of 2008.}

\date{\today}

\begin{abstract} 
We obtain the large $n$ asymptotics  of the partition function $Z_n$ of the 
six-vertex model with domain wall boundary conditions 
in the antiferroelectric phase region, with the weights
$a=\sinh(\ga-t), \;
b=\sinh(\ga+t), \;
c=\sinh(2\ga), \;
|t|<\ga$.
We prove the conjecture of Zinn-Justin, that as $n\to\infty$, 
$Z_n=C\th_4(n\om) F^{n^2}[1+O(n^{-1})]$, where $\om$ and $F$ are given by explicit
expressions in $\ga$ and $t$, and $\th_4(z)$ is the Jacobi theta function. 
The proof is based on the Riemann-Hilbert approach to the large $n$ asymptotic expansion
of the underlying discrete orthogonal polynomials and on the Deift-Zhou nonlinear
steepest descent method. 
\end{abstract}

\maketitle

\section{Introduction and formulation of the main result}

\subsection{Definition of the model}
The six-vertex model, or the model of two-dimensional ice, is stated on a square $n\times n$
lattice with arrows on edges. The arrows obey the rule that at every vertex there 
are two arrows 
pointing in and two arrows pointing out. Such a rule is sometimes 
called the {\it ice-rule}. There are only six possible configurations of arrows at each 
vertex, hence the name of the model, see Fig.~1. 
%
\begin{center}
\begin{figure}[h]\label{arrows}
\begin{center}
\scalebox{0.52}{\includegraphics{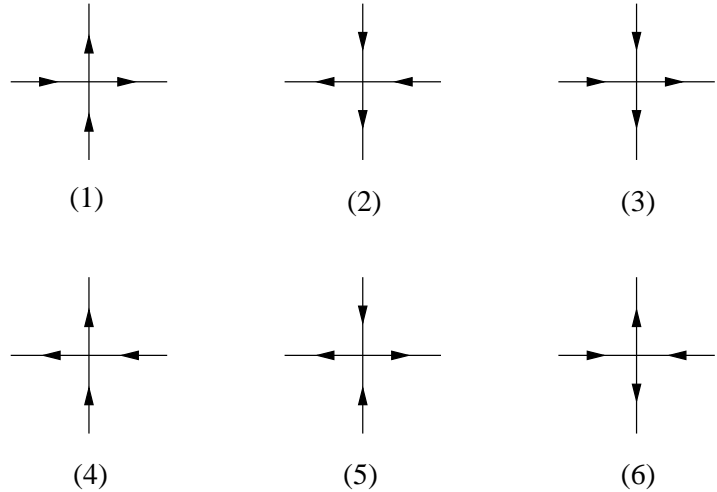}}
\end{center}
  \caption[The six arrow configurations allowed at a vertex]{The six arrow configurations allowed at a vertex.}
 \end{figure}
\end{center}

We will consider the {\it domain wall boundary conditions} (DWBC), 
in which the arrows on the upper and lower boundaries point into the square, 
and the ones on the left and right boundaries point out. 
One possible configuration with DWBC on the $4\times 4$ lattice is shown on Fig.~2.
\begin{center}
 \begin{figure}[h]\label{DWBC}
\begin{center}
   \scalebox{0.52}{\includegraphics{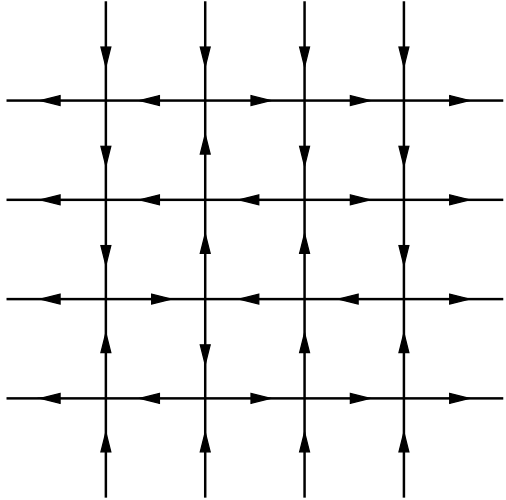}}
\end{center}
        \caption[An example of $4\times4$ configuration]
{An example of $4\times4$ configuration with DWBC.}
   \end{figure}
\end{center}

For each possible vertex state we assign a weight $w_i,\; i=1,\dots,6$, 
and define, as usual, the partition function as a sum over all possible 
arrow configurations of the product of the vertex weights,
\begin{equation}\label{lattice_11}
Z_n=\sum_{{\rm arrow\; configurations}\;\sigma}w(\sigma),
\qquad w(\sigma)=\prod_{x\in V_n} w_{t(x;\sg)}=\prod_{i=1}^6w_i^{N_i(\sigma)},
\end{equation}
where $V_n$ is the $n\times n$ set of vertices,
$t(x;\sg)\in\{1,\ldots,6\}$ is the vertex-type of configuration $\sg$
at vertex $x$ according to Fig. 1, and
$N_i(\sigma)$ is the number of vertices of  type $i$ in  
the configuration $\sg$. The sum is taken over all possible configurations
obeying the given boundary condition. The Gibbs measure is defined then
as
\begin{equation}\label{lattice_12}
\mu_n(\sg)=\frac{w(\sg)}{Z_n}\,.
\end{equation}
Our main goal is to obtain the large $n$ asymptotics of the partition function $Z_n$.

 The six-vertex model has six parameters: the weights $w_i$. By using some conservation
laws it can be reduced to only two parameters (see, e.g., \cite{FS}, \cite{BL1}, and \cite{AR}). 
Namely, we have that
\begin{equation}\label{cl_14}
Z_n(w_1,w_2,w_3,w_4,w_5,w_6)=C(n)Z_n(a,a,b,b,c,c),
\end{equation}
and 
\begin{equation}\label{cl_15}
\mu_n(\sg;w_1,w_2,w_3,w_4,w_5,w_6)=\mu_n(\sg;a,a,b,b,c,c),
\end{equation}
where
\begin{equation}\label{cl_12}
a=\sqrt{w_1w_2},\quad b=\sqrt{w_3w_4},\quad c=\sqrt{w_5w_6},
\end{equation}
and 
\begin{equation}\label{cl_13}
C(n)=
\left(\frac{w_5}{w_6}\right)^{\frac{n}{2}}\,.
\end{equation}
Furthermore, 
\begin{equation}\label{cl_17}
Z_n(a,a,b,b,c,c)=c^{n^2}Z_n\left(\frac{a}{c},\frac{a}{c},\frac{b}{c},\frac{b}{c},1,1\right)
\end{equation}
and
\begin{equation}\label{cl_18}
\mu_n(\sg;a,a,b,b,c,c)=\mu_n\left(\sg;\frac{a}{c},\frac{a}{c},\frac{b}{c},\frac{b}{c},1,1\right),
\end{equation}
so that a general weight reduces to the two parameters, $\frac{a}{c},\frac{b}{c}\,$.

\subsection {Exact solution of the six-vertex model for a finite $n$}
Introduce the parameter
\begin{equation}\label{pf1}
\Delta=\frac{a^2+b^2-c^2}{2ab}\,.
\end{equation}
The phase diagram of the six-vertex model consists of the three 
phase regions: the ferroelectric phase region, $\Delta > 1$;
 the antiferroelectric phase region, $\Delta<-1$; and the disordered phase region, $-1<\Delta<1$. Observe that
$|a-b|>c$ in the ferroelectric phase region and $c>a+b$ in the antiferroelectric phase region, while in the disordered
phase region $a,b,c$ satisfy the triangle inequalities.
\begin{center}
 \begin{figure}[h]\label{PhaseDiagram1}
\begin{center}
   \scalebox{0.5}{\includegraphics{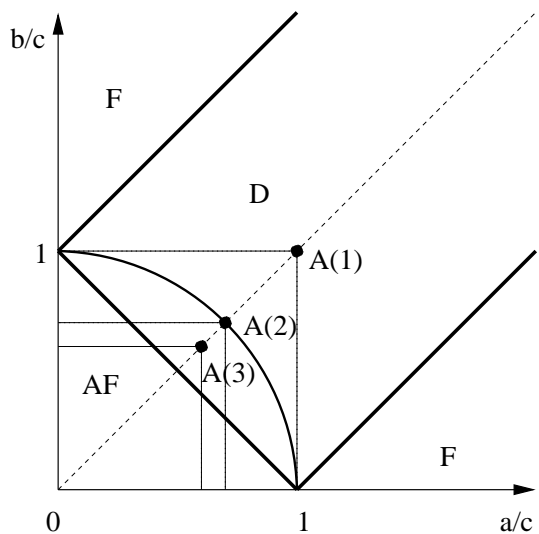}}
\end{center}
        \caption[The phase diagram of the model]{The phase diagram of the model, where {\bf F}, {\bf AF} and {\bf D} mark ferroelectric, antiferroelectric,  and disordered  phases, respectively. The circular arc corresponds to the so-called free fermion line, when $\Delta=0$, and the three
dots correspond to 1-, 2-, and 3-enumeration of alternating sign matrices.}
    \end{figure}
\end{center}
In the three phase regions we parameterize the weights in the standard way:
for the ferroelectric phase,
\begin{equation}\label{pf4}
a=\sinh(t-\ga), \quad
b=\sinh(t+\ga), \quad
c=\sinh(2|\ga|), \quad
0<|\ga|<t\,;
\end{equation}
for the antiferroelectric phase,
\begin{equation}\label{pf5}
a=\sinh(\ga-t), \quad
b=\sinh(\ga+t), \quad
c=\sinh(2\ga), \quad
|t|<\ga\,;
\end{equation}
and for the disordered phase
\begin{equation}\label{pf6}
a=\sin(\ga-t), \quad
b=\sin(\ga+t), \quad
c=\sin(2\ga), \quad
|t|<\ga.
\end{equation}
The phase diagram of the six-vertex model is shown on Fig.~3.
The phase diagram  and the Bethe Ansatz solution  of the six-vertex model for periodic and antiperiodic
boundary conditions are thoroughly
discussed in the works of Lieb \cite{Lieb1}-\cite{Lieb4}, Lieb, Wu \cite{LW},
Sutherland \cite{Sut}, Baxter \cite{Bax}, Batchelor, Baxter, O'Rourke, Yung \cite{BBOY}.
See also the work of Wu, Lin \cite{WL}, in which  the Pfaffian solution for the six-vertex
model with periodic boundary conditions is obtained on the free fermion line, $\Delta=0$.
Brascamp, Kunz and Wu \cite{BKW} prove the equality of the free energy with periodic and
free boundary conditions under various conditions on $a,b,c$, and also they prove the existence of the
spontaneous staggered polarization for sufficiently small values of the parameters $a/c$ and $b/c$. 

As concerns the six-vertex model with DWBC, it was noticed by Kuperberg \cite{Kup}, 
that on the diagonal, 
\begin{equation}\label{pf2}
\frac{a}{c}=\frac{b}{c}=x,
\end{equation}
 the six-vertex model with DWBC is equivalent to the $s$-enumeration of alternating
sign matrices (ASM), in which the weight of each such matrix is equal to $s^{N_-}$, 
where $N_-$ is the number of $(-1)$'s in the matrix and $s=\frac{1}{x^2}$. The exact
solution for a finite $n$ is known for 1-, 2-, and 3-enumerations of ASMs, see the 
works by Kuperberg \cite {Kup} and Colomo-Pronko \cite{CP1} for a solution
based on the Izergin-Korepin formula. A fascinating story of the discovery 
of the ASM formula is presented in the book  \cite{Bre} of Bressoud.
On the free fermion line, $\ga=\frac{\pi}{4}$, the partition function of
the six-vertex model with DWBC has a very simple form: $Z_n=1$. For a nice
short proof of this formula see the work \cite{CP1} of Colomo-Pronko. 

In this paper we will discuss  the {\it antiferroelectric phase region}, and we will use parameterization (\ref{pf5}).
The parameter $\Delta$ in the antiferroelectric phase region reduces to
\begin{equation}\label{pf6c}
\Delta=-\cosh(2\ga).
\end{equation}

 The six-vertex model with DWBC was introduced 
by Korepin in \cite{Kor}, who derived important 
recursion relations for the partition function of the model. These recursion relations 
have been solved by Izergin \cite{Ize}, and this lead to a beautiful determinantal formula 
for the partition function with DWBC.
A detailed proof of this formula, usually called the Izergin-Korepin formula, 
and its generalizations are given in
the paper of Izergin, Coker, and Korepin \cite{ICK}. When the weights are parameterized 
according to (\ref{pf5}), the Izergin-Korepin formula is
\begin{equation} \label{pf7}
Z_n=\frac{[\sinh(\ga-t)\sinh(\ga+t)]^{n^2}\tau_n}{\left(
\prod_{j=0}^{n-1}j!\right)^2}\,,
\end{equation}
where $\tau_n$ is the Hankel determinant,
\begin{equation} \label{pf8}
\tau_n=\det\left(\frac{d^{j+k-2}\phi}{dt^{j+k-2}}\right)_{1\le j,k\le n},
\end{equation} 
and
\begin{equation} \label{pf9}
\phi(t)=\frac{\sinh(2\ga)}{\sinh(\ga-t)\sinh(\ga+t)}\,.
\end{equation}
An elegant derivation of the 
Izergin-Korepin formula from the Yang-Baxter equations is 
given in the papers of Korepin, Zinn-Justin \cite{KZ} and Kuperberg \cite {Kup}
(see also the book of Bressoud \cite{Bre}).

One of the applications of the determinantal formula is that it
implies that 
the function $\tau_n$ solves the Toda equation
\begin{equation} \label{pf10}
\tau_n\tau''_n-{\tau'_n}^2=\tau_{n+1}\tau_{n-1},
\qquad n\ge 1,\qquad ({}')=\frac{\partial }{\partial t}\,,
\end{equation}
cf. the work of Sogo, \cite{Sog}. The Toda equation was used by Korepin and Zinn-Justin \cite{KZ} to 
derive the free energy of the six-vertex model with DWBC, assuming
some ansatz on the behavior of subdominant terms in the large $n$
asymptotics of the free energy.

Another application of the Izergin-Korepin formula is that
$\tau_n$ can be expressed in terms of  a partition function of
a random matrix model and also in terms of related orthogonal polynomials, see
 the paper  \cite{Z-J1} of Zinn-Justin.  In the antiferroelectric phase the
expression in terms of orthogonal polynomials can be obtained as follows. 
For the evaluation of the Hankel determinant, let us
write $\phi(t)$ in the form of the Laplace transform of a discrete measure,
\begin{equation} \label{dph6}
\phi(t)=\frac{\sinh(2\ga)}{\sinh(\ga-t)\sinh(\ga+t)}=
2\sum_{l=-\infty}^\infty e^{2tl-2\ga |l|}.
\end{equation} 
Then
\begin{equation} \label{dph7}
\tau_n=\frac{2^{n^2}}{n!}\sum_{l_1,\ldots,l_n=-\infty}^\infty \Delta(l)^2\prod_{i=1}^n
e^{2tl_i-2\ga |l_i|}, 
\end{equation}
where
\begin{equation} \label{dph8}
\Delta(l)=\prod_{1\le i<j\le n}(l_j-l_i)
\end{equation}
is the Vandermonde determinant, see Appendix \ref{proof_tau_n} in the end of the paper.

Introduce now discrete monic polynomials $P_j(x)=x^j+\dots$ orthogonal 
on the set
$\Z$ with respect to the weight,
\begin{equation} \label{dph9}
w(l)=e^{2tl-2\ga |l|},
\end{equation}
so that
\begin{equation} \label{dph14}
\sum_{l=-\infty}^\infty P_j(l)P_k(l)w(l)=h_k\delta_{jk}.
\end{equation}
Then it follows from (\ref{dph7}) that
\begin{equation} \label{dph15}
\tau_n=2^{n^2}\prod_{k=0}^{n-1}h_k, 
\end{equation}
see Appendix \ref{proof_tau_n_2} in the end of the paper. 

\subsection{Rescaling of the weight}

Set 
\begin{equation} \label{rw1}
\Delta_n=\frac{2\ga }{n}\,,\qquad x=l\Delta_n,\qquad w_n(x)=e^{-n(|x|-\z x)},\qquad \z=\frac{t}{\ga}\,,
\end{equation}
and
\begin{equation} \label{rw2}
 P_{nk}(x)= \Delta_n^k P_k\left(\frac{x}{\Delta_n}\right).
\end{equation}
Consider also the lattice
\begin{equation} \label{rw3}
 L_n=\left\{ x=\frac{2\ga k}{n}\,,\; k\in\Z\right\}.
\end{equation}
Then from (\ref{dph14}) we obtain that the monic polynomials $P_{nk}(x)$ satisfy the orthogonality condition,
\begin{equation} \label{rw4}
\sum_{x\in L_n}  P_{nj}(x) P_{nk}(x) w_n(x)\Delta_n=h_{nk}\delta_{jk},\qquad
h_{nk}=h_k \Delta_n^{2k+1}.
\end{equation}
We can then combine equations (\ref{pf7}), (\ref{dph15}), and (\ref{rw4}) to obtain
\begin{equation} \label{rw5}
Z_n=\left(\frac{nab}{\ga}\right)^{n^2}\prod_{k=0}^{n-1}\frac{h_{nk}}{(k!)^2}
\,,\qquad a=\sinh(\ga-t),\qquad b=\sinh(\ga+t).
\end{equation}
For what follows we will need to extend the weight $w_n(x)$ to the complex plane.
We do so by defining $w_n(z)$ on the complex plane as
\begin{equation} \label{rw6}
w_n(z)=
e^{-nV(z)}
\end{equation}
where 
\begin{equation} \label{rw7}
V(z)=
\left\{
\begin{aligned}
&z-\z z\quad &\textrm{when}\quad \Re z\ge 0,\\
&-z-\z z \quad &\textrm{when}\quad \Re z\le 0,
\end{aligned}
\right.
\end{equation}
so that $V(z)$, and thus $w_n(z)$, is two-valued on the imaginary axis.

\subsection{Main result: Asymptotics of the partition function}

This work is a continuation of the work \cite{BF} of the first author with Vladimir Fokin
and \cite{BL1}, \cite{BL2} of the authors of the present work.
In \cite{BF} the authors obtain the large $n$ asymptotics
of the partition function $Z_n$ in the disordered phase. They prove the conjecture 
of Paul Zinn-Justin \cite{Z-J1} that the large 
$n$ asymptotics of $Z_n$ in the disordered phase has the following form: For
some $\ep>0$,
\begin{equation}\label{main1}
Z_n=Cn^{\kappa}F^{n^2}[1+O(n^{-\ep})].
\end{equation}
Furthermore, they find the exact value of the exponent $\kappa$,
\begin{equation}\label{main2}
\kappa=\frac{1}{12}-\frac{2\ga^2}{3\pi(\pi-2\ga)}\,.
\end{equation}
The value of $F$ in the disordered phase is given by the formula,
\begin{equation}\label{main3}
F=\frac{\pi ab}{2\ga\cos\frac{\pi t}{2\ga}}
\,,\qquad a=\sin(\ga-t),\qquad b=\sin(\ga+t)\,,
\end{equation}
 in parameterization (\ref{pf6}).

In the work \cite{BL1} we obtain the following large $n$ asymptotic formula for $Z_n$
in the ferroelectric phase region: For any $\ep>0$,
\begin{equation} \label{main4}
Z_n=C G^n F^{n^2}\left[1+O\left(e^{- n^{1-\ep}}\right)\right],
\end{equation}
where $C=1-e^{-4\ga}$, $G=e^{\ga-t}$, and $F=\sinh(t+\ga)$ in parameterization (\ref{pf4}).

Finally, in the work \cite{BL2} we obtain the following large $n$ asymptotic formula for $Z_n$
on the border line between the ferroelectric and disordered phase regions: 
\begin{equation} \label{main5} 
Z_n\left(a,a,a+1,a+1,1,1\right)=C n^\kappa
G^{\sqrt n}F^{n^2}[1+O(n^{-1/2})]\,,
\end{equation}
where $C>0$,
\begin{equation} \label{main6}
\kappa=\frac{1}{4}\,,\qquad G=\exp\left[-\z\left(\frac{3}{2}\right)\sqrt{\frac{a}{ \pi }}\right]\,, \qquad F=a+1,
\end{equation}
and
$\z$ is the Riemann zeta-function.

In the present paper we obtain the large $n$ asymptotic formula for $Z_n$ in the antiferroelectric phase
region. The formulation of the main result of the present paper
and the proofs involve the Jacobi theta functions.
Let us review their definition and basic properties.

There are four Jacobi theta functions: 
\begin{equation}\label{main8}
\begin{aligned}
\th_1(z)&=2\sum_{n=0}^{\infty}(-1)^n q^{(n+\frac{1}{2})^2} \sin\big((2n+1)z\big) \\
\th_2(z)&=2\sum_{n=0}^{\infty}q^{(n+\frac{1}{2})^2}\cos\big((2n+1)z\big) \\
\th_3(z)&=1+2\sum_{n=1}^{\infty}q^{n^2}\cos(2nz) \\
\th_4(z)&=1+2\sum_{n=1}^{\infty}(-1)^n q^{n^2} \cos(2nz) \\
\end{aligned}
\end{equation}
where $q$ is the {\it elliptic nome}. We will assume that $1>q>0$. Fig.~4 shows the graphs of $\th_1,\,\th_2$ (left
figure) and $\th_3,\,\th_4$ (right figure) on the interval $[0,\pi]$ for $q=0.5$. 
Observe that $\th_1,\,\th_4$ are increasing on $[0,\frac{\pi}{2}]$ while $\th_2,\,\th_3$
are decreasing on this interval.

\begin{center}
 \begin{figure}[h]\label{theta}
\begin{center}
   \scalebox{0.3}{\rotatebox{0}{\includegraphics{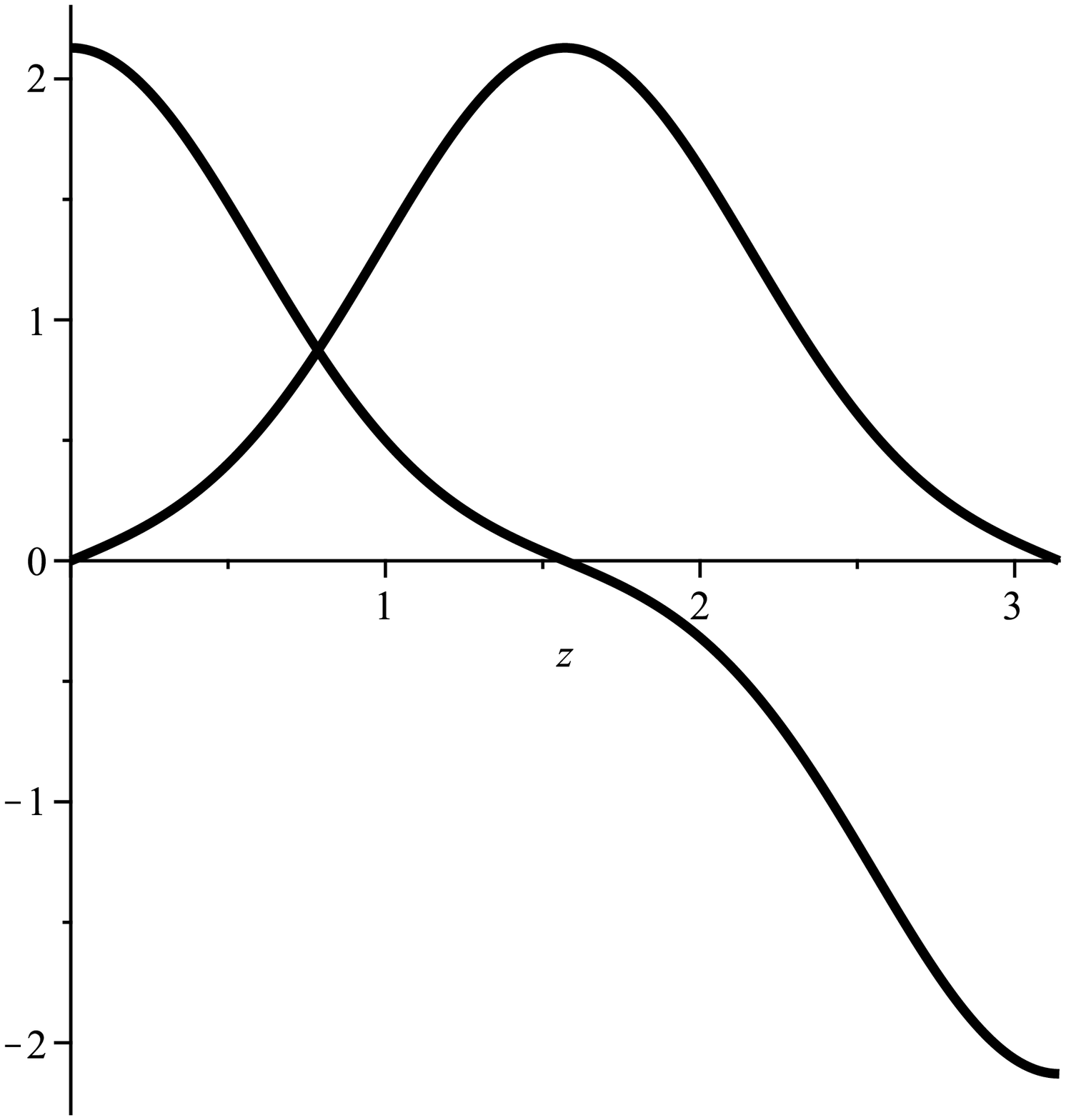}}}\hskip 1cm \scalebox{0.3}{\rotatebox{0}{\includegraphics{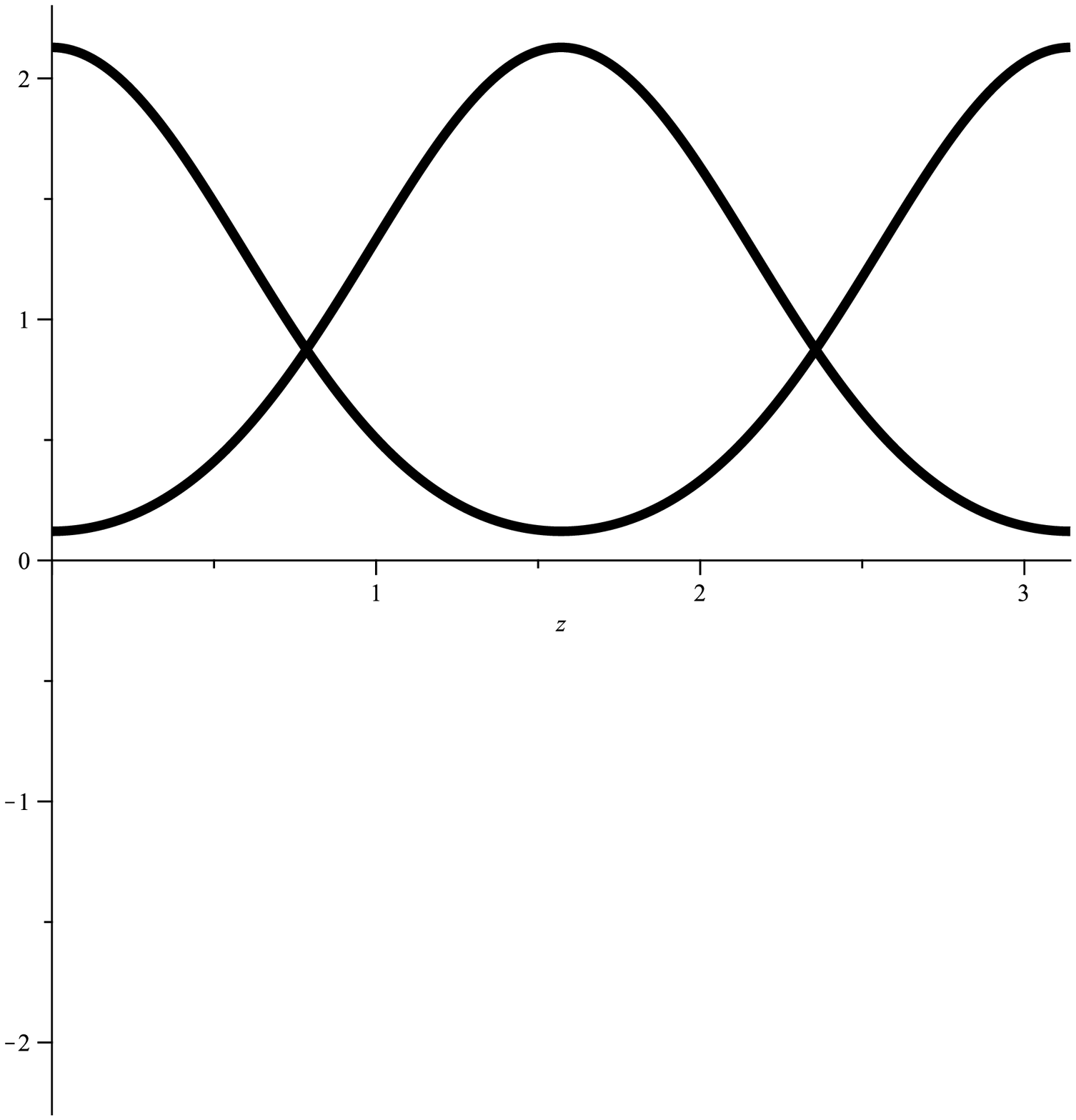}}}
\end{center}
        \caption[The graphs of thetas.]{The graphs of $\th_1,\,\th_2$ (left
figure) and $\th_3,\,\th_4$ (right figure) on the interval $[0,\pi]$ for $q=0.5$.}
    \end{figure}
\end{center}

The Jacobi theta functions
satisfy the following periodicity conditions:
\begin{equation}\label{main9}
\begin{aligned}
&\th_1(z+\pi)=-\th_1(z),\qquad \th_1(z+\pi \tau)=-e^{-2iz}q^{-1}\th_1(z),    \\
&\th_2(z+\pi)=-\th_2(z),\qquad \th_2(z+\pi \tau)=e^{-2iz}q^{-1}\th_2(z),\\
&\th_3(z+\pi)=\th_3(z),\qquad \th_3(z+\pi \tau)=e^{-2iz}q^{-1}\th_3(z),    \\
&\th_4(z+\pi)=\th_4(z),\qquad \th_4(z+\pi \tau)=-e^{-2iz}q^{-1}\th_4(z),
\end{aligned}
\end{equation}
where $\tau$ is a pure imaginary number related to $q$ by the equation,
\begin{equation}\label{main10}
q=e^{i\pi\tau}.
\end{equation}
The theta functions also satisfy the symmetry conditions:
\begin{equation}\label{main11}
\th_1(-z)=-\th_1(z),\quad \th_2(-z)=\th_2(z),\quad \th_3(-z)=\th_3(z),\quad \th_4(-z)=\th_4(z),
\end{equation}
and the equations,
\begin{equation}\label{main12}
\th_1(z)=\th_2\left(z-\frac{\pi}{2}\right),\quad \th_3(z)=\th_4\left(z+\frac{\pi}{2}\right),\quad
 \th_1(z)=-ie^{iz+\frac{i\pi\tau }{4}}\th_4\left(z+\frac{\pi\tau}{2}\right).
\end{equation}
The only zeroes of the theta functions are 
\begin{equation}\label{main13}
\th_1(0)=0,\quad \th_2\left(\frac{\pi}{2}\right)=0,\quad
 \th_3\left(\frac{\pi}{2}+\frac{\pi\tau}{2}\right)=0,\quad 
\th_4\left(\frac{\pi\tau}{2}\right)=0,
\end{equation}
and their shifts by $m\pi+n\pi\tau$; $m,n\in\Z$. There are many non-trivial identities satisfied by the theta functions.  A list of those identities used in this paper is given in Appendix \ref{identities}. 

In the antiferroelectric phase region we use parameterization (\ref{pf5}) with two
parameters $t,\ga$ such that $|t|<\ga$. In what follows we will also use the following
parameters:
\begin{equation}\label{main14}
\z=\frac{t}{\ga}\in(-1,1)\,,\qquad \om=\frac{\pi(1+\z)}{2}\in(0,\pi)\,.
\end{equation}
The elliptic nome for all Jacobi theta functions in this paper will be equal to
\begin{equation}\label{main16}
q=e^{-\frac{\pi^2}{2\ga}}.
\end{equation}

Our main result in the present paper is the following  asymptotic formula for $Z_n$:

\begin{theo} \label{thmain1} As $n\to\infty$, 
\begin{equation}\label{main17}
Z_n=C\th_4\left(n\om\right) F^{n^2}(1+O(n^{-1})),
\end{equation}
where $C>0$ is a constant, and
\begin{equation}\label{main18}
F=\frac{\pi \sinh(\ga-t)\sinh(\ga+t)\th'_1(0)}{2\ga\th_1(\om)}\,.
\end{equation}
\end{theo}

The asymptotic formula (\ref{main17}) proves the conjecture of Zinn-Justin
in \cite{Z-J1}.  The proof of Theorem \ref{thmain1} will be based on the Riemann-Hilbert approach to discrete
orthogonal polynomials. An important first step in this approach is a
construction of the equilibrium measure.

\section{Equilibrium measure}\label{equilibrium}

\subsection{Heuristic motivation and definition of the equilibrium measure}
If we scale the variables in (\ref{dph7}) as $\mu_i=\frac{2\ga l_i}{n}$, then we can rewrite 
formula (\ref{dph7}) as
\begin{equation}\label{eq1}
\tau_n=\frac{2^{n^2}}{n!}\sum_{\mu \in \frac{2\ga}{n} \Z^n}e^{-n^2 H_n(\nu_\mu)},
\end{equation}
where
\begin{equation}\label{eq2}
d\nu_\mu(x)=\frac{1}{n}\sum_{j=1}^{n}\delta(x-\mu_j),
\end{equation}
and
\begin{equation}\label{eq3}
H(\nu)=\iint_{x\not= y} \log \frac{1}{|x-y|}d\nu (x)d\nu (y)+\int (|x| -\z x)d\nu(x).
\end{equation}
where all integrals are over $\mathbb{R}$. 

Due to the factor $(-n^2)$ in the exponent of (\ref{eq1}), we expect the sum, in the large $n$ limit, 
to be focused in a neighborhood of a global minimum of the functional $H$.  
Clearly, we have that $\nu_\mu$ is a probability measure and 
\begin{equation}\label{eq3a}
\nu_\mu (a,b) \leq \frac{b-a}{2\ga} \quad \textrm{for any}\quad -\infty <a<b<\infty,  
\end{equation}
because in (\ref{eq2}), $\mu_j\in \frac{2\ga}{n} \Z$ and $\mu_j\not=\mu_k$ if $j\not= k$.
With these constraints in mind, we define
\begin{equation}\label{eq4}
E_0=\inf_\nu H(\nu)
\end{equation}
where the infimum is taken over all probability measures satisfying (\ref{eq3a}).  
It is possible to prove that there exists a unique minimizer $\nu_0$, so that  
\begin{equation}\label{eq4z}
E_0=H(\nu_0),
\end{equation}
see, e.g., the works of Saff and Totik \cite{ST}, Dragnev and Saff \cite{DS} and Kuijlaars \cite{Kui}.
Furthermore, $\nu_0$ has support on a finite number of intervals, and is absolutely continuous with respect to the Lebesgue measure.  The minimizer $\nu_0$ is called the {\it equilibrium measure}.

Denote the density function of the equilibrium measure as $\rho(x)$, and its resolvent as $\omega$, so we have 
\begin{equation}\label{eq6}
\frac{d\nu_0}{dx}=\rho(x), \quad \omega(z)=\int\frac{\rho(x)dx}{z-x}\,,
\end{equation}
and
\begin{equation}\label{eq6a}
\rho(x)=\frac{1}{2\pi i}\big(\omega(x-i 0)-\omega(x+i 0)\big). 
\end{equation}
The structure of the equilibrium measure $\nu_0$ is studied in the paper of Zinn-Justin  \cite{Z-J1}, who shows
 that $\nu_0$ has support on an interval $[\al, \be]$, with a saturated region $[\al', \be']$ 
in which 
\begin{equation}\label{eq4a}
\rho(x)=\frac{1}{2\ga}\,,\qquad x\in [\al', \be'],
\end{equation}
and two unsaturated regions, $[\al, \al']$ and $[\be', \be]$, 
in which     
\begin{equation}\label{eq4b}
0<\rho(x)<\frac{1}{2\ga}\,,\qquad x\in (\al, \al')\cup (\be', \be),
\end{equation}
see Fig.~5. We also have that
\begin{equation}\label{eq4c}
\al<\al'<0<\be'<\be,
\end{equation}
so that the origin, which is a singular point of the potential $V(x)=|x| -\z x$, lies inside the
saturated region $[\al', \be']$ .

\begin{center}
 \begin{figure}[h]\label{rho}
\begin{center}
   \scalebox{0.6}{\includegraphics{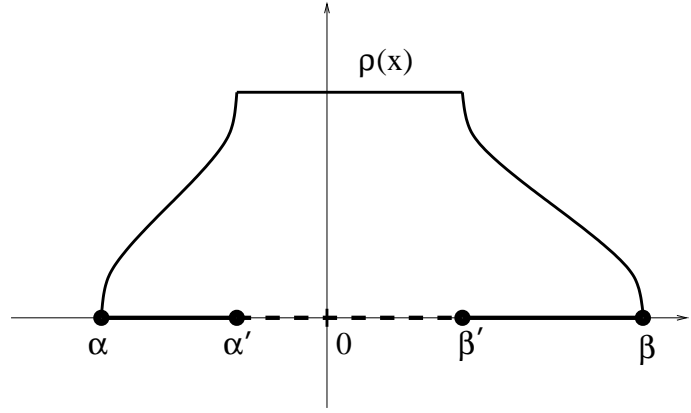}}
\end{center}
        \caption[The density.]{The equilibrium density function $\rho(x)$.}
    \end{figure}
\end{center}

The measure $\nu_0$ is uniquely determined by the Euler-Lagrange variational conditions
\begin{equation}\label{eq5}
2\int \log|x-y| d\nu_0 (y) - (|x| -\z x) \left\{
\begin{aligned}
&=l \quad \textrm{for}\quad x \in [\al,\al'] \cup [\be',\be], \\
&\geq l \quad \textrm{for}\quad x \in [\al',\be'],  \\
&\leq l \quad \textrm{for}\quad x \notin [\al,\be],
\end{aligned}\right.
\end{equation}
where $l$ is the Lagrange multiplier.
The Euler-Lagrange variational conditions imply 
\begin{equation}\label{eq8}
\omega(x-i 0)+\omega(x+i 0)=-\z+\textrm{sgn}(x) \quad \textrm{for} \quad x\in [\al, \al']\cup[\be',\be],
\end{equation}
whereas in the saturated region, we have 
\begin{equation}\label{eq7}
\rho(x)=\frac{1}{2\pi i}\big(\omega(x-i 0)-\omega(x+i 0)\big)=\frac{1}{2\ga} \quad \textrm{for} \quad x\in [\al', \be'].
\end{equation}

Now we will give a detailed description of the equilibrium measure. We begin with explicit formulae
for the end-points of the support of the equilibrium measure.

\subsection{Explicit formulae for the end-points} They are given in the following proposition.

\begin{prop} \label{proposition_end_points}
The end-points of the support of the equilibrium measure $\nu_0$ are equal to
\begin{equation}\label{eq8a}
\begin{aligned}
\al&=-\pi\,\frac{\th_1'(\frac{\om}{2})}{\th_1(\frac{\om}{2})}\,,\qquad 
\al'=-\pi\,\frac{\th_4'(\frac{\om}{2})}{\th_4(\frac{\om}{2})}\,,\\ 
\be'&=-\pi\,\frac{\th_3'(\frac{\om}{2})}{\th_3(\frac{\om}{2})}\,,\qquad 
\be=-\pi\,\frac{\th_2'(\frac{\om}{2})}{\th_2(\frac{\om}{2})}\,.
\end{aligned}
\end{equation}
The differences between the end-points are equal to
\begin{equation}\label{eq8b}
\begin{aligned}
\al'-\al&=\pi\th_4^2(0)\frac{\th_2(\frac{\om}{2})\th_3(\frac{\om}{2})}{\th_1(\frac{\om}{2})\th_4(\frac{\om}{2})}\,  ,
\qquad
\be'-\al'=\pi\th_2^2(0)\frac{\th_1(\frac{\om}{2})\th_2(\frac{\om}{2})}{\th_3(\frac{\om}{2})\th_4(\frac{\om}{2})}\,  ,
\\
\be-\be'&=\pi\th_4(0)^2\frac{\th_1(\frac{\om}{2})\th_4(\frac{\om}{2})}{\th_2(\frac{\om}{2})\th_3(\frac{\om}{2})}\,.
\end{aligned}
\end{equation}
and
\begin{equation}\label{eq8c}
\begin{aligned}
\be-\al &=\pi\th_2^2(0)\frac{\th_3(\frac{\om}{2})\th_4(\frac{\om}{2})}{\th_1(\frac{\om}{2})\th_2(\frac{\om}{2})}\,  , 
\qquad
\be-\al'=\pi\th_3(0)^2\frac{\th_1(\frac{\om}{2})\th_3(\frac{\om}{2})}{\th_2(\frac{\om}{2})\th_4(\frac{\om}{2})}\,,
\\
\be'-\al &=\pi\th_3^2(0)\frac{\th_2(\frac{\om}{2})\th_4(\frac{\om}{2})}{\th_1(\frac{\om}{2})\th_3(\frac{\om}{2})}.
\end{aligned}
\end{equation}
Finally, we have the Zinn-Justin formula for the centroid of the end-points,
\begin{equation}\label{eq8d}
\frac{\al+\al'+\be'+\be}{4}=-\frac{\pi}{2}\frac{\th_2'(\frac{\pi\z}{2})}{\th_2(\frac{\pi\z}{2})}\,.
\end{equation}
\end{prop}

For a proof of Proposition \ref{proposition_end_points} see the next section.

\subsection {Equilibrium density function} The equilibrium density function is described in
the following proposition.

\begin{prop} \label{proposition_equilibrium_density}
The equilibrium density function $\rho(x)$ is given by 
the formulae,
\begin{equation}\label{edf3}
\rho(x)=
\left\{
\begin{aligned}
& \frac{1}{\pi }\int_\al^x\frac{dx'}{\sqrt{(x'-\al)(\al'-x')(\be'-x')(\be-x')}},\quad 
\al\le x\le \al',\\
&\frac{1}{2\ga}\,,\quad \al'\le x\le \be',\\
& \frac{1}{\pi }\int_x^\be\frac{dx'}{\sqrt{(x'-\al)(x'-\al')(x'-\be')(\be-x')}},\quad 
\be'\le x\le \be.
\end{aligned}
\right.
\end{equation}
Also,
\begin{equation}\label{edf3a}
\int_0^\be \rho(x)\,dx=\frac{1+\z}{2}\,.
\end{equation}
The resolvent $\om(z)$ of the equilibrium measure is given as 
\begin{equation}\label{edf1}
\om(z)=\int_z^\infty\frac{dz'}{\sqrt{(z'-\al)(z'-\al')(z'-\be')(z'-\be)}},
\end{equation}
where integration takes place on the sheet of $\sqrt{R(z')} \equiv \sqrt{(z'-\al)(z'-\al')(z'-\be')(z'-\be)}$ for which $\sqrt{R(z')} >0$ for $z'>\be$, with cuts on $[\al,\al']$ and $[\be',\be]$.
\end{prop}

For a proof of this proposition see the next section.

\subsection {$g$-function}
Define the $g$-function on $\C \setminus [-\infty, \be]$ as 
\begin{equation}\label{g0}
g(z)=\int_\al^\be \log(z-x)d\nu_0(x)
\end{equation}
where we take the principal branch for logarithm.

{\bfseries Properties of \( g(z)\):}
\begin{enumerate}
\item \( g(z)\) is analytic in \(\C\setminus
    (-\infty,\be]\).
\item For large $z$,
\begin{equation}\label{g1}
 g(z)=\log z-\sum_{j=1}^\infty \frac{g_j}{z^j}\,,\qquad 
g_j=\int_{\al}^{\be} \frac{x^j}{j}\,d\nu_0(x).
\end{equation}
\item 
\begin{equation}\label{g2}
 g'(z)=\om(z).
\end{equation}
\item From the first relation in (\ref{eq5}) we have that 
\begin{equation}\label{g3a}
g_+(x)+g_-(x)=|x| -\z x+l \quad \textrm{for}\quad x \in [\al,\al'] \cup [\be',\be],
\end{equation}
where $g_+$ and $g_-$ refer to the limiting values of $g$ from the upper and lower half-planes, respectively.
By differentiating this equation we obtain that
\begin{equation}\label{g4}
\om_+(x)+\om_-(x)=g_+'(x)+g_-'(x)=\sign x-\z \quad \textrm{for}\quad x \in [\al,\al'] \cup [\be',\be],
\end{equation}
Consider the function 
\begin{equation}\label{g4a}
f(x)=g_+(x)+g_-(x)-(|x| -\z x+l).
\end{equation}
We have from (\ref{g3a}), (\ref{g4}) that
\begin{equation}\label{g4b}
f(x)=f'(x)=0\quad \textrm{for}\quad x=\al,\al',\be',\be,
\end{equation}
and from (\ref{edf1}) that
\begin{equation}\label{g4c}
f''(x)=-\frac{1}{\sqrt{(x-\al)(x-\al')(x-\be')(x-\be)}}
\quad \textrm{for}\quad x \in (-\infty,\al)\cup(\al',\be') \cup (\be,\infty).
\end{equation}
Since 
\begin{equation}\label{g4d}
f''(x)<0
\quad \textrm{for}\quad x \in (-\infty,\al) \cup (\be,\infty).
\end{equation}
and 
\begin{equation}\label{g4e}
f''(x)>0
\quad \textrm{for}\quad x \in (\al',\be'),\;x\not=0,
\end{equation}
we obtain that 
\begin{equation}\label{g3}
g_+(x)+g_-(x)\,
\left\{
\begin{aligned}
&=|x| -\z x+l \quad \textrm{for}\quad x \in [\al,\al'] \cup [\be',\be], \\
&> |x| -\z x+l \quad \textrm{for}\quad x \in (\al',\be'), \\
&< |x| -\z x+l \quad \textrm{for}\quad x \in \R\setminus[\al,\be].
\end{aligned}
\right.
\end{equation}

\item Equation (\ref{g0}) implies that the function
\begin{equation}\label{g5}
 G(x)\equiv g_+(x)-g_-(x)
\end{equation}
is pure imaginary for all real 
\( x\), and 
\begin{equation}\label{g6}
 G(x)=
\left\{
\begin{aligned}
& 2\pi i \quad\textrm{for} \quad -\infty<x\le \al,\\
& 2\pi i-2\pi i\int_\al^x\rho(s)\,ds \quad\textrm{for} \quad \al\le x\le \al',\\
& 2\pi i \left(\frac{1+\z}{2}-\frac{x}{2\ga}\right) \quad\textrm{for} \quad \al'\le x\le \be',\\
& 2\pi i\int_x^\be\rho(s)\,ds \quad\textrm{for} \quad \be'\le x\le \be,\\
& 0 \quad\textrm{for} \quad \be\le x<\infty.
\end{aligned}
\right.
\end{equation} 
From (\ref{g3}) and (\ref{g6}) we obtain that
\begin{equation}\label{g6a}
 2g_{\pm}(x)=
\left\{
\begin{aligned}
& |x|-\z x+l\pm\left[2\pi i-2\pi i\int_\al^x\rho(s)\,ds\right] \quad\textrm{for} \quad \al\le x\le \al',\\
& |x|-\z x+l\pm 2\pi i\int_x^\be\rho(s)\,ds \quad\textrm{for} \quad \be'\le x\le \be.
\end{aligned}
\right.
\end{equation}
\item 
Also, from (\ref{g6}) 
\begin{equation}\label{g7}
 \left. \frac{dG(x+iy)}{dy}\right|_{y=0}=2\pi \rho(x)>0,
\quad x\in (\al,\be).
\end{equation}
\end{enumerate}
Observe that from (\ref{g3a}) we have that
\begin{equation}\label{g10}
 G(x)=2g_+(x)-V(x)-l=-[2g_-(x)-V(x)-l],\quad x \in [\al,\al'] \cup [\be',\be],
\end{equation}
where $V(x)\equiv |x|-\z x$.

\subsection{Evaluation of the Lagrange multiplier $l$} We have the following proposition.

\begin{prop} \label{proposition_lagrange_multiplier} 
The Lagrange multiplier $l$ solves the equation, 
\begin{equation}\label{lm18}
e^{\frac{l}{2}}=\frac{\pi \th'_1(0)}{2e\th_1(\om)}.
\end{equation}
\end{prop}

For a proof of this proposition see the next section.

\section{Proof of Propositions \ref{proposition_end_points},
 \ref{proposition_equilibrium_density}, and \ref{proposition_lagrange_multiplier}}

{\bf Proof of Proposition \ref{proposition_end_points}.} 
Following Zinn-Justin \cite{Z-J1}, we make the following elliptic change of variables:
\begin{equation}\label{eq9}
u(z)=\frac{1}{2}\sqrt{(\be'-\al)(\be-\al')}
\int_{\be}^{z}\frac{dz'}{\sqrt{(z'-\al)(z'-\al')(z'-\be')(z'-\be)}},
\end{equation}
where integration takes place on the sheet on $\sqrt{R(z')}$ specified in Proposition \ref{proposition_equilibrium_density}.
To understand this integral in terms of the usual elliptic integrals, we first make the change of variables
\begin{equation}\label{eq9a}
v(z')=\frac{(\be-z')(\be'-\al)}{(\be'-z')(\be-\al)}\,,
\end{equation}
so that
\begin{equation}\label{eq9b}
 z'=\frac{\be'(\be-\al)v-\be(\be'-\al)}{(\be-\al)v-(\be'-\al)}\,. 
\end{equation}
Note that $v(\be)=0$, $v(\be')=\infty$, and $v(\al)=1$.  When we substitute $v$ into equation (\ref{eq9}), we have 
\begin{equation}\label{eq10}
u(z)=\frac{1}{2k}\int_{0}^{v(z)}\frac{dv}{\sqrt{v(v-1)(v-\frac{1}{k^2})}}\,,
\end{equation}
where
\begin{equation}\label{eq10a}
k=\sqrt{\frac{(\be-\al)(\be'-\al')}{(\be'-\al)(\be-\al')}}\,.
\end{equation}
We next take $v'=\sqrt{v}$, obtaining
\begin{equation}\label{eq11}
u(z)=\int_{0}^{\sqrt{v(z)}}\frac{dv'}{\sqrt{(1-v'^2)(1-k^2 v'^2)}},
\end{equation}
which corresponds to $\sqrt{v(z)}={\rm sn}(u,k)$, so that
\begin{equation}\label{eq11a}
\frac{(\be-z)(\be'-\al)}{(\be'-z)(\be-\al)}=\sn^2(u),\qquad {\rm sn}(u)= \sn(u,k).
\end{equation}
Notice that $u$ maps the upper $z$-plane conformally and bijectively onto the rectangle $[0,K] \times [0,iK']$, and the lower $z$-plane conformally and bijectively onto the rectangle $[0,K] \times [-iK',0]$, where
\begin{equation} \label{eq12}
\begin{aligned}
&K=u(\al)=\int_{0}^{1}\frac{dv'}{\sqrt{(1-v'^2)(1-k^2 v'^2)}} \qquad\textrm{and}\\
&K'=-iu(\be')=\int_{1}^{\frac{1}{k}}\frac{dv'}{\sqrt{(v'^2-1)(1-k^2 v'^2)}} \\
\end{aligned}
\end{equation}
are the usual complete integrals of the first kind.  More specifically (see Fig.~6), 
\begin{enumerate}
  \item The upper (resp. lower) cusp of the interval $[\be', \be]$ is mapped onto the interval $[0, iK']$ (resp. $[0, -iK']$).
  \item  The upper (resp. lower) cusp of the interval $[\al, \al']$ is mapped onto the interval $[K, K+iK']$ (resp. $[K, K-iK']$).
  \item  The interval $[\al', \be']$ is mapped onto the inteval $[iK', K+iK']$ or the inteval $[-iK', K-iK']$, depending on the path of integration.
  \item The remaining part of the real axis, $[-\infty,\al]\cup[\be,\infty]$, is mapped onto the interval $[0,K]$,
with $u(\infty)=u_*=u_\infty$.
\end{enumerate} 
\begin{center}
 \begin{figure}[h]\label{u-plane}
\begin{center}
   \scalebox{0.5}{\includegraphics{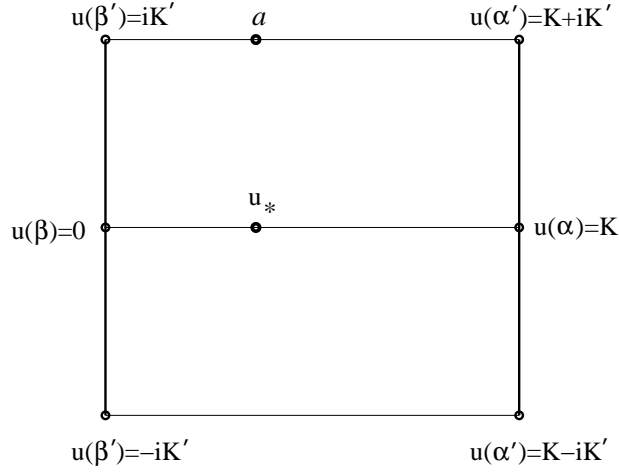}}
\end{center}
        \caption[The $u$-plane.]{The $u$-plane. Here $u_*=u_\infty\equiv u(\infty)$ and $a=u_*+iK'$.}
    \end{figure}
    \end{center}
We will denote the rectangle $[0,K] \times [-iK',iK']$ as $R$, the 
{\it fundamental domain} of the function $z(u)$.  
We can now define
\begin{equation}\label{eq13}
\tilde\omega(u)=\omega(z(u)) \quad \textrm{for} \quad u\in R.
\end{equation} 
The Euler-Lagrange equation (\ref{eq8}) and the equation (\ref{eq7}) then become
\begin{equation}\label{eq14}
\begin{aligned}
\tilde\omega(u)+\tilde\omega(-u)&=1-\z &\ \ \textrm{for}\ u\in [-iK',iK']  \\
\tilde\omega(u)+\tilde\omega(-u+2K)&=-1-\z &\ \ \textrm{for}\ u\in [K-iK',K+iK'] \\
\tilde\omega(u+2iK')-\tilde\omega(u)&=-\frac{i\pi}{\ga} &\ \ \textrm{for}\ u\in [-iK',K-iK']. \\
\end{aligned}
\end{equation}
The function $\omega(z)$ is analytic in $\C \setminus [\al,\be]$, but can be analytically continued from either above or below through any of the cuts $[\al,\al']$, $[\al',\be']$, and $[\be',\be]$. These analytic continuations in the $z$-plane give an analytic continuation of $\tilde\omega$ in the $u$-plane into a neighborhood of $R$, which can then be continued by equations (\ref{eq14}) to the entire $u$-plane.  We therefore have that $\tilde\omega$ is analytic and satisfies equations (\ref{eq14}) throughout the $u$-plane.  The first two equations of (\ref{eq14}) can be combined as
\begin{equation}\label{eq15}
\tilde\omega(u+2K)=\tilde\omega(u)-2.
\end{equation}
It therefore follows that $\tilde\omega$ is a linear function of $u$, as its derivative is a doubly periodic entire function.  We also know from the fact that $\omega(z)\sim \frac{1}{z}$ at infinity that
\begin{equation}\label{eq16}
\tilde\omega(u)=-\frac{2}{\sqrt{(\be'-\al)(\be-\al')}}(u-u_\infty)+O(u-u_\infty)^2
\end{equation} 
in some neighborhood of $u_\infty$, where $u_\infty$ is the image of infinity under the map $u(z)$.  It thus follows from (\ref{eq14}), (\ref{eq15}), and (\ref{eq16}) that
\begin{equation}\label{eq17}
\tilde\omega(u)=-\frac{1}{K}(u-u_\infty)
\end{equation}
and that
\begin{align}
\frac{K'}{K}&=\frac{\pi}{2\ga} \label{eq18a}\\
\sqrt{(\be'-\al)(\be-\al')}&=2K  \label{eq18b}\\
\frac{u_\infty}{K}&=\frac{1-\z}{2}.\label{eq18c}
\end{align}
From (\ref{eq11a}) we obtain that
\begin{equation}\label{eq18e}
\frac{\be'-\al}{\be-\al}=\sn^2(u_\infty).
\end{equation}
This implies that
\begin{equation}\label{eq18f}
\begin{aligned}
\cn^2(u_\infty)&=1-\sn^2(u_\infty)=1-\frac{\be'-\al}{\be-\al}
=\frac{\be-\be'}{\be-\al}\,,\\
\dn^2(u_\infty)&=1-k^2\sn^2(u_\infty)=1-\frac{(\be-\al)(\be'-\al')}{(\be'-\al)(\be-\al')}
\frac{(\be'-\al)}{(\be-\al)}
=\frac{\be-\be'}{\be-\al'}.
\end{aligned}
\end{equation}

From equations (\ref{eq18b}), (\ref{eq18e}), (\ref{eq18f}) we obtain 
the distances between the turning points in terms of $u_\infty$:
\begin{equation}\label{eq19}
\begin{aligned}
\be-\al&=2K\frac{\dn(u_\infty)}{\sn(u_\infty)\cn(u_\infty)}  \\
\be-\al'&=2K\frac{\cn(u_\infty)}{\sn(u_\infty)\dn(u_\infty)} \\
\be-\be'&=2K\frac{\cn(u_\infty)\dn(u_\infty)}{\sn(u_\infty)}.
\end{aligned}
\end{equation}
The functions sn, cn, and dn are expressed in terms of Jacobi theta functions as follows (see, eg. \cite{WW}),
\begin{equation}\label{eq19a}
\sn(u)=\frac{\th_3(0)}{\th_2(0)}\frac{\th_1(\frac{\pi u}{2K})}{\th_4(\frac{\pi u}{2K})}\,, \quad \cn(u)=\frac{\th_4(0)}{\th_2(0)}\frac{\th_2(\frac{\pi u}{2K})}{\th_4(\frac{\pi u}{2K})}\,, \quad \dn(u)=\frac{\th_4(0)}{\th_3(0)}\frac{\th_3(\frac{\pi u}{2K})}{\th_4(\frac{\pi u}{2K})}\,. 
\end{equation}
By (\ref{eq18a}), the half-period ratio $\tau$ and the elliptic nome $q$ of the theta functions are
\begin{equation}\label{eq19b}
\tau=\frac{iK'}{K}=\frac{i\pi}{2\ga} \quad \textrm{and} \quad q=e^{\frac{-\pi K'}{K}}=e^{\frac{-\pi^2}{2\ga}}.
\end{equation}
If we take into account the fact that
\begin{equation}\label{eq19c}
\th_3(0)^2=\frac{2K}{\pi}
\end{equation}
and equation (\ref{eq18c}), we can write equations for the distances between the turning points that involve the original parameters only:
\begin{equation}\label{eq19d}
\begin{aligned}
\be-\al &=\pi\th_2^2(0)\frac{\th_3(\frac{\om}{2})\th_4(\frac{\om}{2})}{\th_1(\frac{\om}{2})\th_2(\frac{\om}{2})}\,  , 
\qquad
\be-\al'=\pi\th_3(0)^2\frac{\th_1(\frac{\om}{2})\th_3(\frac{\om}{2})}{\th_2(\frac{\om}{2})\th_4(\frac{\om}{2})}\,,
\\
\be'-\al &=\pi\th_3^2(0)\frac{\th_2(\frac{\om}{2})\th_4(\frac{\om}{2})}{\th_1(\frac{\om}{2})\th_3(\frac{\om}{2})},
\end{aligned}
\end{equation}
giving (\ref{eq8c}).  These equations determine the end-points $\al,\al',\be',\be$ up to a shift.
To fix the shift we use the equation (\ref{eq5}) at the points $\al'$ and $\be'$ to obtain
\begin{equation}\label{eq20}
\int_{\al '}^{\be '}\big(\om(z+i0)+\om(z-i0)\big)dz=(1-\z)\be'+(1+\z)\al'.
\end{equation}
Writing this integral in terms of $u$ gives
\begin{equation}\label{eq21}
\int_{iK'}^{K+iK'}{\frac{1}{K}(u-u_\infty)r'(u)du}+ \int_{-iK'}^{K-iK'}{\frac{1}{K}(u-u_\infty)r'(u)du}=(1-\z)\be'+(1+\z)\al'
\end{equation}
where
\begin{equation}\label{eq22}
r(u)=\frac{\be'(\be-\al)\sn^2(u)-\be(\be'-\al)}{(\be-\al)\sn^2(u)-(\be'-\al)}= \frac{\be-\frac{\be'\sn^2(u)}{\sn^2(u_\infty)}}{1-\frac{\sn^2(u)}{\sn^2(u_\infty)}}\quad \textrm{and} \quad r'(u)=\frac{d}{du}r(u).
\end{equation}
Note that $r(\pm iK')=\be'$ and $r(K \pm iK')=\al'$. 
Integrating by parts gives
\begin{equation}\label{eq23}
\frac{2}{K}((K-u_\infty)\al'+ \be'u_\infty)-\int_{i K'}^{K+i K'}\frac{r(u)}{K}du-\int_{-i K'}^{K-i K'}\frac{r(u)}{K}du=(1-\z)\be'+(1+\z)\al'
\end{equation}
or equivalently
\begin{equation}\label{eq24}
\int_{i K'}^{K+i K'}r(u)du+\int_{-i K'}^{K-i K'}r(u)du=0.
\end{equation}
We can evaluate these integrals by first writing $r(u)$ in the form
\begin{equation}\label{eq25}
r(u)=\be+\left(\frac{\be-\be'}{\sn^2(u_\infty)}\right)\frac{\sn^2(u)}{1-\frac{\sn^2(u)}{\sn^2(u_\infty)}}.
\end{equation}
and using the functions
\begin{equation}\label{eq25a}
\Theta(u)=\th_4\left(\frac{\pi u}{2K}\right), \quad Z(u)=\frac{\Theta'(u)}{\Theta(u)}.
\end{equation}
The addition formulae for the sn and $Z$ functions are (see \cite{WW}):
\begin{equation}\label{eq25c}
\begin{aligned}
\sn(u \pm a)&=\frac{\sn(u)\cn(a)\dn(a) \pm \sn(a)\cn(u)\dn(u)}{1-k^2 \sn^2(a)\sn^2(u)}\,, \\
Z(u \pm a)&= Z(u) \pm Z(a) \mp k^2 \sn(u)\sn(a)\sn(u \pm a).
\end{aligned}
\end{equation}
Thus we have
\begin{equation}\label{eq25d}
\begin{aligned}
Z(u-a)-Z(u+a)+2Z(a)&=k^2\sn(u)\sn(a)\big(\sn(u + a)+\sn(u - a)\big)  \\
&=\frac{k^2\sn(u)\sn(a)\left[2\sn(u)\cn(a)\dn(a) \right]}{1-k^2 \sn^2(a)\sn^2(u)} \\
&=\frac{2k^2 \sn(a)\cn(a)\dn(a)\sn^2(u)}{1-k^2 \sn^2(a)\sn^2(u)}.
\end{aligned}
\end{equation}
We also have the half- and quarter-period identities
\begin{equation}\label{eq26e}
\sn(u+iK')=\frac{1}{k\sn(u)}\,, \quad \cn(u+iK')=\frac{-i}{k}\frac{\dn(u)}{\sn(u)}\,, \quad \dn(u+iK')=-i\frac{\cn(u)}{\sn(u)}. 
\end{equation}
In particular, notice that $\frac{1}{\sn(u_\infty)}=k\sn(u_\infty+iK')$.  Using the addition formulae (\ref{eq25d}) we can write $r(u)$ as
\begin{equation}\label{eq26}
\begin{aligned}
r(u)&=\be+\left(\frac{\be-\be'}{2k^2\sn(a)\cn(a)\dn(a)\sn^2(u_\infty)}\right)\left(Z(u-a)-Z(u+a)+2Z(a)\right),
\end{aligned}
\end{equation}
where $a=u_\infty+iK'$, (see Fig.~6).
From (\ref{eq26e}) and (\ref{eq19}), it follows that
\begin{equation}\label{eq26f}
\frac{\be-\be'}{2k^2\sn(u_\infty+iK')\cn(u_\infty+iK')\dn(u_\infty+iK')\sn^2(u_\infty)}=-K.
\end{equation}
Thus we can write (\ref{eq26}) as 
\begin{equation}\label{eq26g}
r(u)=\be-K\left[Z(u-u_\infty-iK')-Z(u+u_\infty+iK')+2Z(u_\infty+iK')\right].
\end{equation}

If we write $u=x+iK'$ in the first integral of (\ref{eq24}), and $u=x-iK'$ in the second, we obtain
\begin{equation}\label{eq26c}
\begin{aligned}
\int_0^K \big[2\be-4KZ(u_\infty+iK')-K[Z(x-u_\infty)-Z(x+u_\infty+2iK') \\
+Z(x-u_\infty-2iK')-Z(x+u_\infty)]\big]dx=0
\end{aligned}
\end{equation}
From the periodic properties of $\th_4$, it follows that
\begin{equation}\label{eq26a}
Z(u\pm 2iK')=Z(u)\mp \frac{\pi i}{K},
\end{equation}
so we can write (\ref{eq26c}) as 
\begin{equation}\label{eq26b}
\int_0^K \bigg(2\be-4KZ(u_\infty+iK')-2\pi i +2K\big[Z(x+u_\infty)-Z(x-u_\infty)\big]\bigg)dx=0
\end{equation}
This equation is readily integrated, as $Z$ is the logarithmic derivative of the $\Theta$ function.  Integrating gives
\begin{equation}\label{eq27a}
\begin{aligned}
0&=\left[(2\be-4KZ(u_\infty+iK')-2\pi i)x+2K\log\frac{\Theta(x+u_\infty)}{\Theta(x-u_\infty)}\right]_{x=0}^K \\
&=2K\be-4K^2Z(u_\infty+iK')-2K\pi i + 2K\log\left(\frac{\Theta(K+u_\infty)}{\Theta(K-u_\infty)}\frac{\Theta(-u_\infty)}{\Theta(u_\infty)}\right).
\end{aligned}
\end{equation}
The logarithmic term in this equation is zero due to the evenness and periodicity (period $2K$) of the $\Theta$ function and the fact that the relevant term in the integration is real on the entire contour of integration.  Thus we have that
\begin{equation}\label{eq27}
\be=2KZ(u_\infty+iK')+\pi i.
\end{equation}
From (\ref{main12}), we can deduce that
\begin{equation}\label{eq29}
Z(u_\infty+iK')=-\frac{\pi}{2K}\left(\frac{\th_2'(\frac{\om}{2})}{\th_2(\frac{\om}{2})}+i\right) 
\end{equation}
and write (\ref{eq27}) as
\begin{equation}\label{eq30}
\be=-\pi\,\frac{\th_2'(\frac{\om}{2})}{\th_2(\frac{\om}{2})}\,.
\end{equation}
This equation, together with equations (\ref{eq19d}) and (\ref{eq19b}), determine the end-points,
$\al,\al',\be',\be$. In fact, similar to (\ref{eq30})
we have the following explicit formulae for the other end-points:
\begin{equation}\label{eq34}
\begin{aligned}
\al&=-\pi\,\frac{\th_1'(\frac{\om}{2})}{\th_1(\frac{\om}{2})}\,,\qquad 
\al'=-\pi\,\frac{\th_4'(\frac{\om}{2})}{\th_4(\frac{\om}{2})}\,,\qquad 
\be'&=-\pi\,\frac{\th_3'(\frac{\om}{2})}{\th_3(\frac{\om}{2})}\,.\qquad 
\end{aligned}
\end{equation}
This follows from (\ref{eq19d}), (\ref{eq30}), and the identities (\ref{4}).
Similarly, in addition to the formulae (\ref{eq19d}) for distances between turning points, we get (\ref{eq8b}).

\medskip

{\it Proof of the formula of Zinn-Justin for $\frac{\al+\al'+\be'+\be}{4}$}.       
We immediately have from (\ref{eq34}) that
\begin{equation}
\frac{\al+\al'+\be'+\be}{4}=-\frac{\pi}{4}\left[\frac{\th_1'(\frac{\om}{2})}{\th_1(\frac{\om}{2})}+\frac{\th_2'(\frac{\om}{2})}{\th_2(\frac{\om}{2})}+\frac{\th_3'(\frac{\om}{2})}{\th_3(\frac{\om}{2})}+\frac{\th_4'(\frac{\om}{2})}{\th_4(\frac{\om}{2})}\right].
\end{equation}
From the identity (\ref{2}), we can deduce that
\begin{equation}\label{8}
\frac{\th_1'(2z)}{\th_1(2z)}=\frac{1}{2}\left[\frac{\th_1'(z)}{\th_1(z)}+\frac{\th_2'(z)}{\th_2(z)}
+\frac{\th_3'(z)}{\th_3(z)}+\frac{\th_4'(z)}{\th_4(z)}\right],
\end{equation}
thus we have 
\begin{equation}\label{9}
\frac{\al+\al'+\be'+\be}{4}=-\frac{\pi}{2}\frac{\th_1'(\om)}{\th_1(\om)}
=-\frac{\pi}{2}\frac{\th_1'(\frac{\pi}{2}+\frac{\pi\z}{2})}{\th_1(\frac{\pi}{2}+\frac{\pi\z}{2})}
=-\frac{\pi}{2}\frac{\th_2'(\frac{\pi\z}{2})}{\th_2(\frac{\pi\z}{2})}.
\end{equation}

{\bf Proof of Proposition \ref{proposition_equilibrium_density}.}
From equations (\ref{eq17}), (\ref{eq9}), and (\ref{eq18b}) we obtain formula (\ref{edf1}), 
cf. \cite{Z-J1}. From formula (\ref{edf1})
and equations (\ref{eq6a}), (\ref{eq7}) we obtain that the equilibrium density function $\rho(x)$ is given by 
formulae (\ref{edf3}).   We are left to prove formula (\ref{edf3a}).

By (\ref{eq9}), (\ref{eq18b}), and (\ref{edf3}), on the interval $[\be',\be]$,
\begin{equation}\label{m26}
\rho(x)=\frac{1}{iK\pi}u_+(x) \quad \textrm{for} \quad x\in[\be',\be].
\end{equation}
It follows that
\begin{equation}\label{m26a}
\int_{\be'}^{\be}\rho(x)dx=\frac{1}{iK \pi}\int_{\be'}^{\be}u_+(x)dx=\frac{1}{iK \pi}\int_{iK'}^{0}ur'(u)du,
\end{equation}
where $r(u)$ is defined in (\ref{eq22}). If we use equation (\ref{eq26g}), together with formula (\ref{eq27}), 
we can write $r(u)$ as
\begin{equation}\label{m26b}
r(u)=i\pi - K\big[Z(u-u_\infty-iK')-Z(u+u_\infty+iK')\big]
\end{equation}
Integrating (\ref{m26a}) by parts, we get
\begin{equation}\label{m27}
\int_{iK'}^{0}ur'(u)du=-\be'iK'-\pi K'
-K\left[\log\left(\frac{\Theta(-u_\infty-iK')\Theta(u_\infty+2iK')}{\Theta(-u_\infty)\Theta(u_\infty+iK')}\right)
\right].
\end{equation}
Using the fact that $\Theta$ is an even function and the identity
\begin{equation}\label{m27a}
\Theta(u+2iK')=e^{i\pi}e^{-i\pi \tau}e^{\frac{-i\pi u}{K}} \Theta(u),
\end{equation}
we can write (\ref{m27}) as
\begin{equation}\label{m27b}
\begin{aligned}
\int_{iK'}^{0}ur'(u)du&=-\be'iK'-\pi K'+K\left(i\pi+\pi\frac{K'}{K}-\frac{i\pi u_\infty}{K}\right) \\
&=i(K\pi -\be'K'-\pi u_\infty).
\end{aligned}
\end{equation}

\begin{rem}
There is some question here as to which branch of the logarithm to take, 
but it is clear that we have chosen the correct branch, as it is the only one that gives 
$0<\int_{\be'}^{\be}\rho(x)dx<1$.
\end{rem}

Thus, from (\ref{m26a}) and (\ref{m27b}), we have
\begin{equation}\label{m27c}
\int_{\be'}^{\be}\rho(x)dx=\frac{1}{iK\pi}i(K\pi -\be'K'-\pi u_\infty)
=1-\frac{\be' K'}{\pi K}-\frac{u_\infty}{K} =1-\frac{\be'}{2\ga}-\frac{1-\z}{2},
\end{equation}
hence by (\ref{edf3}),
\begin{equation}\label{m27d}
\int_{0}^{\be}\rho(x)dx=\int_{0}^{\be'}\rho(x)dx+\int_{\be'}^{\be}\rho(x)dx
=\frac{\be'}{2\ga}+1-\frac{\be'}{2\ga}-\frac{1-\z}{2}=\frac{1+\z}{2},
\end{equation}
which proves formula (\ref{edf3a}).

{\bf Proof of Proposition \ref{proposition_lagrange_multiplier}.}
By taking $x=\be$ we obtain from (\ref{g3}) that
\begin{equation}\label{lm2}
l=2g(\be)-V(\be)=2g(\be)-(1-\z)\be.
\end{equation}
We also have that
\begin{equation}\label{lm3}
\lim_{A\to\infty} [g(A)-\log A]=0,
\end{equation}
hence
\begin{equation}\label{lm4}
\begin{aligned}
l&=-2\lim_{A\to\infty}\left[g(A)-g(\be)-\log A\right]-(1-\z)\be\\
&=-2\lim_{A\to\infty}\left[\int_{\be}^A\om(z)\,dz-\log A\right]-(1-\z)\be.
\end{aligned}
\end{equation}
Writing this integral in terms of $u$ (so $z=r(u)$) gives
\begin{equation}\label{lm5}
\begin{aligned}
l&=2\lim_{A\to\infty}\left[\int_{0}^{B}\frac{1}{K}(u-u_{\infty})r'(u)\,du+\log A\right]-(1-\z)\be\\
&=2\lim_{B\to u_\infty}\left[\int_{0}^{B}\frac{1}{K}(u-u_{\infty})r'(u)\,du+\log r(B)\right]-(1-\z)\be,
\end{aligned}
\end{equation}
where $A=r(B)$. Integrating by parts gives
\begin{equation}\label{lm6}
\begin{aligned}
l=2\lim_{B\to u_\infty}\left[\left.\frac{1}{K}(u-u_{\infty})r(u)\right|_{u=0}^{B}
-\frac{1}{K}\int_{0}^{B}r(u)\,du+\log r(B)\right]-(1-\z)\be.
\end{aligned}
\end{equation}
From (\ref{eq22}) we obtain that $r(0)=\be$ and
\begin{equation}\label{lm7}
\begin{aligned}
\lim_{B\to u_\infty}(B-u_{\infty})r(B)=-\frac{(\be-\be')\sn(u_{\infty})}{2\sn'(u_{\infty})}
=-\frac{(\be-\be')\sn(u_{\infty})}{2\cn(u_{\infty})\dn(u_{\infty})}=-K\,,
\end{aligned}
\end{equation}
hence
\begin{equation}\label{lm8}
\begin{aligned}
l&=2\left[-1+\frac{\be(1-\z)}{2}\right]
-2\lim_{B\to u_\infty}\left[\frac{1}{K}\int_{0}^{B}r(u)\,du-\log r(B)\right]-(1-\z)\be\\
&=-2
-2\lim_{B\to u_\infty}\left[\frac{1}{K}\int_{0}^{B}r(u)\,du-\log r(B)\right].
\end{aligned}
\end{equation}
Using equation (\ref{m26b}) for $r(u)$, we immediately get that
\begin{equation}\label{lm9}
\frac{1}{K} \int_0^B r(u)du = \frac{B\pi i}{K}+\log\left[\frac{\Theta(B+u_\infty+iK')}{\Theta(B-u_\infty-iK')}\right]
\end{equation}
Now using equation (\ref{eq22}) for $r(u)$, we have
\begin{equation}\label{lm10}
\begin{aligned}
&\lim_{B\to u_\infty}\left[\frac{1}{K}\int_{0}^{B}r(u)\,du-\log r(B)\right]\\
&=\frac{u_\infty \pi i}{K}+\lim_{B\to u_\infty}\log\left[\frac{\Theta(B+u_\infty+iK')\big(\sn^2(u_\infty)-\sn^2(B)\big)}{\Theta(B-u_\infty-iK')\big(\be\sn^2(u_\infty)-\be'\sn^2(B)\big)}\right] \\
&=\frac{u_\infty \pi i}{K}+\log\left[\frac{\Theta(2u_\infty+iK')2\sn(u_\infty)\sn'(u_\infty)}{\Theta'(iK')(\be-\be')\sn^2(u_\infty)}\right] \\
&=\frac{u_\infty \pi i}{K}+\log\left[\frac{2e^{\frac{-i\pi u_\infty}{K}}\th_1(\frac{\pi u_\infty}{K})}{\pi \th_1'(0)}\right] \\
&=\log\left[\frac{2\th_1(\om)}{\pi \th_1'(0)}\right].
\end{aligned}
\end{equation}
Plugging this into (\ref{lm8}) gives
\begin{equation}\label{lm16}
l=-2+2\log\left(\frac{\pi \th_1'(0)}{2\th_1(\om)}\right),
\end{equation}
and thus we obtain that
\begin{equation}\label{lm18a}
e^{\frac{l}{2}}=\frac{\pi \th'_1(0)}{2e\th_1(\om)}.
\end{equation}

\section{Riemann-Hilbert approach: Interpolation problem}\label{RHA}

The Riemann-Hilbert approach to discrete orthogonal polynomials
 is based on the following Interpolation Problem (IP), which was
introduced in the paper \cite{BoB} of Borodin and Boyarchenko
under the name of the discrete Riemann-Hilbert problem.
See also the monograph  \cite{BKMM}
of Baik, Kriecherbauer, McLaughlin, and Miller, in which it is called the
Interpolation Problem. 
 
We will consider the lattice $L_n$ defined in (\ref{rw3}) and
the weight $w_n(x)$ defined in (\ref{rw1}). 

{\bf Interpolation Problem}. For a given $n=0,1,\ldots$, find a $2\times 2$ matrix-valued function
$\bold P_n(z)=(\bold P_{nij}(z))_{1\le i,j\le 2}$ with the following properties:
\begin{enumerate}
\item
{\it Analyticity}: $\bold P_n(z)$ is an analytic function of $z$ for $z\in\C\setminus L_n$.
\item
{\it Residues at poles}: At each node $x\in L_n$, the elements $\bold P_{n11}(z)$ and
$\bold P_{n21}(z)$ of the matrix $\bold P_n(z)$ are analytic functions of $z$, and the elements $\bold P_{n12}(z)$ and
$\bold P_{n22}(z)$ have a simple pole with the residues,
\begin{equation} \label{IP1}
\underset{z=x}{\rm Res}\; \bold P_{nj2}(z)=w_n(x)\bold P_{nj1}(x),\quad j=1,2.
\end{equation}
\item
{\it Asymptotics at infinity}: There exists a function $r(x)>0$ on  $L_n$ such that 
\begin{equation} \label{IP2a}
\lim_{x\to\infty} r(x)=0,
\end{equation} 
and such that as $z\to\infty$, $\bold P_n(z)$ admits the asymptotic expansion,
\begin{equation} \label{IP2}
\bold P_n(z)\sim \left( \I+\frac {\bold P_1}{z}+\frac {\bold P_2}{z^2}+\ldots\right)
\begin{pmatrix}
z^n & 0 \\
0 & z^{-n}
\end{pmatrix},\qquad z\in \C\setminus \left[\bigcup_{x\in L_n}^\infty D\big(x,r(x)\big)\right],
\end{equation}
where $D(x,r(x))$ denotes a disk of radius $r(x)>0$ centered at $x$ and $\I$ is the identity matrix.
\end{enumerate}

It is not difficult to see (see \cite{BoB} and \cite{BKMM}) that 
the IP has a unique solution, which is
\begin{equation} \label{IP3}
\bold P_n(z)=
\begin{pmatrix}
P_{nn}(z) & C(w_nP_{nn})(z) \\
(h_{n,n-1})^{-1}P_{n,n-1}(z) & (h_{n,n-1})^{-1}C(w_nP_{n,n-1})(z)
\end{pmatrix}
\end{equation}
where the Cauchy transformation $C$ is defined by the formula,
\begin{equation} \label{IP4}
C(f)(z)=\sum_{x\in L_n}\frac{f(x)}{z-x}\,,
\end{equation}
and $P_{nk}(z)=z^k+\ldots$ are monic polynomials orthogonal with the weight $w_n(x)$.
Because of the orthogonality condition, as $z\to\infty$,
\begin{equation} \label{IP5}
C(w_nP_{nk})(z)=\sum_{x\in L_n}^\infty\frac{w_n(x)P_{nk}(x)}{z-x}
\sim \sum_{x\in L_n}^\infty w_n(x)P_{nk}(x)\sum_{j=0}^\infty \frac{x^j}{z^{j+1}}=\frac{h_{nk}}{z^{k+1}}
+\sum_{j=k+2}^\infty \frac{a_j}{z^j}\,,
\end{equation}
which justifies asymptotic expansion (\ref{IP2}).
We have that 
\begin{equation} \label{IP6}
h_{nn}=[\bold P_1]_{12},\qquad h_{n,n-1}^{-1}=[\bold P_1]_{21}.
\end{equation}

\section{Reduction of IP to RHP}
 
\subsection{Preliminary Considerations} We would like to reduce
the Interpolation Problem to a Riemann-Hilbert Problem (RHP). Introduce the function,
\begin{equation} \label{redp1}
\Pi(z)=\frac{2\ga}{n\pi}\sin \frac{n \pi z}{2\ga }\,. 
\end{equation}
Observe that
\begin{equation} \label{redp3}
\Pi(x_k)=0,\quad \Pi'(x_k)=(-1)^k,\quad\exp\left(\frac{in \pi x_k}{2\ga}\right)
=(-1)^k\quad \textrm{for}\quad x_k=\frac{2\ga k}{n}\in L_n\,.
\end{equation}
Introduce the upper triangular matrices,
\begin{equation} \label{redp4}
\bold D^u_{\pm}(z)
=\begin{pmatrix}
1 & -\frac{w_n(z)}{\Pi(z)}e^{\pm \frac{in \pi z}{2\ga}}   \\
0 & 1
\end{pmatrix},
\end{equation}
and the lower triangular matrices,
\begin{equation} \label{redp5}
\bold D^l_{\pm}
=\begin{pmatrix}
\Pi(z)^{-1} & 0   \\
-\frac{1}{w_n(z)}e^{\pm \frac{in\pi z}{2\ga}} & \Pi(z)
\end{pmatrix}
=\begin{pmatrix}
\Pi(z)^{-1} & 0   \\
0 & \Pi(z)
\end{pmatrix}
\begin{pmatrix}
1 & 0   \\
-\frac{1}{\Pi(z)w_n(z)}e^{\pm \frac{in\pi z}{2\ga}} & 1
\end{pmatrix}.
\end{equation}
Define the matrix-valued functions,
\begin{equation} \label{redp6}
\bold R^u_n=\bold P_n(z)\times
\left\{
\begin{aligned}
& \bold D^u_+(z),\quad \textrm{when}\quad \Im z\ge0,\\
& \bold D^u_-(z),\quad \textrm{when}\quad \Im z\le0,
\end{aligned}
\right.
\end{equation}
and 
\begin{equation} \label{redp7}
\bold R^l_n=\bold P_n(z)\times
\left\{
\begin{aligned}
& \bold D^l_+(z),\quad \textrm{when}\quad \Im z\ge0,\\
& \bold D^l_-(z),\quad \textrm{when}\quad \Im z\le0.
\end{aligned}
\right.
\end{equation}
From (\ref{IP3}) we have that
\begin{equation} \label{redp8}
\begin{aligned}
\bold R_n^u(z)&=
\begin{pmatrix}
P_{nn}(z) & -\frac{w_n(z)P_{nn}(z)}{\Pi(z)}e^{\pm \frac{in\pi z}{2\ga}}
+C(w_n P_{nn})(z)   \\
h_{n,n-1}^{-1}P_{n,n-1}(z)  
& -\frac{w_n(z)h_{n,n-1}^{-1}P_{n,n-1}(z)}{\Pi(z)}e^{\pm \frac{in\pi z}{2\ga}}
+h_{n,n-1}^{-1}C(w_n P_{n,n-1})(z) 
\end{pmatrix}, \\
&\quad\textrm{when}\quad \pm \Im z\ge 0,
\end{aligned}
\end{equation}
and
\begin{equation} \label{redp9}
\begin{aligned}
\bold R_n^l(z)&=
\begin{pmatrix}
\frac{P_{nn}(z)}{\Pi(z)}-\frac{C(w_n P_{nn})(z)}{w_n(z)}e^{\pm \frac{in\pi z}{2 \ga}} 
  & \Pi(z)C(w_n P_{nn})(z)   \\
\frac{h_{n-1}^{-1}P_{n,n-1}(z)}{\Pi(z)}-\frac{h_{n,n-1}^{-1}C(w_n P_{n,n-1})(z)}{w_n(z)}e^{\pm \frac{in\pi z}{2\ga}} 
  & \Pi(z)h_{n,n-1}^{-1}C(w_n P_{n,n-1})(z)  
\end{pmatrix}, \\
&\quad\textrm{when}\quad \pm \Im z\ge 0.
\end{aligned}
\end{equation}
Observe that the functions $\bold R_n^u(z)$, $\bold R_n^l(z)$ are meromorphic on the closed quadrants of
the complex plane and they are two-valued on the real and imaginary axes. Their possible poles
are located on the lattice $L_n$. An important result is that, in fact, 
due to some cancellations they do not have any poles at all.
We have the following proposition.

\begin{prop}
The matrix-valued functions $\bold R_n^u(z)$ and $\bold R_n^l(z)$ have no poles and
on the real line they satisfy the following jump conditions at $x\in\R$:
\begin{equation} \label{redp9a}
\bold R_{n+}^u(x)=\bold R_{n-}^u(x) j_R^u(x),\qquad
j_R^u(x)=
\begin{pmatrix}
1 & -\frac{n\pi iw_n(x)}{\ga}\, \\
0 & 1
\end{pmatrix},
\end{equation}
and
\begin{equation} \label{redp9b}
\bold R_{n+}^l(x)=\bold R_{n-}^l(x) j_R^l(x),\qquad
j_R^l(x)=
\begin{pmatrix}
1 & 0 \\
-\frac{n\pi i}{\ga w_n(x)} & 1
\end{pmatrix},
\end{equation}
\end{prop}

\begin{proof} It follows from the definition of $\bold R_n^u(z)$ 
that all possible poles of $\bold R_n^u(z)$ are located on the lattice $L_n$.
Let us show that the residue of all these poles is equal to zero.
Consider any $x_k\in L_n$. The residue of the matrix element $\bold R_{n,12}^u(z)$ at $x_k$
is  equal to
\begin{equation} \label{redp10}
\underset{z=x_k}{\rm Res}\;\bold R_{n,12}^u(z)=
-\frac{w_n(x_k)P_{nn}(x_k)}{(-1)^k}(-1)^k+w_n(x_k)P_{nn}(x_k)=0.
\end{equation}
Similarly we get that
\begin{equation} \label{redp11}
\underset{z=x_k}{\rm Res}\;\bold R_{n,22}(z)=0,
\end{equation}
hence $\bold R_n^u(z)$ has no pole at $x_k$.

Similarly, the residue of the matrix element $\bold R_{n,11}^l(z)$ at $x_k$ is  equal to
\begin{equation} \label{redp12}
\underset{z=x_k}{\rm Res}\;\bold R_{n,11}^l(z)=\frac{P_{nn}(x_k)}{(-1)^k}
-\frac{w_n(x_k)P_{nn}(x_k)(-1)^k}{w_n(x_k)}=0.
\end{equation}
In the same way we obtain that
\begin{equation} \label{redp13}
\underset{z=x_k}{\rm Res}\;\bold R_{n,21}(z)=0.
\end{equation}
In the entry $\bold R_{n,21}^l(z)$, the pole of the function $C(w_n P_n)(z)$ at $z=x_k$ 
is cancelled by the zero of the function $\Pi(z)$,
hence $\bold R_{n,21}^l(z)$ has no pole at $x_k$. Similarly, $\bold R_{n,22}^l(z)$ has no pole at $x_k$
as well, hence $\bold R_n^l(z)$ has no pole at $x_k$. 

Let us evaluate to jump matrices at $x\in\R$. From (\ref{redp6}) we have that
\begin{equation} \label{redp14}
j_R^u(x)=\bold D^u_-(x)^{-1}\bold D^u_+(x)=
\begin{pmatrix}
1 & -\frac{w_n(x)}{\Pi(x)} 2i\sin \frac{n\pi x}{2\ga}   \\
0 & 1
\end{pmatrix}
=\begin{pmatrix}
1 & -\frac{n\pi iw_n(x)}{\ga}\, \\
0 & 1
\end{pmatrix},
\end{equation}
which proves (\ref{redp9a}). Similarly,
\begin{equation} \label{redp15}
j_R^l(x)=\bold D^l_-(x)^{-1}\bold D^l_+(x)=
\begin{pmatrix}
1 &  0  \\
-\frac{1}{\Pi(x)w_n(x)} 2i\sin \frac{n\pi x}{2\ga} & 1
\end{pmatrix}
=\begin{pmatrix}
1 & 0 \\
-\frac{n\pi i}{\ga w_n(x)} & 1
\end{pmatrix},
\end{equation}
which proves (\ref{redp9b}).\end{proof}

\subsection{Reduction of IP to RHP}
Let us discuss how to reduce the Interpolation Problem to a Riemann-Hilbert Problem.
We follow the work \cite{BKMM} with some modifications. Denote
\begin{equation} \label{red2}
\Delta=L_n\cap [\al',\be'],\qquad \nabla=L_n\setminus\Delta.
\end{equation}
Consider the oriented contour $\Sigma$ on the complex plane depicted on Fig.~7,
in which the horizontal lines are $\Im z=\ep, 0,-\ep$, where $\ep >0$ is a small positive constant 
which will be determined later, and the vertical segments pass through
the points $z=\al'$ and $z=\be'$. Consider the regions $\Om^\De_{\pm}$ and $\Om^\nabla_{\pm}$ 
bounded by the contour $\Sigma$, see Fig.~7. Observe that the regions $\Om^\nabla_{\pm}$ consist
of two connected components, to the left and to the right of $\Om^\Delta_{\pm}$. 
\begin{center}
 \begin{figure}[h]\label{sigma}
\begin{center}
   \scalebox{0.7}{\includegraphics{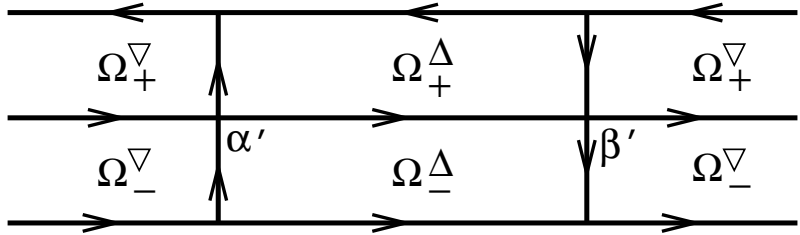}}
\end{center}
        \caption[The contour $\Sigma$.]{The contour $\Sigma$.}
    \end{figure}
\end{center}

Define 
\begin{equation} \label{red3}
\bold R_n(z)=
\left\{
\begin{aligned}
&\bold K_n\bold R_n^u(z)\bold K_n^{-1},\quad \textrm{for}\quad z\in\Om_{\pm}^\nabla,\\
&\bold K_n\bold R_n^l(z)\bold K_n^{-1},\quad \textrm{for}\quad z\in\Om_{\pm}^\De,\\
&\bold K_n \bold P_n(z)\bold K_n^{-1}, \quad \textrm{otherwise}.
\end{aligned}
\right.
\end{equation}
where $\bold K_n=\begin{pmatrix} 1 & 0 \\ 0 & -\frac{n\pi i}{\ga} \end{pmatrix}$.

\begin{center}
 \begin{figure}[h]\label{sigma_R}
\begin{center}
   \scalebox{0.7}{\includegraphics{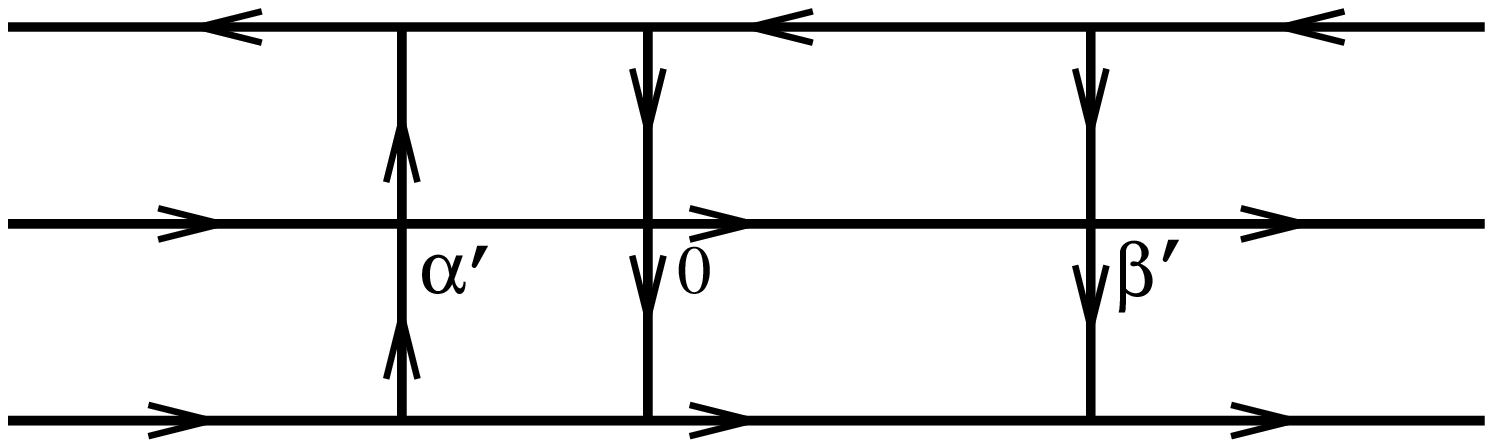}}
\end{center}
        \caption[The contour $\Sigma$.]{The contour $\Sigma_R$.}
    \end{figure}
\end{center}

Define a contour $\Sigma_R$ by adding to the contour $\Sigma$ a vertical segment
$[i\ep,-i\ep]$, see Fig.~8. If $A\subset \C$ is a set on the complex plane and $b\in\C$
then, as usual, we denote
\begin{equation} \label{red3a}
A+b=\{z=a+b,\; a\in A\}.
\end{equation}

\begin{prop}
The matrix-valued function $\bold R_n(z)$ has the following jumps on the contour $\Sigma_R$:
\begin{equation} \label{red4}
\bold R_{n+}(z)=\bold R_{n-}(z)j_R(z),
\end{equation}
where
\begin{equation} \label{red5}
j_R(z)=
\left\{
\begin{aligned}
& \begin{pmatrix} 1 & w(z) \\ 0 & 1 \end{pmatrix} \quad \textrm{when}\quad z\in (-\infty,\al')\cup (\be',\infty),\\
& \begin{pmatrix} 1 & 0 \\ -(\frac{n\pi}{\ga})^2w(z)^{-1} & 1 \end{pmatrix} \quad \textrm{when}\quad z\in  [\al',\be'],\\
& \bold K_n\bold D_{\pm}^u(z)\bold K_n^{-1}=\begin{pmatrix} 1 & -\frac{i\ga}{n\pi}
\frac{w_n(z)e^{\pm\frac{in \pi z}{2\ga}}}{\Pi(z)} \\ 0 & 1 \end{pmatrix} \\ 
&\hskip 3cm\textrm{when }\quad z\in  (-\infty,\al')\cup (\be',\infty)\pm i\ep,\\
&\bold K_n\bold D_{\pm}^l(z)\bold K_n^{-1}=
\begin{pmatrix}
\Pi(z)^{-1} & 0 \\ 
\frac{in\pi}{\ga}\frac{e^{\pm\frac{in\pi z}{2\ga}}}{w_n(z)} & \Pi(z) 
\end{pmatrix}
\quad \textrm{when }\quad z\in  (\al',\be')\pm i\ep\\
& \bold K_n\bold D_{\pm}^l(z)^{-1}\bold D_{\pm}^u(z)\bold K_n^{-1}=
\begin{pmatrix} 
\Pi(z) & \frac{\ga}{n\pi i}w_n(z)e^{\pm\frac{in\pi z}{2\ga}} \\ 
-\frac{n\pi i}{\ga}w_n(z)^{-1}e^{\pm\frac{in\pi z}{2\ga}} 
& \mp \frac{n\pi i}{\ga}e^{\pm\frac{in\pi z}{2\ga}} 
\end{pmatrix}  \\
&\hspace{3 cm} \textrm{when }\quad z\in  (0,\pm i\ep)+\al' \quad
\textrm{or}\quad z\in  (0,\pm i\ep)+ \be',\\
& \bold K_n\bold D^0_{\pm}(z)\bold K_n^{-1} \quad \textrm{when }\quad z\in  (0,\pm i\ep),
\end{aligned}
\right.
\end{equation}
and
\begin{equation} \label{red6}
\bold D^0_{\pm}(z)=
\begin{pmatrix}
1 &  0  \\
-\frac{2\sinh(nz)e^{-n\z z}e^{\pm \frac{in \pi z}{2\ga}}}{\Pi(z)}  & 1
\end{pmatrix}.
\end{equation}
\end{prop}
Notice that the jumps on vertical contours close to the origin, $\bold D^0_{\pm}(z)$, are exponentially close to the identity matrix.

\section{First transformation of the RHP}

Define the matrix function $\bold T_n(z)$ as follows from the equation,
\begin{equation}\label{ft1}
\bold R_n(z)=e^{\frac{nl}{2}\sigma_3}\bold T_n(z)e^{n(g(z)-\frac{l}{2})\sigma_3},
\end{equation}
where the Lagrange multiplier $l$ and the function $g(z)$ are as described 
in Section \ref{equilibrium} and $\sigma_3=\begin{pmatrix} 1 & 0 \\ 0 &-1 \end{pmatrix}$ is the third Pauli matrix.  
Then $\bold T_n(z)$ satisfies the following Riemann-Hilbert Problem:
\begin{enumerate}
  \item 
  $\bold T_n(z)$ is analytic in $\C \setminus \Sigma_R$. 
  \item $\bold T_{n+}(z)=\bold T_{n-}(z)j_T(z)$ for $z\in\Sg_R$, where
  \begin{equation}\label{ft2}
   j_T(z)=
  \left\{
  \begin{aligned}
 & e^{n(g_-(z)-\frac{l}{2})\sigma_3}j_R(z)e^{-n(g_+(z)-\frac{l}{2})\sigma_3} \quad \textrm{for} \quad z \in \R \\
  &e^{n(g(z)-\frac{l}{2})\sigma_3}j_R(z)e^{-n(g(z)-\frac{l}{2})\sigma_3} \quad \textrm{for} \quad z \in \Sigma_R \setminus \R.
  \end{aligned}\right.
  \end{equation}
  \item \begin{equation} \label{ft3}
\bold T_n(z)\sim  \I+\frac {\bold T_1}{z}+\frac {\bold T_2}{z^2}+\ldots,
\qquad \textrm{as} \quad z \to \infty.
\end{equation}
\end{enumerate}
From (\ref{g1}) we have that
\begin{equation} \label{ft3a}
g(z)=\log z+O(z^{-1}),
\qquad \textrm{as} \quad z \to \infty.
\end{equation}
This implies that
\begin{equation} \label{ft3b}
[\bold T_1]_{12}=e^{-nl}[\bold R_1]_{12}.
\end{equation}
Let's take a closer look at the behavior of the jump matrix $j_T$ described in (\ref{ft2}) on the horizontal segments of $\Sigma_R$.  We have that 
\begin{equation} \label{ft4}
j_T(z)=
\left\{
\begin{aligned}
&\begin{pmatrix} e^{-nG(z)} & e^{n(g_+(z)+g_-(z)-V(z)-l)} \\ 0 & e^{nG(z)} \end{pmatrix}
\quad \textrm{when}\quad z\in (-\infty,\al')\cup (\be',\infty),\\
&\begin{pmatrix} e^{-nG(z)} & 0 \\ -(\frac{n\pi}{\ga})^2 e^{-n(g_+(z)+g_-(z)-V(z)-l)} & e^{nG(z)} \end{pmatrix}\quad \textrm{when}\quad z\in  (\al',\be'),\\
&\begin{pmatrix} 1 & \pm\frac{e^{\pm nG(z)}}{1-e^{\frac{\mp i n\pi}{\ga}}e^{\frac{\ep n \pi x}{\ga}}} \\ 0 & 1 \end{pmatrix}\quad \textrm{when}\quad z= x \pm i \ep \in  (\al, \al')\cup (\be',\be)\pm i\ep,\\
&\begin{pmatrix} 1 & \pm\frac{e^{ n(2g(z)-l-V(z))}}{1-e^{\frac{\mp i n\pi x}{\ga}}e^{\frac{\ep n \pi}{\ga}}} \\ 0 & 1 \end{pmatrix}\quad \textrm{when}\quad z= x \pm i \ep \in  (-\infty, \al)\cup (\be,\infty)\pm i\ep,\\
&\begin{pmatrix} \Pi(z)^{-1} & 0 \\ \frac{in\pi}{\ga}e^{\pm\frac{in \pi x}{2\ga}}e^{-n(2g(z)-V(z)-l)} & \Pi(z)\end{pmatrix}\quad \textrm{when}\quad z\in  (\al',\be')\pm i\ep.
\end{aligned}\right.
\end{equation}
According to the properties of the $g$-function, we have the following proposition:
\begin{prop}
The jump function $j_T$ has the following large $n$ asymptotics:
\begin{equation}\label{ft6}
j_T(z)=
\left\{
\begin{aligned}
&\begin{pmatrix} e^{-nG(z)} & 0 \\ O(e^{-nC(z)}) & e^{nG(z)} \end{pmatrix} 
\quad &\textrm{for} \quad &z\in  (\al', \be'), \\
&\begin{pmatrix} e^{-nG(z)} & 1 \\ 0 & e^{nG(z)} \end{pmatrix}
\quad &\textrm{for} \quad &z\in (\al, \al')\cup(\be',\be), \\
&\begin{pmatrix}1 & O(e^{-nC(z)}) \\ 0 & 1 \end{pmatrix} 
\quad &\textrm{for} \quad &z\in  (-\infty,\al)\cup(\be,\infty), \\
&\begin{pmatrix} 1 & e^{\pm nG(z)}O(e^{-\frac{\ep n\pi}{\ga}}) \\ 0 & 1 \end{pmatrix}
 \quad &\textrm{for}\quad &z \in  (\al,\al')\cup (\be',\be)\pm i\ep,
\end{aligned} \right. 
\end{equation}
where $C(z)$ is a positive continuous function on any subset of the given interval which is bounded away from the endpoints of each interval and satisfies
\begin{equation}\label{ft7}
C(z)>c|z+1| \ \textrm{for some} \ c>0.
\end{equation}
\end{prop}

\section{Second transformation of the RHP}

The second transformation is based on two observations.  The first is the well known 
``opening of the lenses'' in a neighborhood of the unconstrained support of the equilibrium measure.  
Namely, notice that, for $x\in (\al,\al')\cup(\be',\be)$, the jump matrix $j_T(x)$ factorizes as 
\begin{equation}\label{st1}
\begin{aligned}
j_T(x)=\begin{pmatrix} e^{-nG(z)} & 1 \\ 0 & e^{nG(z)} \end{pmatrix}=&\begin{pmatrix} 1 & 0 \\ e^{nG(x)} & 1 \end{pmatrix}
\begin{pmatrix} 0 & 1 \\ -1 & 0 \end{pmatrix} \begin{pmatrix} 1 & 0 \\ e^{-nG(x)} & 1 \end{pmatrix} \\
&=j_-(x)j_M j_+(x),
\end{aligned}
\end{equation}
which allows us to reduce the jump matrix $j_T$ to the one $j_M$ plus asymptotically small 
jumps on the lens boundaries.
The second observation consists of two facts.  Firstly, the jumps on the segments 
$ [\al', \be'] \pm i\ep$ behave, for large $n$, as $\pm e^{\pm\frac{in\pi z}{2\ga}}$.  
Secondly, note that, for $x\in [\al', \be']$, $G(x)$ is a linear function with slope $-\frac{\pi i}{\ga}$.  
With these facts in mind, we make the second transformation of the RHP.  Let
\begin{equation}\label{st2}
\bold S_n(z)=\left\{
\begin{aligned}
&\bold T_n(z)j_+(z)^{-1} \quad \textrm{for} \quad z\in \{(\al, \al')\cup (\be',\be)\}\times (0,i\ep) \\
&\bold T_n(z)j_-(z)\quad \textrm{for} \quad z\in \{(\al, \al')\cup (\be',\be)\}\times (0,-i\ep) \\
&\bold T_n(z)\begin{pmatrix} -\frac{\ga}{n\pi i}e^{-\frac{in\pi z}{2\ga}} & 0 \\ 0 & -\frac{n\pi i}{\ga}e^{\frac{in\pi z}{2\ga}} \end{pmatrix} \quad \textrm{for} \quad z\in (\al', \be') \times (0,i\ep) \\
&\bold T_n(z)\begin{pmatrix} \frac{\ga}{n\pi i}e^{\frac{in\pi z}{2\ga}} & 0 \\ 0 & \frac{n\pi i}{\ga}e^{-\frac{in\pi z}{2\ga}} \end{pmatrix} \quad \textrm{for} \quad z\in (\al', \be') \times (0,-i\ep) \\
&\bold T_n(z)  \quad \textrm{otherwise}.
\end{aligned}\right.
\end{equation}

This function satisfies a similar RHP to $\bold T$, but jumps now occur on a new contour, $\Sigma_S$, 
which is obtained from $\Sigma_R$ by adding the two segments $(\al-i\ep, \al+i\ep)$ and $(\be-i\ep, \be+i\ep)$, see Fig.~9.

\begin{center}
 \begin{figure}[h]\label{sigma_S}
\begin{center}
   \scalebox{0.7}{\includegraphics{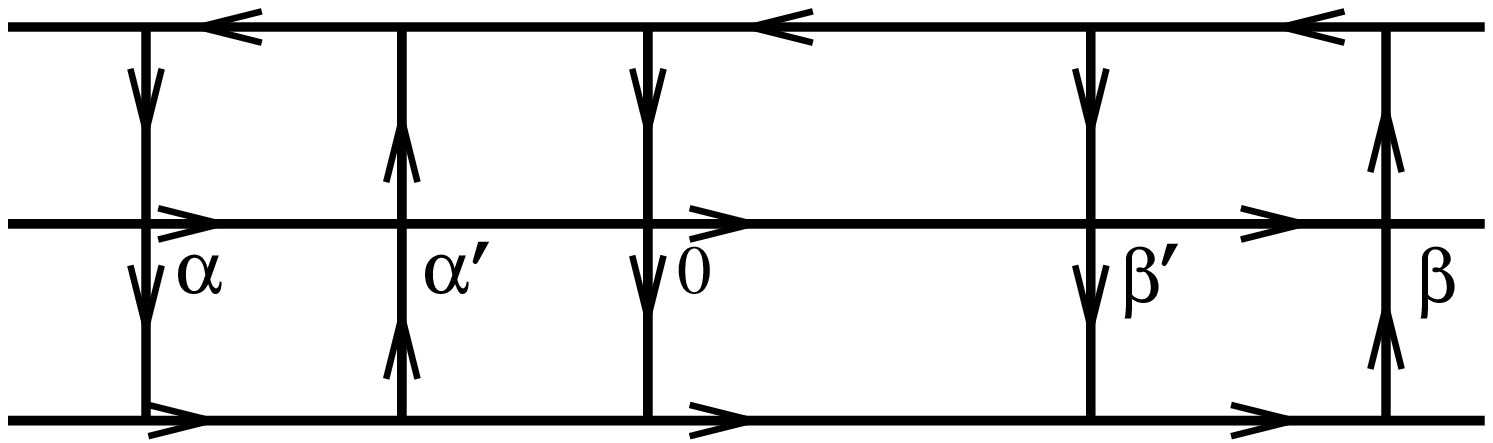}}
\end{center}
        \caption[The contour $\Sigma_S$.]{The contour $\Sigma_S$.}
    \end{figure}
\end{center}

On the horizontal segments for which the jump function $j_S$ differs from $j_T$ , we have that, as $n\to \infty$,
\begin{equation}\label{st3}
j_S(z)=\left\{
\begin{aligned}
& \begin{pmatrix} 0 & 1 \\ -1 & 0 \end{pmatrix}
 \quad &\textrm{for} \quad z\in (\al, \al')\cup (\be',\be) \\
& \begin{pmatrix} 1+O(e^{-\ep n\pi/\ga}) & O(e^{n[G(z)-\ep \pi/\ga]})  \\ -e^{-nG(z)} & 1\end{pmatrix} \quad &\textrm{for} \quad z-i\ep\in (\al, \al')\cup (\be',\be)\\
&\begin{pmatrix} 1+O(e^{-\ep n\pi/\ga}) & O(e^{n[-G(z)-\ep n\pi/\ga]}) \\ e^{nG(z)} & 1\end{pmatrix} \quad &\textrm{for} \quad z+i\ep\in (\al, \al')\cup (\be',\be) \\
&\begin{pmatrix} 1+O(e^{-\ep n\pi / \ga}) & 0 \\ \frac{n\pi i}{\ga}e^{-n(2g(z)-l-V(z))} & 1+O(e^{-\ep n\pi / \ga}) \end{pmatrix} \quad &\textrm{for} \quad z \in [\al',\be'] \pm i\ep \\
&\begin{pmatrix} -e^{-n\pi i(1+\z)} & 0 \\ e^{-n(g_+(z)+g_-(z)-l-V(z))} &  -e^{n\pi i(1+\z)}\end{pmatrix} \quad &\textrm{for} \quad z\in [\al',\be'].
\end{aligned}\right.
\end{equation}
By formula (\ref{g6}) for the $G$-function and the upper constraint (\ref{eq4b}) on the density $\rho$, we obtain that, for sufficiently small $\ep > 0$,
\begin{equation}
0 < \mp \Re G(x\pm i\ep) = 2\pi \rho(x) +O(\ep^2) < \frac{\pi \ep}{\ga} + O(\ep^2). 
\end{equation}
This, combined with property (\ref{g3}) of the $g$-function, imply that all jumps on horizontal segments are exponentially close to the identity matrix, provided that they are bounded away from the interval $[\al,\be]$.
For what follows we denote
\begin{equation} \label{st4}
\Om_n=\pi+n2\pi\int_0^\be \rho(x)\,dx=\pi+n\pi(1+\z),
\end{equation}
so that
\begin{equation} \label{st5}
-e^{-n\pi i(1+\z)}=e^{-i\Om_n}.
\end{equation}

\section{Model RHP}

The model RHP appears when we drop in the jump matrix $j_S(z)$ the terms that vanish as $n\to\infty$:
\begin{enumerate}
  \item 
  $\bold M(z)$ is analytic in $\C \setminus [\al,\be]$. 
  \item $\bold M_{+}(z)=\bold M_{-}(z)j_M(z)$ for $z\in[\al,\be]$, where
  \begin{equation}\label{m1}
   j_M(z)=
  \left\{
  \begin{aligned}
 & \begin{pmatrix} 0 & 1 \\ -1 & 0 \end{pmatrix} \quad \textrm{for} \quad z \in [\al,\al']\cup[\be,\be'] \\
  &\; e^{-i\Omega_n\sigma_3} \quad \textrm{for} \quad z \in [\al',\be'].
  \end{aligned}\right.
  \end{equation}
 \item 
\begin{equation} \label{m2}
\bold M(z)\sim \I+\frac {\bold M_1}{z}+\frac {\bold M_2}{z^2}+\ldots
\qquad \textrm{as} \quad z \to \infty.
\end{equation}
\end{enumerate}

This model problem was first solved, in the general multi-cut case, in \cite{DKMVZ}, and is solved in two steps.  In the first step, we solve the following auxiliary RHP:
\begin{enumerate}
\item $\bold Q(z)$ is analytic in $\C\setminus [\al,\al']\cup[\be',\be]$, 
\item $\bold Q_+(z)= \bold Q_-(z)\begin{pmatrix} 0&1 \\ -1&0 \end{pmatrix}$ for $z\in [\al,\al']\cup[\be',\be]$ 
\item $\bold Q(z)=\I+O(z^{-1})$ as $z\to\infty$.
\end{enumerate}
This RHP has the unique solution (see \cite{DKMVZ})
\begin{equation}\label{m2a}
\bold Q(z) = \begin{pmatrix} \frac{\ga(z)+\ga^{-1}(z)}{2} &  \frac{\ga(z)-\ga^{-1}(z)}{-2i} \\ \frac{\ga(z)-\ga^{-1}(z)}{2i} & \frac{\ga(z)+\ga^{-1}(z)}{2} \end{pmatrix}
\end{equation}
where
\begin{equation}\label{m2b}
\ga(z)=\left(\frac{(z-\al)(z-\be')}{(z-\al')(z-\be)}\right)^{1/4}
\end{equation}
with cuts on $[\al,\al']\cup[\be',\be]$.
 
To solve the model RHP described in (\ref{m1}) and (\ref{m2}), we again use elliptic functions.  Define the function
\begin{equation}\label{m3}
f(s)=\frac{\th_3(s+d+c)}{\th_3(s+d)}
\end{equation}
where $\th_3$ is as defined in (\ref{main8}) with elliptic nome $q=e^{i\pi\tau}=e^{\frac{-\pi^2}{2\ga}}$ ( $\tau=\frac{i\pi}{2\ga}$), and $d$ and $c$ are arbitrary complex numbers.   Notice that $f$ has the following periodic properties:
\begin{equation}\label{m4}
f(s+\pi)=f(s); \quad f(s+\pi\tau)=e^{-2ic}f(s)
\end{equation}
and that $f$ is an even function.  Now let 
\begin{equation}\label{m5}
\tilde{u}(z)=\frac{\pi}{2K}u(z)=\frac{\pi}{2}\int_\be^z\frac{dz'}{\sqrt{R(z')}}
\end{equation}
where $u$ is as defined in (\ref{eq9}).  Then $\tilde{u}$ is two-valued on $[\al, \be]$ and satisfies 
\begin{equation}\label{m6}
\tilde{u}_+(x)-\tilde{u}_-(x)=\pi\tau \quad \textrm{for} \quad x\in[\al',\be']. 
\end{equation}
Also,
\begin{equation}\label{m6a}
\tilde u_\pm(\al)=\frac{\pi}{2}\,,\quad 
\tilde u_\pm(\al')=\frac{\pi}{2}\pm\frac{\pi\tau}{2}\,,\quad 
\tilde u_\pm(\be')=\pm\frac{\pi\tau}{2}\,,\quad 
\tilde u_\pm(\be)=0\,, 
\end{equation}
cf. Fig.~6.
Because $\sqrt{R(x)}_+=-\sqrt{R(x)}_-$ for $x\in[\al,\al']\cup[\be',\be]$, it immediately follows that 
\begin{equation}\label{m7}
\tilde{u}_+(x)+\tilde{u}_-(x)=0 \quad \textrm{for} \quad x\in[\be',\be], 
\end{equation}
and that
\begin{equation}\label{m8}
\begin{aligned}
\tilde{u}_+(x)+\tilde{u}_-(x)&=\tilde{u}_+(\al')-\tilde{u}_+(\be')+\tilde{u}_-(\al')-\tilde{u}_-(\be')=\pi \quad \textrm{for} \quad x\in[\al,\al'].
\end{aligned}
\end{equation}

We now define
\begin{equation}\label{m9}
f_1(z)=\frac{\th_3(\tilde{u}(z)+d+\frac{\Omega_n}{2})}{\th_3(\tilde{u}(z)+d)}\,, \quad f_2(z)=\frac{\th_3(-\tilde{u}(z)+d+\frac{\Omega_n}{2})}{\th_3(-\tilde{u}(z)+d)}\quad \textrm{for}\quad z\in \C\setminus [\al,\be]
\end{equation}
where $d$ is an arbitrary complex number.
It then follows from (\ref{m4}) and(\ref{m6}) that
\begin{equation}\label{m10}
f_{1+}(x)=e^{-i\Omega_n}f_{1-}(x) \quad \textrm{and} \quad f_{2+}(x)=e^{i\Omega_n}f_{2-}(x) \quad \textrm{for} \quad x\in [\al',\be'],
\end{equation}
and from (\ref{m4}), (\ref{m7}), and (\ref{m8}) that
\begin{equation}\label{m11}
f_{1+}(x)=f_{2-}(x) \quad \textrm{and} \quad f_{2+}(x)=f_{1-}(x) \quad \textrm{for} \quad x\in[\al,\al']\cup[\be'\be].
\end{equation}
Define the matrix valued function
\begin{equation}\label{m12}
\bold F(z)=\begin{pmatrix} \frac{\th_3(\tilde{u}(z)+d_1+\frac{\Omega_n}{2})}{\th_3(\tilde{u}(z)+d_1)} & \frac{\th_3(-\tilde{u}(z)+d_1+\frac{\Omega_n}{2})}{\th_3(-\tilde{u}(z)+d_1)} \\ \frac{\th_3(\tilde{u}(z)+d_2+\frac{\Omega_n}{2})}{\th_3(\tilde{u}(z)+d_2)} & \frac{\th_3(-\tilde{u}(z)+d_2+\frac{\Omega_n}{2})}{\th_3(-\tilde{u}(z)+d_2)}\end{pmatrix}
\end{equation}
where $d_1$ and $d_2$ are yet undetermined complex constants.  Then, from (\ref{m10}) and(\ref{m11}) we have that
\begin{equation}\label{m13}
\begin{aligned}
\bold F_+(x)&=\bold F_-(x)\begin{pmatrix} e^{-i\Omega_n} & 0 \\ 0 & e^{i\Omega_n} \end{pmatrix} \quad &\textrm{for} \quad &x\in[\al',\be'] \\
\bold F_+(x)&=\bold F_-(x)\begin{pmatrix} 0 & 1 \\ 1 & 0 \end{pmatrix} \quad &\textrm{for} \quad &x\in[\al,\al']\cup[\be',\be].
\end{aligned}
\end{equation}

We can now combine (\ref{m2a}) and (\ref{m12}) to obtain
\begin{equation}\label{m14}
\bold M(z)= \bold F(\infty)^{-1}\begin{pmatrix} \frac{\ga(z)+\ga^{-1}(z)}{2} \frac{\th_3(\tilde{u}(z)+d_1+\frac{\Omega_n}{2})}{\th_3(\tilde{u}(z)+d_1)} & \frac{\ga(z)-\ga^{-1}(z)}{-2i} \frac{\th_3(-\tilde{u}(z)+d_1+\frac{\Omega_n}{2})}{\th_3(-\tilde{u}(z)+d_1)} \\ \frac{\ga(z)-\ga^{-1}(z)}{2i} \frac{\th_3(\tilde{u}(z)+d_2+\frac{\Omega_n}{2})}{\th_3(\tilde{u}(z)+d_2)} &  \frac{\ga(z)+\ga^{-1}(z)}{2} \frac{\th_3(-\tilde{u}(z)+d_2+\frac{\Omega_n}{2})}{\th_3(-\tilde{u}(z)+d_2)}\end{pmatrix}
\end{equation}
where
\begin{equation}\label{m15}
\bold F(\infty) = \begin{pmatrix} \frac{\th_3(\tilde{u}_\infty+d_1+\frac{\Omega_n}{2})}{\th_3(\tilde{u}_\infty+d_1)} & 0 \\ 0 &  \frac{\th_3(-\tilde{u}_\infty+d_2+\frac{\Omega_n}{2})}{\th_3(-\tilde{u}_\infty+d_2)}\end{pmatrix}.
\end{equation}
and $\tilde{u}_\infty \equiv \tilde{u}(\infty)$.
This matrix satisfies conditions (\ref{m1}) and (\ref{m2}) of the model RHP, but may not be analytic on $\C\setminus [\al,\be]$, as it may have some poles at the zeroes of $\th_3(\pm \tilde{u}(z) + d_{1,2})$.  However, we can choose the constants $d_1$ and $d_2$ such that these zeroes coincide with the zeroes of $\ga(z)\pm\ga^{-1}(z)$ and are thus cancelled in the product.  

First consider the zeroes of $\ga(z)\pm\ga^{-1}(z)$  These are the zeroes of $\ga^2(z)\pm1$ and thus of  $\ga^4(z)-1$, thus there is only one zero, which uniquely solves the equation
\begin{equation}\label{m16}
p(z)\equiv \frac{(z-\al)(z-\be')}{(z-\al')(z-\be)}=1, 
\end{equation}
which is 
\begin{equation}\label{m17}
x_0=\frac{\be\al'-\al\be'}{(\al'-\al)+(\be-\be')} \in (\al', \be').
\end{equation}
It is easy to check that $\ga(x_0)=1$, thus $x_0$ is the unique zero of $\ga(z)-\ga^{-1}(z)$, whereas there are no zeroes of $\ga(z)+\ga^{-1}(z)$ on the specified sheet.  We use here the change of variables $v$ defined in (\ref{eq9a}).  Notice that, by (\ref{eq19}),
\begin{equation}\label{m17a}
v(x_0)=\frac{\be'-\al}{\be'-\al'}=\frac{\dn^2(u_\infty)}{k^2\cn^2(u_\infty)}
\end{equation}
implying that
\begin{equation}\label{m17b}
\sn^2\big(u(x_0)\big)=\frac{\dn^2(u_\infty)}{k^2\cn^2(u_\infty)}.
\end{equation}
Since $x_0\in (\al', \be')$, we must have $u(x_0)\in(iK', K+iK')$ (if we choose to take $u_+$).  Since $\sn^2$ is a one-to-one function on this interval there is a unique point $u_0 \in (iK', K+iK')$ such that $\sn^2(u_0)=\frac{\dn^2(u_\infty)}{k^2\cn^2(u_\infty)}$.  The simple period identity
\begin{equation}\label{m17c}
\sn(u+K+iK')=\frac{\dn(u)}{k\cn(u)}
\end{equation} 
along with (\ref{m17b}) gives that we must have
\begin{equation}\label{m17d}
u_0=u(x_0)=K-u_\infty+iK'
\end{equation}
thus
\begin{equation}\label{m17e}
\tilde{u}(x_0)=\frac{\pi}{2K}(K-u_\infty+iK')=\frac{\tau\pi}{2}+\frac{\pi}{2}-\tilde{u}_\infty.
\end{equation}

We now consider zeroes of the function $\th_3(\tilde{u}(z)-d)\equiv \th_3(-\tilde{u}(z)+d)$.  The zeroes of this function are the solutions to the equation
\begin{equation}\label{m18}
\tilde{u}(z)-d=(2m+1)\frac{\pi}{2}+(2k+1)\frac{\tau\pi}{2}
\end{equation}
for any $m,k\in\Z$.  Because $\tilde{u}$ maps the first sheet of $X$ to the rectangular domain $[0,\frac{\pi}{2}]\times[-\frac{\tau\pi}{2},\frac{\tau\pi}{2}]$, it is clear that this equation can have at most one solution, and without any loss of generality we may take $m=k=0$.  Then, if we want the solution of this equation to be $x_0$, we need to let
\begin{equation}\label{m19}
d=\tilde{u}(x_0)-\frac{\pi}{2}(1+\tau)=-\tilde{u}_\infty.
\end{equation}
This choice of $d$ also ensures that $\th_3(\tilde{u}(z)+d)\equiv \th_3(-\tilde{u}(z)-d)$ has no zeroes on the first sheet of $X$.  We can then let 
\begin{equation}\label{m20}
d_1=d, \quad d_2=-d
\end{equation}
so that (\ref{m14}) and (\ref{m15}) become
\begin{equation}\label{m21}
\begin{aligned}
\bold M(z)&= \bold F(\infty)^{-1}\begin{pmatrix} \frac{\ga(z)+\ga^{-1}(z)}{2} \frac{\th_3(\tilde{u}(z)+d+\frac{\Omega_n}{2})}{\th_3(\tilde{u}(z)+d)} & \frac{\ga(z)-\ga^{-1}(z)}{-2i} \frac{\th_3(-\tilde{u}(z)+d+\frac{\Omega_n}{2})}{\th_3(-\tilde{u}(z)+d)} \\ \frac{\ga(z)-\ga^{-1}(z)}{2i} \frac{\th_3(\tilde{u}(z)-d+\frac{\Omega_n}{2})}{\th_3(\tilde{u}(z)-d)} &  \frac{\ga(z)+\ga^{-1}(z)}{2} \frac{\th_3(-\tilde{u}(z)-d+\frac{\Omega_n}{2})}{\th_3(-\tilde{u}(z)-d)}\end{pmatrix} \\
&=\bold F(\infty)^{-1}\begin{pmatrix} \frac{\ga(z)+\ga^{-1}(z)}{2} \frac{\th_3(\tilde{u}(z)+(n+\frac{1}{2})\om)}{\th_4(\tilde{u}(z)+\frac{\om}{2})} & \frac{\ga(z)-\ga^{-1}(z)}{-2i} \frac{\th_3(\tilde{u}(z)-(n+\frac{1}{2})\om)}{\th_4(\tilde{u}(z)-\frac{\om}{2})} \\ \frac{\ga(z)-\ga^{-1}(z)}{2i} \frac{\th_3(\tilde{u}(z)+(n-\frac{1}{2})\om)}{\th_4(\tilde{u}(z)-\frac{\om}{2})} &  \frac{\ga(z)+\ga^{-1}(z)}{2} \frac{\th_3(\tilde{u}(z)-(n-\frac{1}{2})\om)}{\th_4(\tilde{u}(z)+\frac{\om}{2})}\end{pmatrix}
\end{aligned}
\end{equation}
where
\begin{equation}\label{m22}
\bold F(\infty) = \begin{pmatrix} \frac{\th_3(\frac{\Omega_n}{2})}{\th_3(0)} & 0 \\ 0 &  \frac{\th_3(\frac{\Omega_n}{2})}{\th_3(0)}\end{pmatrix}=\begin{pmatrix} \frac{\th_4(n\om)}{\th_3(0)} & 0 \\ 0 &  \frac{\th_4(n\om)}{\th_3(0)}\end{pmatrix},
\end{equation}
solving the model RHP.  The asymptotics at infinity are
\begin{equation}\label{m23}
\bold M(z)=I + \frac{\bold M_1}{z} + O(z^{-2})
\end{equation}
where the matrix $\bold M_1$ has the form
\begin{equation}\label{m24}
\bold M_1=\begin{pmatrix} * & \frac{\th_3(-\tilde{u}_\infty+d+\frac{\Omega_n}{2})\th_3(\tilde{u}_\infty+d)}{\th_3(\tilde{u}_\infty+d+\frac{\Omega_n}{2})\th_3(-\tilde{u}_\infty+d)}\frac{(\be-\be')+(\al'-\al)}{-4i} \\  \frac{\th_3(\tilde{u}_\infty-d+\frac{\Omega_n}{2})\th_3(-\tilde{u}_\infty-d)}{\th_3(-\tilde{u}_\infty-d+\frac{\Omega_n}{2})\th_3(\tilde{u}_\infty-d)}\frac{(\be-\be')+(\al'-\al)}{4i} & * \end{pmatrix}.
\end{equation}
The matrix $M_1$ can be written in a cleaner fashion and in terms of the original parameters as follows:  

\begin{prop} \label{proposition_model}We have that
\begin{equation}\label{m36}
 \begin{aligned}
\left[\bold M_1\right]_{12}=\frac{i A(\om) \th_4\big((n+1)\om\big)}{\th_4(n\om)}\,,\qquad
 \left[\bold M_1\right]_{21}=\frac{A(\om) \th_4(n\om)}{i\th_4\big((n-1)\om\big)}\,.
 \end{aligned}
 \end{equation}
where
\begin{equation}\label{m37}
\om=\frac{\pi(1+\z)}{2}\,,\qquad  A(\om)=\frac{\pi \th_1'(0)}{2\th_1(\om)}\,.
\end{equation}
\end{prop} 

For a proof of this proposition see Appendix \ref{proof_model} in the end of the paper.
Notice that since $\bold M$ solves the model RHP, we have that
\begin{equation}\label{m38}
\det\bold M(z)=1,\qquad z\in\C.
\end{equation}

 \medskip
 
\section{Parametrix at outer turning points}

We now consider small disks $D(\al,\ep)$ and $D(\be,\ep)$ centered at the outer turning points.  Denote $D=D(\al,\ep) \cup D(\be,\ep)$.  We will seek a local parametrix $\bold U_n(z)$ defined on $D$ such that\begin{enumerate}
\item 
\begin{equation}\label{pm0}
\bold U_n(z) \ \textrm{is analytic on} \ D \setminus \Sigma_S.
\end{equation}
\item
\begin{equation}\label{pm1}
\bold U_{n+}(z)=\bold U_{n-}(z)j_S(z) \quad \textrm{for} \quad z\in D \cap \Sigma_S.
\end{equation}
\item
\begin{equation}\label{pm2}
\bold U_n(z)=\bold M(z) \big(I+O(n^{-1})\big) \quad \textrm{uniformly for} \  z\in \partial D.
\end{equation}
\end{enumerate}
We first construct the parametrix near $\be$.  The jumps $j_S$ are given by
\begin{equation}\label{pm3}
j_S(z)=
\left\{
\begin{aligned}
&\begin{pmatrix} 0 & 1 \\ -1 & 0 \end{pmatrix} \quad \textrm{for} \ z\in (\be-\ep,\be), \\
&\begin{pmatrix} 1& 0 \\- e^{-nG(z)} & 1 \end{pmatrix}\quad \textrm{for} \ z\in (\be,\be+i\ep), \\
&\begin{pmatrix} 1& 0 \\ e^{nG(z)} & 1 \end{pmatrix}\quad \textrm{for} \ z\in (\be,\be-i\ep), \\
&\begin{pmatrix} e^{-nG(z)} & e^{n(g_+(z)+g_-(z)-V(z)-l)} \\ 0 &  e^{nG(z)} \end{pmatrix} \quad \textrm{for} \ z\in (\be,\be+\ep).
\end{aligned}\right.
\end{equation}
If we let
\begin{equation}\label{pm4}
\bold U_n(z)=\bold Q_n(z)e^{-n(g(z)-\frac{V(z)}{2}-\frac{l}{2})\sigma_3},
\end{equation}
then the jump conditions on $\bold Q_n$ become
\begin{equation}\label{pm5}
\bold Q_{n+}(z)=\bold Q_{n-}(z)j_Q(z)
\end{equation}
where
\begin{equation}\label{pm6}
j_Q(z)=
\left\{
\begin{aligned}
&\begin{pmatrix} 0 & 1 \\ -1 & 0 \end{pmatrix} \quad \textrm{for} \ z\in (\be-\ep,\be), \\
&\begin{pmatrix} 1& 0 \\ -1 & 1 \end{pmatrix}\quad \textrm{for} \ z\in (\be,\be+i\ep), \\
&\begin{pmatrix} 1& 0 \\ 1 & 1 \end{pmatrix}\quad \textrm{for} \ z\in (\be,\be-i\ep), \\
&\begin{pmatrix} 1 & 1 \\ 0 &  1 \end{pmatrix} \quad \textrm{for} \ z\in (\be,\be+\ep).
\end{aligned}\right.
\end{equation}
where orientation is from left to right on horizontal contours, and down to up on vertical contours, according to Figure \ref{sigma_S}.

$\bold Q_n$ can be constructed using Airy functions.  The Airy function solves the differential equation $y''=zy$, and has the following asymptotics at infinity:
\begin{equation}\label{pm6a}
\begin{aligned}
&\Ai (z) = \frac{1}{2\sqrt{\pi}z^{1/4}}e^{-\frac{2}{3}z^{3/2}}\left(1-\frac{5}{48}z^{-3/2}+O(z^{-3})\right) \\
&\Ai' (z) = -\frac{1}{2\sqrt{\pi}}z^{1/4}e^{-\frac{2}{3}z^{3/2}}\left(1+\frac{7}{48}z^{-3/2}+O(z^{-3})\right) \\
\end{aligned}
\end{equation}
as $\ z\to \infty$ with $\arg z \in (-\pi+\ep, \pi-\ep)$ for any $\ep > 0$.
If we let
\begin{equation}\label{pm7}
y_0(z)=\Ai (z), \quad y_1(z)=\omega \Ai (\omega z), \quad y_2(z)=\omega^2 \Ai (\omega^2 z)
\end{equation}
where $\omega = e^{\frac{2\pi i}{3}}$, then the functions $y_0, y_1,$ and $y_2$ satisfy the relation
\begin{equation}\label{pm8}
y_0(z)+y_1(z)+y_2(z)=0.
\end{equation}
If we take
\begin{equation}\label{pm9}
\Phi_\be(z)=\left\{
\begin{aligned}
&\begin{pmatrix} y_0(z) & -y_2(z) \\ y_0'(z) & -y_2'(z) \end{pmatrix} \quad \textrm{for} \quad  \arg z \in \left(0,\frac{\pi}{2}\right), \\
&\begin{pmatrix} -y_1(z) & -y_2(z) \\ -y_1'(z) & -y_2'(z) \end{pmatrix} \quad \textrm{for} \quad  \arg z \in \left(\frac{\pi}{2},\pi\right), \\
&\begin{pmatrix} -y_2(z) & y_1(z) \\ -y_2'(z) & y_1'(z) \end{pmatrix} \quad \textrm{for} \quad  \arg z \in \left(-\pi,-\frac{\pi}{2}\right), \\
&\begin{pmatrix} y_0(z) & y_1(z) \\ y_0'(z) & y_1'(z) \end{pmatrix} \quad \textrm{for} \quad  \arg z \in \left(-\frac{\pi}{2},0\right),
\end{aligned}
\right.
\end{equation}
then $\Phi_\be$ satisfies jump conditions similar to (\ref{pm6}), but for jumps on rays emanating from the origin rather than from $\be$.  We thus need to map the disk $D(\be, \ep)$ onto some convex neighborhood of the origin in order to take advantage of the function $\Phi_\be$.  Our mapping should match the asymptotics of the Airy function in order to have the matching property (\ref{pm2}).

To this end notice that, by (\ref{edf3}), for $t \in [\be',\be]$, as $t\to\be$,
\begin{equation}\label{pm9a}
\rho(t)=C(\be-t)^{1/2}+O\big((\be-t)^{3/2}\big), \quad C>0.
\end{equation}
It follows that, as $z\to \be$ for $z \in (\be',\be)$,
\begin{equation}\label{pm10}
\int_z^{\be}\rho(t)dt=C_0(\be-z)^{3/2}+O\big((\be-z)^{5/2}\big)\,, \quad C_0=\frac{2}{3} C.
\end{equation}
Thus
\begin{equation}\label{pm10a}
\psi_\be(z)=-\left\{\frac{3\pi}{2}\int_z^{\be}\rho(t)dt\right\}^{2/3}
\end{equation}
is analytic at $\be$, and so extends to a conformal map from $D(\be, \ep)$ (for small enough $\ep$) onto a convex neighborhood of the origin.  Furthermore, 
\begin{equation}\label{pm10b}
\psi_\be(\be)=0\,,  \quad \psi_\be'(\be)>0\,;
\end{equation}
therefore $\psi_\be$ is real negative on $(\be-\ep, \be)$, and real positive on $(\be, \be+\ep)$.  Also, we can slightly deform the vertical pieces of the contour $\Sigma_S$ close to $\be$, so that
\begin{equation}\label{pm10bb}
\psi_\be\{D(\be,\ep) \cap \Sigma_S\}=(-\ep,\ep) \cup (-i\ep,i\ep).
\end{equation}
We now set
\begin{equation}\label{pm10c}
\bold Q_n(z)=\bold E_n^\be(z)\Phi_\be\big(n^{2/3}\psi_\be(z)\big)
\end{equation}
so that 
\begin{equation}\label{pm11}
\bold U_n(z)=\bold E_n^\be(z)\Phi_\be\big(n^{2/3}\psi_\be(z)\big)e^{-n(g(z)-\frac{V(z)}{2}-\frac{l}{2})\sigma_3}
\end{equation}
where
\begin{equation}\label{pm12}
\bold E_n^\be(z)=\bold M(z)\bold L_n^\be(z)^{-1}\,, \quad \bold L_n^\be(z)=\frac{1}{2\sqrt{\pi}}\begin{pmatrix} n^{-1/6} \psi_\be^{-1/4}(z) & 0 \\ 0 & n^{1/6} \psi_\be^{1/4}(z)\end{pmatrix}\begin{pmatrix} 1 & i \\ -1 & i\end{pmatrix}\,,
\end{equation}
where we take the branch of $ \psi_\be^{1/4}$ which is positive on $(\be, \be+\ep)$ and has a cut on  $(\be-\ep,\be)$.  We claim that $\bold E_n^\be(z)$ is analytic in $D(\be, \ep)$, thus $\bold U(z)$ has the jump conditions of $j_S$.  This is clear, as both $\bold M$ and $\bold L_n^\be$ have jump the same constant jump, $\begin{pmatrix} 0 & 1 \\ -1 & 0 \end{pmatrix}$, on the interval $(\be-\ep, \be]$, and are analytic elsewhere.  The only other possible singularity for either $\bold M$ or $\bold L_n^\be$ is the isolated singularity at $\be$, and this is at most a fourth-root singularity, thus removable.  It follows that $\bold E_n^\be(z)=\bold M(z)\bold L_n^\be(z)^{-1}$ is analytic on $D(\be, \ep)$, thus $\bold U_n$ has the prescribed jumps in $D(\be, \ep)$.

We are left only to prove the matching condition (\ref{pm2}).  Using (\ref{pm6a}), one can check that, for $z$ in each of the sectors of analyticity, $\Phi_\be(n^{2/3}\psi_\be(z))$ satisfies the following asymptotics as $n\to\infty$:
\begin{equation}\label{pm13}
\begin{aligned}
\Phi_\be\big(n^{2/3}\psi_\be(z)\big)&=\frac{1}{2\sqrt{\pi}}n^{\frac{1}{6}\sigma_3}\psi_\be(z)^{-\frac{1}{4}\sigma_3}
\left[ \begin{pmatrix} 1& i \\ -1 & i \end{pmatrix} + \frac{\psi_\be(z)^{-3/2}}{48n}
\begin{pmatrix} -5 & 5i \\ -7 & -7i \end{pmatrix}+O(n^{-2})\right]\\
&\qquad \times e^{-\frac{2}{3}n\psi_\be(z)^{3/2}\sigma_3}
\end{aligned}
\end{equation}
where we always take the principal branch of $\psi_\be(z)^{3/2}$.  As such, $\psi_\be(z)^{3/2}$ is two-valued for $z\in (\be-\ep, \be)$, so that 
\begin{equation}\label{pm14}
\left[\frac{2}{3}\psi_\be(x)^{3/2}\right]_{\pm}=\mp \pi i\int_x^\be \rho(t)dt.
\end{equation}
Notice that, by (\ref{g3}) and (\ref{g6a}), for $x \in (\be-\ep, \be)$, 
\begin{equation}\label{pm15}
2g_{\pm}(x)-V(x)=l\pm2\pi i\int_x^\be\rho(t)dt.
\end{equation}
This implies that
\begin{equation}\label{pm16}
[2g_\pm(\be)-V(\be)]-[2g_\pm(x)-V(x)]=\mp 2\pi i \int_x^\be \rho(t)dt\,. 
\end{equation}
Combining these equations with (\ref{pm14}) gives
\begin{equation}\label{pm17}
\left[\frac{2}{3}\psi_\be(x)^{3/2}\right]_{\pm}=\frac{1}{2}\bigg[\big(2g_{\pm}(\be)-V(\be)\big)-\big(2g_{\pm}(x)-V(x)\big)\bigg].
\end{equation}
This equation can be extended into the upper and lower planes, respectively, giving
\begin{equation}\label{pm18}
\frac{2}{3}\psi_\be(z)^{3/2}=\frac{1}{2}\bigg[\big(2g_{\pm}(\be)-V(\be)\big)-\big(2g(z)-V(z)\big)\bigg] \quad \textrm{for} \ \pm \Im z >0.
\end{equation}
Since, by (\ref{pm15}), $2g_{\pm}(\be)-V(\be)=l$, we get that
\begin{equation}\label{pm19}
\frac{2}{3}\psi_\be(z)^{3/2}=-g(z)+\frac{V(z)}{2}+\frac{l}{2}
\end{equation}
for $z$ throughout $D(\be,\ep)$.
Plugging (\ref{pm13}) and (\ref{pm19}) into (\ref{pm11}), we get
\begin{equation}\label{pm20}
\begin{aligned}
\bold U_n(z)&=\bold M(z)\bold L_n^\be(z)^{-1}\frac{1}{2\sqrt{\pi}}n^{-\frac{1}{6}\sigma_3}\psi_\be(z)^{-\frac{1}{4}\sigma_3}
\bigg[ \begin{pmatrix} 1& i \\ -1 & i \end{pmatrix} + \frac{\psi_\be(z)^{-3/2}}{48n}
\begin{pmatrix} -5 & 5i \\ -7 & -7i \end{pmatrix}\\
&\quad +O(n^{-2})\bigg] \ e^{n(g(z)-\frac{V(z)}{2}-\frac{l}{2})\sigma_3}e^{-n(g(z)-\frac{V(z)}{2}-\frac{l}{2})\sigma_3} \\
&=\bold M(z)\left[ I + \frac{\psi_\be(z)^{-3/2}}{48n}\begin{pmatrix} 1 & 6i \\  6i & -1 \end{pmatrix}+O(n^{-2})\right] .
\end{aligned}
\end{equation}
Thus we have that $\bold U_n$ satisfies conditions (\ref{pm0}), (\ref{pm1}), and (\ref{pm2}).

A similar construction gives the parametrix at the $\al$.  Namely, if we let
\begin{equation}\label{pm21}
\psi_\al(z)=-\left\{\frac{3\pi}{2}\int_\al^z \rho(t)dt\right\}^{2/3},
\end{equation}
then $\psi_\al$ is analytic throughout $D(\al,\ep)$, real valued on the real line, and has negative derivative at $\al$.  Close to $\al$, the jumps $j_Q$ become
\begin{equation}\label{pm21a}
j_Q(z)=
\left\{
\begin{aligned}
&\begin{pmatrix} 1 & 1 \\ 0 & 1 \end{pmatrix} \quad \textrm{for} \ z\in (\al-\ep,\al), \\
&\begin{pmatrix} 1& 0 \\ -1 & 1 \end{pmatrix}\quad \textrm{for} \ z\in (\al,\al+i\ep), \\
&\begin{pmatrix} 1& 0 \\ 1 & 1 \end{pmatrix}\quad \textrm{for} \ z\in (\al,\al-i\ep), \\
&\begin{pmatrix} 0 & 1 \\ -1 & 0 \end{pmatrix} \quad \textrm{for} \ z\in (\al,\al+\ep),
\end{aligned}\right.
\end{equation}
where orientation is taken left to right on horizontal contours, and up to down on vertical contours according to Figure \ref{sigma_S}.  After the change of variables $\psi_\al$ (and a slight deformation of vertical contours), these jumps become the following jumps close to the origin:
\begin{equation}\label{pm21b}
j_Q\big(\psi_\al(z)\big)=\left\{
\begin{aligned}
&\begin{pmatrix} 0 & 1 \\ -1 & 0 \end{pmatrix} \quad \textrm{for} \ \psi_\al(z)\in (-\ep,0), \\
&\begin{pmatrix} 1 & 0 \\ 1 & 1 \end{pmatrix}\quad \textrm{for} \ \psi_\al(z)\in (0,i\ep), \\
&\begin{pmatrix} 1 & 0 \\ -1 & 1 \end{pmatrix}\quad \textrm{for} \ \psi_\al(z)\in (0,-i\ep), \\
&\begin{pmatrix} 1 & 1 \\ 0 & 1 \end{pmatrix} \quad \textrm{for} \ \psi_\al(z)\in (0,\ep),
\end{aligned}\right.
\end{equation}
where orientation is taken right to left on horizontal contours, and down to up on vertical contours.
These jump conditions are satisfied by the function
\begin{equation}\label{pm21c}
\Phi_\al(z)=\Phi_\be(z)\begin{pmatrix} 1 & 0 \\ 0 & -1 \end{pmatrix}.
\end{equation}
Then we can take 
\begin{equation}\label{pm22}
\bold U_n(z)=\bold E_n^\al(z)\Phi_\al\big(n^{2/3}\psi_\al(z)\big)e^{-n(g(z)-\frac{V(z)}{2}-\frac{l}{2})\sigma_3}
\end{equation}
for $z\in D(\al,\ep)$, where
\begin{equation}\label{pm23}
\begin{aligned}
&\bold E_n^\al(z)=\bold M(z)\bold L_n^\al(z)^{-1}, \\
& \bold L_n^\al(z)=\frac{1}{2\sqrt{\pi}}
\begin{pmatrix} n^{-1/6} \psi_\al^{-1/4}(z) & 0 \\ 0 & n^{1/6} \psi_\al^{1/4}(z)\end{pmatrix}
\begin{pmatrix} -1 & i \\ 1 & i\end{pmatrix}.
 \end{aligned}
\end{equation}
Similar to (\ref{pm13}), we have that in each sector of analyticity, $\Phi_\al(n^{2/3}\psi_\al(z))$ satisfies
\begin{equation}\label{pm24}
\begin{aligned}
\Phi_\al\big(n^{2/3}\psi_\al(z)\big)&=\frac{1}{2\sqrt{\pi}}
n^{-\frac{1}{6}\sigma_3}\psi_\al(z)^{-\frac{1}{4}\sigma_3}
\bigg[ \begin{pmatrix} 1& -i \\ -1 & -i \end{pmatrix} + \frac{\psi_\al(z)^{-3/2}}{48n}
\begin{pmatrix} -5 & -5i \\ -7 & 7i \end{pmatrix}+O(n^{-2})\bigg]\\
&\qquad \times e^{-\frac{2}{3}n\psi_\al(z)^{3/2}\sigma_3}.
\end{aligned}
\end{equation}
Once again, we have that, for $x\in (\al,\al+\ep)$, $\psi_\al(x)^{3/2}$ takes limiting values from above and below, so that
\begin{equation}\label{pm25}
\left[\frac{2}{3}\psi_\al(x)^{3/2}\right]_{\pm}=\pm \pi i\int_\al^x \rho(t)dt.
\end{equation}
In analogue to (\ref{pm17}), we have
\begin{equation}\label{pm25a}
\frac{2}{3}\psi_\al(z)^{3/2}=\frac{1}{2}\bigg[\big(2g_{\pm}(\al)-V(\al)\big)-\big(2g(z)-V(z)\big)\bigg]
 \quad \textrm{for} \ \pm \Im z >0.
\end{equation}
Since, by (\ref{pm15}), $2g_{\pm}(\al)-V(\al)=l\pm\pi i$, we get that
\begin{equation}\label{pm26}
\frac{2}{3}\psi_\al(z)^{3/2}=-g(z)+\frac{V(z)}{2}+\frac{l}{2}\pm\pi i \quad \textrm{for} \quad \pm \Im z >0.
\end{equation}
Plugging (\ref{pm26}) into (\ref{pm22}) and (\ref{pm24}) gives, as $n\to \infty$,
\begin{equation}\label{pm27}
\begin{aligned}
\bold U_n(z)&=\bold M(z)\bold L_n^\al(z)^{-1}\frac{1}{2\sqrt{\pi}}
n^{\frac{1}{6}\sigma_3}\psi_\al(z)^{-\frac{1}{4}\sigma_3}
\bigg[ \begin{pmatrix} 1& -i \\ -1 & -i \end{pmatrix} + \frac{\psi_\al(z)^{-3/2}}{48n}
\begin{pmatrix} -5 & -5i \\ -7 & 7i \end{pmatrix}\\
&\quad +O(n^{-2})\bigg]  e^{n(g(z)-\frac{V(z)}{2}-\frac{l}{2})\sigma_3}
\begin{pmatrix} -1 & 0 \\ 0 & -1 \end{pmatrix}e^{-n(g(z)-\frac{V(z)}{2}-\frac{l}{2})\sigma_3} \\
&=\bold M(z)\begin{pmatrix} -1 & i \\ 1 & i \end{pmatrix}^{-1}
\left[ \begin{pmatrix} -1& i \\ 1 & i \end{pmatrix}+ \frac{\psi_\al(z)^{-3/2}}{48n}
\begin{pmatrix} 5 & 5i \\ 7 & -7i \end{pmatrix}+O(n^{-2})\right] \\
&=\bold M(z)\left[ I + \frac{\psi_\al(z)^{-3/2}}{48n}
\begin{pmatrix} 1 & -6i \\  -6i & -1 \end{pmatrix}+O(n^{-2})\right] .
\end{aligned}
\end{equation}

\medskip

\section{Parametrix at the inner turning points}

We now consider small disks $D(\al',\ep)$ and $D(\be',\ep)$ centered at the inner turning points.  Denote $\tilde{D}=D(\al',\ep) \cup D(\be',\ep)$.  We will seek a local parametrix $\bold U_n(z)$ defined on $\tilde{D}$ such that
\begin{enumerate}
\item 
\begin{equation}\label{pmi0}
\bold U_n(z) \ \textrm{is analytic on} \ \tilde{D} \setminus \Sigma_S.
\end{equation}
\item
\begin{equation}\label{pmi1}
\bold U_{n+}(z)=\bold U_{n-}(z)j_S(z) \quad \textrm{for} \quad z\in \tilde{D} \cap \Sigma_S.
\end{equation}
\item
\begin{equation}\label{pmi2}
\bold U_n(z)=\bold M(z) \big(I+O(n^{-1})\big) \quad \textrm{uniformly for} \  z\in \partial \tilde{D}.
\end{equation}
\end{enumerate}
We first construct the parametrix near $\al'$.  Let 
\begin{equation}\label{pmi3}
\bold U_{n}(z)=
\tilde{\bold Q}_{n}(z)e^{\mp \frac{in\pi z}{2\ga}\sigma_3}e^{-n(g(z)-\frac{V(z)}{2}-\frac{l}{2})\sigma_3} \quad \textrm{for} \quad \pm \Im z > 0\,.
\end{equation}
Then the jumps for $\tilde{\bold Q}_n$ are
\begin{equation}\label{pmi4}
j_{\tilde{Q}}(z)=\left\{
\begin{aligned}
&\begin{pmatrix} 0 & 1 \\ -1 & 0 \end{pmatrix} \quad &\textrm{for}\quad z \in (\al'-\ep,\al'), \\
&\begin{pmatrix} -1 & 0 \\ -1 & -1 \end{pmatrix}\quad & \textrm{for}\quad z \in (\al',\al'+\ep), \\
&\begin{pmatrix} 1 & -1 \\ 0 & 1 \end{pmatrix} \quad & \textrm{for}\quad z \in (\al',\al'+i \ep), \\
&\begin{pmatrix} 1 & 1 \\ 0 & 1 \end{pmatrix}  \quad & \textrm{for}\quad z \in (\al',\al'-i \ep), \\
\end{aligned}\right.
\end{equation}
where orientation is taken from left to right on horizontal contours, and down to up on vertical contours according to Figure \ref{sigma_S}.
A proof of this statement is given in Appendix \ref{proof_of_jump}.
We now take
\begin{equation}\label{pmi5}
\Phi_{\al'}(z)=\left\{
\begin{aligned}
&\begin{pmatrix}y_2(z) & -y_0(z) \\ y_2'(z) & -y_0'(z) \end{pmatrix} 
\quad \textrm{for} \quad \arg z \in \left(0,\frac{\pi}{2}\right), \\
&\begin{pmatrix}y_2(z) & y_1(z) \\ y_2'(z) & y_1'(z) \end{pmatrix} 
\quad \textrm{for} \quad \arg z \in \left(\frac{\pi}{2},\pi \right), \\
&\begin{pmatrix}y_1(z) & -y_2(z) \\ y_1'(z) & -y_2'(z) \end{pmatrix} 
\quad \textrm{for} \quad \arg z \in \left(-\pi,-\frac{\pi}{2}\right), \\
&\begin{pmatrix}y_1(z) & y_0(z) \\ y_1'(z) & y_0'(z) \end{pmatrix} 
\quad \textrm{for} \quad \arg z \in \left(-\frac{\pi}{2},0 \right).
\end{aligned}
\right.
\end{equation}
Then $\Phi_{\al'}(z)$ solves a RHP similar to that of $\tilde{\bold Q}_n$, 
but for jumps emanating from the origin rather than from $\al'$.

Notice that, by (\ref{edf3}), for $t \in [\al,\al']$, as $t\to\al'$,
\begin{equation}\label{pmi6}
\rho(t)=\frac{1}{2\ga}-C(\al'-t)^{1/2}+O\big((\al'-t)^{3/2}\big), \quad C>0.
\end{equation}
It follows that, as $z\to \al'$, for $z \in (\al,\al')$,
\begin{equation}\label{pmi7}
\int_z^{\al'}\left(\frac{1}{2\ga}-\rho(t)\right)dt=C(\al'-z)^{3/2}+O((\al'-z)^{5/2})\,, \quad C_0=\frac{2}{3}C 
\end{equation}
Thus,
\begin{equation}\label{pmi8}
\psi_{\al'}(z)=-\left\{\frac{3\pi}{2}\int_z^{\al'}\left(\frac{1}{2\ga}-\rho(t)\right)dt\right\}^{2/3}
\end{equation}
is analytic at $\al'$, and so extends to a conformal map from $D(\al', \ep)$ onto a convex neighborhood of the origin.  Furthermore, 
\begin{equation}\label{pmi9}
\psi_{\al'}(\al')=0\,, \quad \psi_{\al'}'(\al')>0;
\end{equation}
consequently, $\psi_{\al'}$ is real negative on $(\al'-\ep, \al')$, and real positive on $(\al', \al'+\ep)$.  Again, we can slightly deform the vertical pieces of the contour $\Sigma_S$ close to $\al'$, so that
\begin{equation}\label{pm10bbb}
\psi_{\al'}\big\{D(\al',\ep) \cap \Sigma_S\big\}=(-\ep,\ep) \cup (-i\ep,i\ep)
\end{equation}
We now take 
\begin{equation}\label{pmi11}
\tilde{\bold Q}_n(z)=\bold E_n^{\al'}(z)\Phi_{\al'}\big(n^{2/3}\psi_{\al'}(z)\big)
\end{equation}
where
\begin{equation}\label{pmi12}
\begin{aligned}
&\bold E^{\al'}_n(z)=\bold M(z)e^{\pm\frac{i\Omega_n}{2}\sigma_3}\tilde{\bold L}_n(z)^{-1} \quad \textrm{for} \quad \pm \Im z \geq 0, \\
&\bold L^{\al'}_n(z)=\frac{1}{2\sqrt{\pi}}\begin{pmatrix} n^{-1/6} \psi_{\al'}^{-1/4}(z) & 0 \\ 0 & n^{1/6} \psi_{\al'}^{1/4}(z)\end{pmatrix}\begin{pmatrix} 1 & i \\ 1 & -i \end{pmatrix},
\end{aligned}
\end{equation}
and we take the branch of $ \psi_{\al'}^{1/4}$ which is positive on $(\al', \al'+\ep)$ and has a cut on  $(\al'-\ep,\al')$.  $\bold U_n$ then becomes
\begin{equation}\label{pmi13}
\begin{aligned}
\bold U_{n}(z)=
\bold M(z)e^{\pm \frac{i\Omega_n}{2}\sigma_3}\bold L^{\al'}_n(z)^{-1}\Phi_{\al'}\big(n^{2/3}\psi_{\al'}(z)\big)e^{\mp \frac{in\pi z}{2\ga}\sigma_3}e^{-n(g(z)-\frac{V(z)}{2}-\frac{l}{2})\sigma_3} \\ \hskip 5cm \textrm{for} \quad \pm \Im z > 0 .
\end{aligned}
\end{equation}
The function $\Phi_{\al'}(n^{2/3}\psi_{\al'}(z))$ has the jumps $j_S$, and we claim that the prefactor $\bold E_n^{\al'}$ is analytic in $D(\al', \ep)$, thus does not change these jumps.  This can be seen, as 
\begin{equation}\label{pmi14}
\bold M_+(z)e^{\frac{i\Omega_n}{2}\sigma_3}=\bold M_-(z)e^{-\frac{i\Omega_n}{2}\sigma_3}e^{\frac{i\Omega_n}{2}\sigma_3}j_Me^{\frac{i\Omega_n}{2}\sigma_3}
\end{equation}
thus the jump for the function $\bold M(z)e^{\pm\frac{i\Omega_n}{2}\sigma_3}$ is
\begin{equation}\label{pmi15}
e^{\frac{i\Omega_n}{2}\sigma_3}j_Me^{\frac{i\Omega_n}{2}\sigma_3}=\left\{
\begin{aligned}
& e^{\frac{i\Omega_n}{2}\sigma_3}\begin{pmatrix} 0 & 1 \\ -1 & 0 \end{pmatrix} e^{\frac{i\Omega_n}{2}\sigma_3}\quad \textrm{for} \quad z \in (\al'-\ep,\al'), \\
 &e^{\frac{i\Omega_n}{2}\sigma_3} e^{-i\Omega_n\sigma_3} e^{\frac{i\Omega_n}{2}\sigma_3}\quad \textrm{for} \quad z \in (\al',\al'+\ep),
\end{aligned}\right.
\end{equation}
or equivalently, 
\begin{equation}\label{pmi16}
e^{\frac{i\Omega_n}{2}\sigma_3}j_Me^{\frac{i\Omega_n}{2}\sigma_3}=\left\{
\begin{aligned}
&\begin{pmatrix} 0 & 1 \\ -1 & 0 \end{pmatrix} \quad \textrm{for} \quad z \in (\al'-\ep,\al'), \\
 &\begin{pmatrix} 1 & 0 \\ 0 &1 \end{pmatrix}\quad \textrm{for} \quad z \in (\al',\al'+\ep),
\end{aligned}\right.
\end{equation}
which is exactly the same as the jump conditions for $\bold L^{\al'}_n$.  Thus $\bold E^{\al'}_n(z)=\bold M(z)e^{\pm\frac{i\Omega_n}{2}\sigma_3}\bold L^{\al'}_n(z)^{-1}$ has no jumps in $D(\al',\ep)$.  The only other possible singularity for $\bold E^{\al'}_n$ is at $\al'$, and this singularity is at most a fourth root singularity, thus removable.  Thus, $\bold E^{\al'}_n$ is analytic in $D(\al', \ep)$, and $\tilde{\bold Q}_n$ has the prescribed jumps.

We left check that $\bold U_n$ satisfies the matching condition (\ref{pmi2}).  The large $n$ asymptotics of $\Phi_{\al'}(n^{2/3}\psi_{\al'}(z))$ are given in the different regions of analyticity as follows:
\begin{equation}\label{pmi17}
\begin{aligned}
\Phi_{\al'}\big(n^{2/3}\psi_{\al'}(z)\big)&=\frac{1}{2\sqrt{\pi}}
n^{-\frac{1}{6}\sigma_3}\psi_{\al'}(z)^{-\frac{1}{4}\sigma_3}
\bigg[\mp \begin{pmatrix} i & 1 \\ i & -1 \end{pmatrix}
\pm \frac{\psi_{\al'}(z)^{-3/2}}{48n}\begin{pmatrix} -5i & 5 \\ 7i & 7 \end{pmatrix}\\
&\quad +O(n^{-2})\bigg]
e^{\frac{2}{3}n\psi_{\al'}(z)^{3/2}\sigma_3} 
\hspace{1 cm} \textrm{for}  \quad \pm \Im z>0, \\
\end{aligned}
\end{equation}
where we always take the principal branch of $\psi_{\al'}(z)^{3/2}$.  As such, $\psi_{\al'}(z)^{3/2}$ is
two-valued for $x \in (\al'-\ep, \al)$, so that
\begin{equation}\label{pmi18}
\left[\frac{2}{3}\psi_{\al'}(x)^{3/2}\right]_{\pm}=\mp\pi i\int_x^{\al'}\left(\frac{1}{2\ga}-\rho(t)\right)dt
=\mp\frac{\pi i}{2\ga}(\al'-x)\pm\pi i\int_x^{\al'}\rho(t)dt.
\end{equation}
From (\ref{g3}) and (\ref{g6a}), we have that
\begin{equation}\label{pmi19}
2g_+(x)-V(x)=l+2\pi i \int_x^\be \rho(t)dt\,,\quad 2g_-(x)-V(x)=l-2\pi i \int_x^\be \rho(t)dt 
\end{equation}
for $x\in (\al'-\ep,\al')$.  These equations imply that
\begin{equation}\label{pmi19a}
\big(2g_\pm(x)-V(x)\big)-\big(2g_\pm(\al')-V(\al')\big)=\pm 2\pi i \int_x^{\al'} \rho(t)dt. 
\end{equation}
We can therefore write (\ref{pmi18}) as
\begin{equation}\label{pmi19b}
\left[\frac{2}{3}\psi_{\al'}(x)^{3/2}\right]_{\pm}=\mp\frac{\pi i}{2\ga}(\al'-x)+\frac{1}{2}\bigg[\big(2g_{\pm}(x)-V(x)\big)-\big(2g_{\pm}(\al')-V(\al')\big)\bigg].
\end{equation}
We can extend these equations into the upper and lower half-plane, respectively, obtaining
\begin{equation}\label{pmi20}
\frac{2}{3}\psi_{\al'}(z)^{3/2}=
\mp \frac{\pi i}{2\ga}(\al'-z)+\frac{1}{2}\bigg[\big(2g(z)-V(z)\big)-\big(2g_\pm (\al')-V(\al')\big)\bigg] \quad \textrm{for} \quad \pm \Im z>0\,.
\end{equation}
Using (\ref{pmi19}) at $x=\al'$, we can write
\begin{equation}\label{pmi21}
\frac{2}{3}\psi_{\al'}(z)^{3/2}=
\mp \frac{\pi i}{2\ga}(\al'-z)+g(z)-\frac{V(z)}{2}-\frac{l}{2} \mp \pi i \int_{\al'}^\be \rho(t)dt 
\quad \textrm{for} \quad \pm \Im z>0\,,
\end{equation}
or equivalently,
\begin{equation}\label{pmi22}
\frac{2}{3}\psi_{\al'}(z)^{3/2}=
g(z)-\frac{V(z)}{2}-\frac{l}{2}\pm \frac{\pi iz}{2\ga} \mp \frac{i(\Omega_n-\pi)}{2n} \quad \textrm{for} \quad \pm \Im z>0\,.
\end{equation}
Plugging (\ref{pmi17}) and (\ref{pmi21}) into (\ref{pmi13}) gives
\begin{equation}\label{pmi22a}
\begin{aligned}
\bold U_{n}(z)&=\bold M(z)e^{\pm \frac{i\Omega_n}{2}\sigma_3}\bold L_n^{\al'}(z)^{-1}\frac{1}{2\sqrt{\pi}}
n^{-\frac{1}{6}\sigma_3}\psi_{\al'}(z)^{-\frac{1}{4}\sigma_3}\\
&\quad \times \bigg[\mp \begin{pmatrix} i & 1 \\ i & -1 \end{pmatrix} \pm \frac{\psi_{\al'}(z)^{-3/2}}{48n}
\begin{pmatrix} -5i & 5 \\ 7i & 7 \end{pmatrix}+O(n^{-2})\bigg] \\
&\quad \times e^{n(g(z)-\frac{l}{2}-\frac{V(z)}{2})\sigma_3}e^{\mp \frac{i\Omega_n}{2}\sigma_3}e^{\pm \frac{i\pi}{2}\sigma_3}
e^{\pm \frac{in\pi z}{2\ga}\sigma_3}e^{\mp\frac{in\pi z}{2\ga}\sigma_3}e^{-n(g(z)-\frac{V(z)}{2}-\frac{l}{2})\sigma_3} \\
&=\bold M(z)e^{\pm \frac{i\Omega_n}{2}\sigma_3}\bold L_n^{\al'}(z)^{-1}\frac{1}{2\sqrt{\pi}}
n^{-\frac{1}{6}\sigma_3}\psi_{\al'}(z)^{-\frac{1}{4}\sigma_3} \\
& \quad \times \ \left[\begin{pmatrix} 1 & i \\ 1 & -i \end{pmatrix}+\frac{\psi_{\al'}(z)^{-3/2}}{48n}
\begin{pmatrix} 5 & -5i \\ -7 & -7i \end{pmatrix}+O(n^{-2})\right]e^{\mp \frac{i\Omega_n}{2}\sigma_3} \\
&=\bold M(z)\left[I+\frac{\psi_{\al'}(z)^{-3/2}}{48n}e^{\pm i\frac{\Omega_n}{2}\sigma_3}
\begin{pmatrix} -1 & -6i \\ -6i & 1 \end{pmatrix}e^{\mp i\frac{\Omega_n}{2}\sigma_3}+O(n^{-2})\right] \\
&=\bold M(z)\left[I+\frac{\psi_{\al'}(z)^{-3/2}}{48n}
\begin{pmatrix} -1 & -6ie^{\pm i\Omega_n} \\ -6ie^{\mp i\Omega_n} & 1 \end{pmatrix}+O(n^{-2})\right] 
\quad \textrm{for} \ \pm \Im (z) >0.
\end{aligned}
\end{equation}

We can make a similar construction near $\be'$.  Let
\begin{equation}\label{pmi24}
\psi_{\be'}(z)=-\left\{\frac{3\pi}{2}\int_{\be'}^z\left(\frac{1}{2\ga}-\rho(t)dt\right)\right\}^{2/3}.
\end{equation}
This function is analytic in $D(\be',\ep)$ and has negative derivative at $\be'$, thus $\Im z$ and $\Im \psi_{\be'}(z)$ have opposite signs for $z \in D(\be',\ep)$.  
Then the jumps for $\tilde{\bold Q}_n$ are
\begin{equation}\label{pmi24a}
j_{\tilde{Q}}(z)=\left\{
\begin{aligned}
&\begin{pmatrix} 0 & 1 \\ -1 & 0 \end{pmatrix} \quad &\textrm{for}\quad z \in (\be',\be'+\ep), \\
&\begin{pmatrix} -1 & 0 \\ -1 & -1 \end{pmatrix}\quad & \textrm{for}\quad z \in (\be'-\ep,\be'), \\
&\begin{pmatrix} 1 & -1 \\ 0 & 1 \end{pmatrix} \quad & \textrm{for}\quad z \in (\be',\be'+i\ep), \\
&\begin{pmatrix} 1 & 1 \\ 0 & 1 \end{pmatrix}  \quad & \textrm{for}\quad z \in (\be',\be'-i \ep),\\
\end{aligned}\right.
\end{equation}
where the contour is oriented from left to right on horizontal segments and up to down on vertical segments according to Figure \ref{sigma_S}.  After a slight deformation of the vertical contours and the change of variables $\psi_{\be'}$, these jumps become the following jumps close to the origin:
\begin{equation}\label{pmi24b}
j_{\tilde{Q}}(\psi_{\be'}(z))=\left\{
\begin{aligned}
&\begin{pmatrix} 0 & 1 \\ -1 & 0 \end{pmatrix} \quad &\textrm{for}\quad \psi_{\be'}(z) \in (-\ep,0), \\
&\begin{pmatrix} -1 & 0 \\ -1 & -1 \end{pmatrix}\quad & \textrm{for}\quad \psi_{\be'}(z) \in (0,\ep), \\
&\begin{pmatrix} 1 & -1 \\ 0 & 1 \end{pmatrix} \quad & \textrm{for}\quad \psi_{\be'}(z) \in (-i \ep,0), \\
&\begin{pmatrix} 1 & 1 \\ 0 & 1 \end{pmatrix}  \quad & \textrm{for}\quad \psi_{\be'}(z) \in (0,i \ep), \\
\end{aligned}\right.
\end{equation}
where the contour is oriented from right to left on horizontal segments and down to up on vertical segments.  These jump conditions are satisfied by the function 
\begin{equation}\label{pmi25c}
\Phi_{\be'}(z)=\Phi_{\al'}(z)\begin{pmatrix} 1 & 0 \\ 0 & -1 \end{pmatrix}.
\end{equation}
Then we can take for $z\in D(\be',\ep)$, 
\begin{equation}\label{pmi25}
\bold U_{n}(z)=
\begin{aligned}
&\bold M(z)e^{\frac{\pm i\Omega_n}{2}\sigma_3}\bold L_n^{\be'}(z)^{-1}\Phi_{\be'}(n^{2/3}\psi_{\be'}(z))e^{\mp \frac{in\pi z}{2\ga}\sigma_3}e^{-n(g(z)-\frac{V(z)}{2}-\frac{l}{2})\sigma_3} \\
&\hspace{5cm}\textrm{for} \quad \pm \Im z > 0,\\
\end{aligned}
\end{equation}
where
\begin{equation}\label{pmi26}
\bold L_n^{\be'}(z)=\frac{1}{2\sqrt{\pi}}
\begin{pmatrix} n^{-1/6} \psi_{\be'}^{-1/4}(z) & 0 \\ 0 & n^{1/6} \psi_{\be'}^{1/4}(z)\end{pmatrix}
\begin{pmatrix} -1 & i \\ -1 & -i \end{pmatrix}.
\end{equation}
We once again have 
\begin{equation}\label{pmi27}
\begin{aligned}
\Phi_{\be'}\big(n^{2/3}\psi_{\be'}(z)\big)&=\frac{1}{2\sqrt{\pi}}n^{-\frac{1}{6}\sigma_3}\psi_{\be'}(z)^{-\frac{1}{4}\sigma_3}
\bigg[\mp \begin{pmatrix} i & -1 \\ i & 1 \end{pmatrix} \mp \frac{\psi_{\be'}(z)^{-3/2}}{48n}
\begin{pmatrix} 5i & 5 \\ -7i & 7 \end{pmatrix}\\
&\quad +O(n^{-2})\bigg]
e^{\frac{2}{3}n\psi_{\be'}(z)^{3/2}\sigma_3} 
\hspace{5mm} \textrm{for} \; \pm \Im \psi_{\be'}(z)>0 \; (\textrm{so} \; \pm \Im z < 0), 
\end{aligned}
\end{equation}
and for $z \in D(\be',\ep)$, 
\begin{equation}\label{pmi28}
\frac{2}{3} \psi_{\be'}^{3/2}(z)=
\pm \frac{\pi i z}{2\ga} +g(z)-\frac{V(z)}{2}-\frac{l}{2} \mp \frac{i(\Omega_n-\pi)}{2n} \quad \textrm{for} \quad \pm \Im z >0.
\end{equation}
Combining (\ref{pmi25}), (\ref{pmi27}), and (\ref{pmi28}) gives
\begin{equation}\label{pmi29}
\begin{aligned}
\bold U_{n}(z)&=\bold M(z)e^{\pm \frac{i\Omega_n}{2}\sigma_3}
\bold L_n^{\be'}(z)^{-1}\frac{1}{2\sqrt{\pi}}n^{-\frac{1}{6}\sigma_3}
\psi_{\be'}(z)^{-\frac{1}{4}\sigma_3}\\
&\qquad \times\left[\pm \begin{pmatrix} i & -1 \\ i & 1 \end{pmatrix} \pm \frac{\psi_{\be'}(z)^{-3/2}}{48n}
\begin{pmatrix} 5i & 5 \\ -7i & 7 \end{pmatrix} +O(n^{-2})\right] \\
&\qquad \times e^{\frac{\pm i n \pi z}{2\ga}\sigma_3} e^{n(g(z)-\frac{V(z)}{2}-\frac{l}{2})\sigma_3}
e^{\mp \frac{i\Omega_n}{2}\sigma_3}e^{\pm \frac{i\pi}{2}\sigma_3}e^{\mp \frac{in\pi z}{2\ga}\sigma_3}
e^{-n(g(z)-\frac{V(z)}{2}-\frac{l}{2})\sigma_3} \\
&=\bold M(z)e^{\pm \frac{i\Omega_n}{2}\sigma_3}
\left[I+\frac{\psi_{\be'}(z)^{-3/2}}{48n}\begin{pmatrix} -1 & 6i \\ 6i & 1 \end{pmatrix} +O(n^{-2})\right]
e^{\mp \frac{i\Omega_n}{2}\sigma_3} \\
&=\bold M(z)\left[I+\frac{\psi_{\be'}(z)^{-3/2}}{48n}
\begin{pmatrix} -1 & 6ie^{\pm i\Omega_n} \\ 6ie^{\mp i\Omega_n} & 1 \end{pmatrix} +O(n^{-2})\right] 
\quad \textrm{for} \ \pm \Im z >0. \\
\end{aligned}
\end{equation}

\medskip

\section{The third and final transformation of the RHP}

We now consider the contour $\Sigma_X$, which consists of the circles $\partial D(\al, \ep), \ \partial D(\al', \ep), \ \partial D(\be', \ep),$ and $\partial D(\be, \ep)$, all oriented counterclockwise, together with the parts of 
$\Sigma_S\setminus\big([\al,\al']\cup[\be',\be]\big)$ 
which lie outside of the disks $D(\al, \ep), \ D(\al', \ep), \ D(\be', \ep),$ and $D(\be, \ep)$, see Fig.~10.

\begin{center}
 \begin{figure}[h]\label{sigma_X}
\begin{center}
   \scalebox{0.7}{\includegraphics{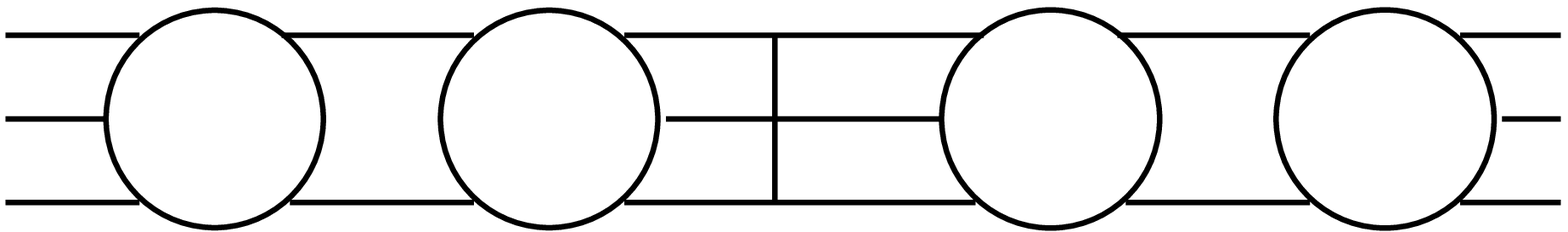}}
\end{center}
        \caption[The contour $\Sigma_X$.]{The contour $\Sigma_X$.}
    \end{figure}
\end{center}

We let
\begin{equation}\label{tt1}
\bold X_n(z)=\left\{
\begin{aligned}
&\bold S_n(z) \bold M(z)^{-1} \quad \textrm{for} \ z \ \textrm{outside the disks }  D(\al, \ep), \ D(\al', \ep), \ D(\be', \ep), \ D(\be, \ep), \\
&\bold S_n(z) \bold U_n(z)^{-1} \quad \textrm{for} \ z \ \textrm{inside the disks }  D(\al, \ep), \ D(\al', \ep), \ D(\be', \ep), \ D(\be, \ep). \\
\end{aligned}\right.
\end{equation}
Then $\bold X_n(z)$ solves the following RHP:
\begin{enumerate}
\item   
$\bold X_n(z)$ is analytic on $\C \setminus \Sigma_X$.
\item
$\bold X_n(z)$ has the jump properties
\begin{equation}\label{tt2}
\bold X_{n+}(x)=\bold X_{n-}(z)j_X(z)
\end{equation}
where
\begin{equation}\label{tt3}
j_X(z)=\left\{
\begin{aligned}
&\bold M(z)\bold U_n(z)^{-1} \quad \textrm{for} \ z \ \textrm{on the circles} \\
&\bold M(z)j_S\bold M(z)^{-1} \quad \textrm{otherwise}.
\end{aligned}\right.
\end{equation}
\item
As $z\to\infty$, 
\begin{equation}\label{tt4}
\bold X_n(z)\sim \I+\frac {\bold X_1}{z}+\frac {\bold X_2}{z^2}+\ldots
\end{equation}
\end{enumerate}
Additionally, we have that $j_X(z)$ is uniformly close to the identity in the following sense:
\begin{equation}\label{tt5}
j_X(z)=\left\{
\begin{aligned}
&I+O(n^{-1}) \quad \textrm{uniformly on the circles} \\
&I+O(e^{-C(z)n}) \quad \textrm{on the rest of} \ \Sigma_X, 
\end{aligned}\right.
\end{equation}
where $C(z)$ is a positive, continuous function satisfying (\ref{ft7}).
If we set
\begin{equation}\label{tt6}
j_X^0(z)=j_X(z)-I,
\end{equation}
then (\ref{tt5}) becomes
\begin{equation}\label{tt7}
j_X^0(z)=\left\{
\begin{aligned}
&O(n^{-1}) \quad \textrm{uniformly on the circles} \\
&O(e^{-C(z)n}) \quad \textrm{on the rest of} \ \Sigma_X.
\end{aligned}\right.
\end{equation}
The solution to the RHP for $\bold X_n$ is based on the following lemma:
\begin{lem}
Suppose v(z) is a function on $\Sigma_X$ solving the equation
\begin{equation}\label{tt8}
v(z)=I-\frac{1}{2\pi i} \int_{\Sigma_X} \frac{v(u)j_X^0(u)}{z_--u}du \quad \textrm{for} \ z\in\Sigma_X
\end{equation}
where $z_-$ means the value of the integral on the minus side of $\Sigma_X$.  Then 
\begin{equation}\label{tt9}
\bold X_n(z)=I-\frac{1}{2\pi i} \int_{\Sigma_X} \frac{v(u)j_X^0(u)}{z-u}du \quad \textrm{for} \ z\in\C\setminus\Sigma_X
\end{equation}
solves the RHP for $\bold X_n$.
\end{lem}
The proof of this lemma is immediate from the jump property of the Cauchy transform.  By assumption
\begin{equation}\label{tt12}
\bold X_{n-}(z)=v(z) 
\end{equation}
and the additive jump of the Cauchy transform gives
\begin{equation}\label{tt13}
\bold X_{n+}(z)-\bold X_{n-}(z)=v(z)j_X^0(z)=\bold X_{n-}(z)j_X^0(z),
\end{equation}
thus $\bold X_{n+}(z)=\bold X_{n-}(z)j_X(z)$.  Asymptotics at infinity are given by (\ref{tt9}).

The solution to equation (\ref{tt8}) is given by a series of perturbation theory.  Namely, the solution is
\begin{equation}\label{tt14}
v(z)=I+\sum_{k=1}^\infty v_k(z)
\end{equation}
where 
\begin{equation}\label{tt15}
v_k(z)=-\frac{1}{2\pi i} \int_{\Sigma_X} \frac{v_{k-1}(u)j_X^0(u)}{z-u}du\,, \quad v_0(z)=I.
\end{equation}
This function clearly solves (\ref{tt8}) provided the series converges, which is does, for sufficiently large $n$.  Indeed, by (\ref{tt5}), 
\begin{equation}\label{tt16}
|v_k(z)| \leq \left(\frac{C}{n}\right)^k\frac{1}{1+|z|} \quad \textrm{for some constant} \ C>0
\end{equation}
thus the series (\ref{tt14}) is dominated by a convergent geometric series and thus converges absolutely.  This in turn gives
\begin{equation}\label{tt17}
\bold X_n(z)=I+\sum_{k=1}^\infty\bold X_{n,k}(z)
\end{equation}
where
\begin{equation}\label{tt18}
\bold X_{n,k}(z)=-\frac{1}{2\pi i}\int_{\Sigma_X}\frac{v_{k-1}(u)j_X^0(u)}{z-u}du.
\end{equation}
We will need to compute
\begin{equation}\label{tt19}
\bold X_{n,1}(z)=-\frac{1}{2\pi i}\int_{\Sigma_X}\frac{j_X^0(u)}{z-u}du.
\end{equation}

\medskip

\section{Evaluation of $\bold X_{1}$}

We are interested in the matrix $\bold X_1$, which gives the $\frac{1}{z}$-term of $\bold X_n(z)$
at infinity, see (\ref{tt4}). By (\ref{tt9}),
\begin{equation}\label{eval1}
\bold X_1 =-\frac{1}{2\pi i} \int_{\Sigma_X} v(u)j_X^0(u)\,du, 
\end{equation}
hence by (\ref{tt14}), (\ref{tt16}),
\begin{equation}\label{eval1a}
\bold X_1 =-\frac{1}{2\pi i} \int_{\Sigma_X} j_X^0(u)\,du +O(n^{-2}). 
\end{equation}
We would like to evaluate the integral,
\begin{equation}\label{eval2}
-\frac{1}{2\pi i}\int_{\Sigma_X}j_X^0(u)du
\end{equation}
with an error of the order of $n^{-2}$. By (\ref {tt7}), it is enough to 
evaluate this integral over the circles $\partial D(\al,\ep)$, $\partial D(\al',\ep)$, 
$\partial D(\be',\ep)$, and $\partial D(\be,\ep)$. As we will see in the next
section, the matrix-valued function
$j_X^0(z)$ is analytic in the punctured disks, hence
\begin{equation}\label{eval2a}
\bold X_1 =-\left(\underset{z=\al}{\Res}+\underset{z=\al'}{\Res}+\underset{z=\be'}{\Res}
+\underset{z=\be}{\Res}\right)j_X^0(z)+O(n^{-2}).
\end{equation}
We will be especially interested
in evaluation of the [12] element of the matrix $\bold X_1$, and we will prove the following
asymptotic formula. Introduce the numbers,
\begin{equation}\label{eval12f}
\begin{aligned}
\eta_\al&=\left[5\frac{\th_4^{''}(n\om+\frac{\om}{2})}{\th_4(n\om+\frac{\om}{2})}
-5\frac{\th_3^{''}(\frac{\om}{2})}{\th_3(\frac{\om}{2})}
+7\left(\frac{\th_4^{'}(n\om+\frac{\om}{2})}{\th_4(n\om+\frac{\om}{2})}\right)^2
+17\left(\frac{\th_3^{'}(\frac{\om}{2})}{\th_3(\frac{\om}{2})}\right)^2
-24\frac{\th_4^{'}(n\om+\frac{\om}{2})\th_3^{'}(\frac{\om}{2})}{\th_4(n\om+\frac{\om}{2})\th_3(\frac{\om}{2})}\right]\,, \\
\eta_{\al'}&=-\left[5\frac{\th_1^{''}(n\om+\frac{\om}{2})}{\th_1(n\om+\frac{\om}{2})}
-5\frac{\th_2^{''}(\frac{\om}{2})}{\th_2(\frac{\om}{2})}
+7\left(\frac{\th_1^{'}(n\om+\frac{\om}{2})}{\th_1(n\om+\frac{\om}{2})}\right)^2
+17\left(\frac{\th_2^{'}(\frac{\om}{2})}{\th_2(\frac{\om}{2})}\right)^2
-24\frac{\th_1^{'}(n\om+\frac{\om}{2})\th_2^{'}(\frac{\om}{2})}{\th_1(n\om+\frac{\om}{2})\th_2(\frac{\om}{2})}\right]\,, \\
\eta_{\be'}&=-\left[5\frac{\th_2^{''}(n\om+\frac{\om}{2})}{\th_2(n\om+\frac{\om}{2})}
-5\frac{\th_1^{''}(\frac{\om}{2})}{\th_1(\frac{\om}{2})}
+7\left(\frac{\th_2^{'}(n\om+\frac{\om}{2})}{\th_2(n\om+\frac{\om}{2})}\right)^2
+17\left(\frac{\th_1^{'}(\frac{\om}{2})}{\th_1(\frac{\om}{2})}\right)^2
-24\frac{\th_2^{'}(n\om+\frac{\om}{2})\th_1^{'}(\frac{\om}{2})}{\th_2(n\om+\frac{\om}{2})\th_1(\frac{\om}{2})}\right]\,,\\
\eta_\be&=\left[5\frac{\th_3^{''}(n\om+\frac{\om}{2})}{\th_3(n\om+\frac{\om}{2})}
-5\frac{\th_4^{''}(\frac{\om}{2})}{\th_4(\frac{\om}{2})}
+7\left(\frac{\th_3^{'}(n\om+\frac{\om}{2})}{\th_3(n\om+\frac{\om}{2})}\right)^2
+17\left(\frac{\th_4^{'}(\frac{\om}{2})}{\th_4(\frac{\om}{2})}\right)^2
-24\frac{\th_3^{'}(n\om+\frac{\om}{2})\th_4^{'}(\frac{\om}{2})}{\th_3(n\om+\frac{\om}{2})\th_4(\frac{\om}{2})}\right]\,,\\
\end{aligned}
\end{equation}
and
\begin{equation}\label{eval12i}
\begin{aligned}
C_\al&=\frac{7}{2}\,(\be'-\al)+\frac{3}{2}\,(\be-\al)+\frac{3}{2}\,(\al'-\al)
-\frac{(\al'-\al)(\be-\al)}{(\be'-\al)}, \\
C_{\al'}&=-\frac{7}{2}\,(\be-\al')-\frac{3}{2}\,(\be'-\al')+\frac{3}{2}\,(\al'-\al)
-\frac{(\al'-\al)(\be'-\al')}{(\be-\al')}, \\
C_{\be'}&=-\frac{7}{2}\,(\be'-\al)-\frac{3}{2}(\be'-\al')+\frac{3}{2}\,(\be-\be')
-\frac{(\be-\be')(\be'-\al')}{(\be'-\al)}, \\
C_\be&=\frac{7}{2}\,(\be-\al')+\frac{3}{2}\,(\be-\al)+\frac{3}{2}\,(\be-\be')
-\frac{(\be-\be')(\be-\al)}{(\be-\al')}\,.
\end{aligned}
\end{equation}
Introduce also the numbers,
\begin{equation}\label{eval12b}
\begin{aligned}
\Xi_\al&=\frac{\th_3^2(0)\th_4^2(n\om+\frac{\om}{2})}{\th_3^2(\frac{\om}{2})\th_4^2(n\om)}\,, \qquad
\Xi_{\al'}=\frac{\th_3^2(0)\th_1^2(n\om+\frac{\om}{2})}{\th_2^2(\frac{\om}{2})\th_4^2(n\om)}\,, \\
\Xi_{\be'}&=\frac{\th_3^2(0)\th_2^2(n\om+\frac{\om}{2})}{\th_1^2(\frac{\om}{2})\th_4^2(n\om)}\,,\qquad
\Xi_\be=\frac{\th_3^2(0)\th_3^2(n\om+\frac{\om}{2})}{\th_4^2(\frac{\om}{2})\th_4^2(n\om)}\,,
\end{aligned}
\end{equation}
and
\begin{equation}\label{eval12e}
\begin{aligned}
\xi_\al&=\frac{\th'_3(\frac{\om}{2})}{\th_3(\frac{\om}{2})}-\frac{\th'_4(n\om+\frac{\om}{2})}{\th_4(n\om+\frac{\om}{2})}
\,,\qquad
\xi_{\al'}=-\left(\frac{\th'_2(\frac{\om}{2})}{\th_2(\frac{\om}{2})}-\frac{\th'_1(n\om+\frac{\om}{2})}{\th_1(n\om+\frac{\om}{2})}\right)
\,,\\
\xi_{\be'}&=\frac{\th'_1(\frac{\om}{2})}{\th_1(\frac{\om}{2})}-\frac{\th'_2(n\om+\frac{\om}{2})}{\th_2(n\om+\frac{\om}{2})}\,,\qquad
\xi_\be=-\left(\frac{\th'_4(\frac{\om}{2})}{\th_4(\frac{\om}{2})}-\frac{\th'_3(n\om+\frac{\om}{2})}{\th_3(n\om+\frac{\om}{2})}\right)
\,.
\end{aligned}
\end{equation}

\begin{lem} \label{X1} As $n\to\infty$,
\begin{equation}\label{x_1}
[\bold X_1]_{12} =\frac{1}{n}\left(X_{\al}+X_{\al'}+X_{\be'}+X_{\be}\right)+O(n^{-2}),
\end{equation}
where
\begin{equation}\label{x_2}
\begin{aligned}
X_\al&=\frac{i\Xi_\al}{96}
 \left(C_\al+12\pi \xi_\al+\frac{\pi^2 \eta_\al}{2(\be'-\al)}\right)\,,\\
X_{\al'}&=\frac{i\Xi_\al'}{96}
\left(C_{\al'}+12\pi \xi_{\al'}+\frac{\pi^2\eta_{\al'}}{2(\be-\al')}\right)\,,\\
X_{\be'}&=\frac{i\Xi_{\be'}}{96}
\left(C_{\be'}+12\pi \xi_{\be'}+\frac{\pi^2\eta_{\be'}}{2({\be'}-\al)}\right)\,,\\
X_\be&=\frac{i\Xi_\be}{96}
\left(C_\be+12 \pi\xi_\be+\frac{\pi^2\eta_\be}{2(\be-\al')}\right)\,.
\end{aligned}
\end{equation}
\end{lem} 

Proof of this lemma is given in the next section.

\section{Proof of Lemma \ref{X1}}\label{proof_X1}

On the circles, $\partial D(\al,\ep)$ 
and $\partial D(\be,\ep)$, we have
\begin{equation}\label{eval3}
j_X(z)=\left\{
\begin{aligned}
&= I - \frac{\psi_\al(z)^{-3/2}}{48n}\bold M(z)\begin{pmatrix} 1 & -6i \\  -6i & -1 \end{pmatrix} 
\bold M^{-1}(z)+O(n^{-2}) \quad  \textrm{for} \quad z \in \partial D(\al,\ep), \\
&=I - \frac{\psi_\be(z)^{-3/2}}{48n}\bold M(z)\begin{pmatrix} 1 & 6i \\  6i & -1 \end{pmatrix} 
\bold M^{-1}(z)+O(n^{-2}) \quad \textrm{for} \quad z \in \partial D(\be,\ep), \\
\end{aligned}\right.
\end{equation} 
and on the circles, $\partial D(\al',\ep)$ and $\partial D(\be',\ep)$, we have
\begin{equation}\label{eval4}
j_X(z)=\left\{
\begin{aligned}
&I - \frac{\psi_{\al'}(z)^{-3/2}}{48n}\bold M(z)\begin{pmatrix} -1 & -6ie^{\pm i\Omega_n} \\  -6ie^{\mp i\Omega_n} & 1 \end{pmatrix}\bold M^{-1}(z)+O(n^{-2})  \\
&\hspace{8 cm} \textrm{for} \ z\in\partial D(\al',\ep),  \ \pm \Im z >0, \\
&I - \frac{\psi_{\be'}(z)^{-3/2}}{48n}\bold M(z)\begin{pmatrix} -1 & 6ie^{\pm i\Omega_n} \\  6ie^{\mp i\Omega_n} & 1 \end{pmatrix}\bold M^{-1}(z)+O(n^{-2}) \quad \\
&\hspace{8 cm}  \textrm{for} \ z\in\partial D(\be',\ep),  \ \pm \Im z >0. \\
\end{aligned}\right.
\end{equation}
Thus
\begin{equation}\label{eval5}
j_X^0(z)=\left\{
\begin{aligned}
 - \frac{\psi_\al(z)^{-3/2}}{48n}\bold M(z)
&\begin{pmatrix} 1 & -6i \\  -6i & -1 \end{pmatrix}\bold M^{-1}(z)+O(n^{-2}) \ \ \textrm{for} \ z \in \partial D(\al, \ep),  \\
 - \frac{\psi_\be(z)^{-3/2}}{48n}\bold M(z)
&\begin{pmatrix} 1 & 6i \\  6i & -1 \end{pmatrix}\bold M^{-1}(z)+O(n^{-2}) \ \ \textrm{for} \ z \in \partial D(\be, \ep),  \\
  - \frac{\psi_{\al'}(z)^{-3/2}}{48n}\bold M(z)
&\begin{pmatrix} -1 & -6ie^{\pm i\Omega_n} \\  -6ie^{\mp i\Omega_n} & 1 \end{pmatrix}\bold M^{-1}(z)+O(n^{-2}) \ \textrm{for} \; z \in \partial D(\al', \ep), \\
 - \frac{\psi_{\be'}(z)^{-3/2}}{48n}\bold M(z)
&\begin{pmatrix} -1 & 6ie^{\pm i\Omega_n} \\  6ie^{\mp i\Omega_n} & 1 \end{pmatrix}\bold M^{-1}(z)+O(n^{-2}) \ \textrm{for} \; z \in \partial D(\be', \ep), \\
\end{aligned}\right.
\end{equation}
for $\pm \Im z >0$.
To simplify notation, we will write the model solution given in (\ref{m21}) and (\ref{m22}) as
\begin{equation}\label{eval6}
\bold M(z)=\frac{1}{2}\begin{pmatrix} \big(\ga(z)+\ga^{-1}(z)\big)\th_{11}(z) & i\big(\ga(z)-\ga^{-1}(z)\big)\th_{12}(z) \\ - i\big(\ga(z)-\ga^{-1}(z)\big)\th_{21}(z) & \big(\ga(z)+\ga^{-1}(z)\big)\th_{22}(z)\end{pmatrix}
\end{equation}
where
\begin{equation}\label{eval7}
\begin{aligned}
\th_{11}(z)&=\frac{\th_3(\tilde{u}(z)-\tilde{u}_\infty+\frac{\Omega_n}{2})\th_3(0)}{\th_3(\tilde{u}(z)-\tilde{u}_\infty)\th_3(\frac{\Omega_n}{2})} \quad  \quad \th_{12}(z)=\frac{\th_3(\tilde{u}(z)+\tilde{u}_\infty-\frac{\Omega_n}{2})\th_3(0)}{\th_3(\tilde{u}(z)+\tilde{u}_\infty)\th_3(\frac{\Omega_n}{2})} \\
\th_{21}(z)&=\frac{\th_3(\tilde{u}(z)+\tilde{u}_\infty+\frac{\Omega_n}{2})\th_3(0)}{\th_3(\tilde{u}(z)+\tilde{u}_\infty)\th_3(\frac{\Omega_n}{2})} \quad \quad \th_{22}(z)=\frac{\th_3(\tilde{u}(z)-\tilde{u}_\infty-\frac{\Omega_n}{2})\th_3(0)}{\th_3(\tilde{u}(z)-\tilde{u}_\infty)\th_3(\frac{\Omega_n}{2})}.
\end{aligned}
\end{equation}
Notice that each of the functions $\th_{ij}$ is analytic throughout the complex plane, except on the intervals $(\al, \al')$ and $(\be',\be)$, where they satisfy the relations
\begin{equation}\label{eval8}
[\th_{11}]_\pm=[\th_{12}]_\mp \,, \quad [\th_{21}]_\pm=[\th_{22}]_\mp,
\end{equation}
and on the interval $(\al',\be')$, where they satisfy
\begin{equation}\label{eval9}
[\th_{11}]_+=e^{-i\Omega_n}[\th_{11}]_- \ , \  [\th_{12}]_+=e^{i\Omega_n}[\th_{12}]_- \  , \ [\th_{21}]_+=e^{-i\Omega_n}[\th_{21}]_- \ , \ [\th_{22}]_+=e^{i\Omega_n}[\th_{22}]_- .
\end{equation}
Multiplying out equations (\ref{eval5}) gives
\begin{equation}\label{eval10}
j_X^0(z)= \frac{\psi_\xi(z)^{-3/2}}{48n}\begin{pmatrix}j_{11}^\xi & j_{12}^\xi \\ j_{21}^\xi &  j_{22}^\xi \end{pmatrix}+O(n^{-2}) \quad \textrm{for} \ z \in \partial D(\xi,\ep), \quad \xi=\al,\al',\be',\be,
\end{equation}
where
\begin{equation}\label{eval11}
\begin{aligned}
j_{12}^\al&=\frac{i}{2}\left[3\big((\ga^2(z)+\ga^{-2}(z)\big)(\th_{11}^2-\th_{12}^2)+\big(\ga^2(z)-\ga^{-2}(z)\big)\th_{11}\th_{12}+6(\th_{11}^2-\th_{12}^2)\right], \\
j_{12}^{\al'}&=\frac{i}{2}\left[3\big(\ga^2(z)+\ga^{-2}(z)\big)J(z)-\big(\ga^2(z)-\ga^{-2}(z)\big)\th_{11}\th_{12}+6K(z)\right], \\
j_{12}^{\be'}&=-\frac{i}{2}\left[3\big(\ga^2(z)+\ga^{-2}(z)\big)J(z)+\big(\ga^2(z)-\ga^{-2}(z)\big)\th_{11}\th_{12}+6K(z)\right], \\
j_{12}^\be&=-\frac{i}{2}\left[3(\big(\ga^2(z)+\ga^{-2}(z)\big)(\th_{11}^2-\th_{12}^2)-\big(\ga^2(z)-\ga^{-2}(z)\big)\th_{11}\th_{12}+6(\th_{11}^2-\th_{12}^2)\right], \\
\end{aligned}
\end{equation}
and
\begin{equation}\label{eval12}
J(z)=\left\{
\begin{aligned}
&\th_{11}^2e^{i\Omega_n}+\th_{12}^2e^{-i\Omega_n} \quad \textrm{for} \quad \Im z >0, \\
&\th_{11}^2e^{-i\Omega_n}+\th_{12}^2e^{i\Omega_n} \quad \textrm{for} \quad \Im z <0, \\
\end{aligned}\right.
\end{equation}
\begin{equation}\label{eval13}
K(z)=\left\{
\begin{aligned}
&\th_{11}^2e^{i\Omega_n}-\th_{12}^2e^{-i\Omega_n} \quad \textrm{for} \quad \Im z >0, \\
&\th_{11}^2e^{-i\Omega_n}-\th_{12}^2e^{i\Omega_n} \quad \textrm{for} \quad \Im z <0. \\
\end{aligned}\right.
\end{equation}

In order to integrate $j_X^0(z)$, let us examine the behavior of the various functions described 
above near each of the turning points.
Introduce the numbers
\begin{equation}\label{eval20}
\begin{aligned}
A_\al&=\sqrt{(\al'-\al)(\be'-\al)(\be-\al)} \,  , \qquad  
A_{\al'}=\sqrt{(\al'-\al)(\be'-\al')(\be-\al')} \,  ,  \\
A_{\be'}&=\sqrt{(\be'-\al)(\be'-\al')(\be-\be')} \,  , \qquad
A_\be=\sqrt{(\be-\al)(\be-\al')(\be-\be')}\,,
\end{aligned}
\end{equation}
and
\begin{equation}\label{eval12a}
\begin{aligned}
 B_\al&=\frac{1}{\al'-\al}+\frac{1}{\be'-\al}+\frac{1}{\be-\al}\,,\qquad 
B_{\al'}=-\frac{1}{\al'-\al}+\frac{1}{\be'-\al'}+\frac{1}{\be-\al'}\,, \\
 B_{\be'}&=\frac{1}{\be'-\al}+\frac{1}{\be'-\al'}-\frac{1}{\be-\be'}\,,\qquad
B_\be=\frac{1}{\be-\al}+\frac{1}{\be-\al'}+\frac{1}{\be-\be'}\,.
\end{aligned}
\end{equation}
For $x\in (\be,\be+\ep)$, we have
\begin{equation}\label{eval14}
\tilde{u}(x)=\frac{\pi}{2K}u(x)=\frac{\pi}{A_\be}\sqrt{x-\be}+O((x-\be)^{3/2});
\end{equation}
for $x\in (\al-\ep,\al)$,
\begin{equation}\label{eval15}
\tilde{u}(x)=\frac{\pi}{2K}u(x)=\frac{\pi}{2}-\frac{\pi}{A_\al}\sqrt{\al-x}+O((\al-x)^{3/2});
\end{equation}
for $x \in (\al', \al'+\ep)$,
\begin{equation}\label{eval16}
\tilde{u}_{\pm}(x)=\frac{\pi}{2K}u_{\pm}(x)=\frac{\pi}{2} \pm \frac{\tau\pi}{2}
-\frac{\pi}{A_{\al'}}\sqrt{x-\al'}+O((x-\al')^{3/2});
\end{equation}
and for $x \in (\be'-\ep,\be')$,
\begin{equation}\label{eval17}
\tilde{u}_{\pm}(x)=\frac{\pi}{2K}u_{\pm}(x)= \pm \frac{\tau \pi}{2}
+\frac{\pi}{A_{\be'}}\sqrt{\be'-x}+O((\be'-x)^{-3/2}).
\end{equation}
Also, from (\ref{m2b}) we have that
\begin{equation}\label{eval18}
\ga^2(z)\pm \ga^{-2}(z)=\sqrt{\frac{(z-\al)(z-\be')}{(z-\al')(z-\be)}}\pm
\sqrt{\frac{(z-\al')(z-\be)}{(z-\al)(z-\be')}}\,,
\end{equation}
and from (\ref{pm14}), (\ref{pm21}), (\ref{pmi8}), (\ref{pmi24}) , and (\ref{edf3}) that 
\begin{equation}\label{eval19}
\begin{aligned}
\psi_\al^{-3/2}(z)&=(\al-z)^{-3/2}\left[\frac{A_\al}{2}+\frac{1}{20}A_{\al}B_\al(\al-z)+O\big((\al-z)^2\big)\right],  \\
\psi_{\al'}^{-3/2}(z)&=(z-\al')^{-3/2}\left[\frac{A_{\al'}}{2}-\frac{1}{20}A_{\al'}B_{\al'}(z-\al')+O\big((z-\al'\big)^2)\right],  \\
\psi_{\be'}^{-3/2}(z)&=(\be'-z)^{-3/2}\left[\frac{A_{\be'}}{2}-\frac{1}{20}A_{\be'}B_{\be'}(\be'-z)+O\big((\be'-z)^2\big)\right],  \\
\psi_\be^{-3/2}(z)&=(z-\be)^{-3/2}\left[\frac{A_\be}{2}+\frac{1}{20}A_{\be}B_\be(z-\be)+O\big((z-\be)^2\big)\right].
\end{aligned}
\end{equation}
It follows that the functions $(\ga^2\pm\ga^{-2})\psi^{-3/2}(z)$ are meromorphic in a neighborhood of each of the turning points.  
In particular, at $z=\al$, we have
\begin{equation}\label{eval19a}\begin{aligned}
(\ga^2\pm\ga^{-2})\psi_\al^{-3/2}(z)&=\pm \frac{(\al'-\al)(\be-\al)}{2(\al-z)^2}
+\frac{1}{(\al-z)}\left[\frac{(\be'-\al)}{2} \pm \frac{3}{10}(\be-\al) \right. \\
&\left. \quad\pm \frac{3}{10}(\al'-\al) 
  \mp \frac{1}{5}\frac{(\al'-\al)(\be-\al)}{(\be'-\al)} \right]+O(1) ;        
\end{aligned}
\end{equation}
at $z=\be$, we have
\begin{equation}\label{eval19b}\begin{aligned}
(\ga^2\pm\ga^{-2})\psi_\be^{-3/2}(z)&=\frac{(\be-\al)(\be-\be')}{2(z-\be)^2}
+\frac{1}{(z-\be)}\left[\pm \frac{1}{2}(\be-\al')+ \frac{3}{10}(\be-\al)  \right. \\
&\left. \quad  + \frac{3}{10}(\be-\be') 
- \frac{1}{5}\frac{(\be-\al)(\be-\be')}{(\be-\al')} \right]+O(1) ;
\end{aligned}\end{equation}
at $z=\al'$, we have
\begin{equation}\label{eval19c}\begin{aligned}
(\ga^2\pm\ga^{-2})\psi_{\al'}^{-3/2}(z)&=\frac{(\al'-\al)(\be'-\al')}{2(z-\al')^2}
+\frac{1}{(z-\al')}\left[\pm \frac{(\be-\al')}{2} +\frac{3}{10}(\be'-\al') \right. \\
&\left. \quad  - \frac{3}{10}(\al'-\al) 
+ \frac{1}{5}\frac{(\al'-\al)(\be'-\al')}{(\be-\al')} \right]+O(1) ;
\end{aligned}
\end{equation}
and at $z=\be'$, we have
\begin{equation}\label{eval19d}\begin{aligned}
(\ga^2\pm\ga^{-2})\psi_{\be'}^{-3/2}(z)&=\pm \frac{(\be'-\al')(\be-\be')}{2(\be'-z)^2}
+\frac{1}{(\be'-z)}\left[\frac{(\be'-\al)}{2} \pm \frac{3}{10}(\be'-\al')  \right. \\
&\left. \quad \mp \frac{3}{10}(\be-\be') 
\pm \frac{1}{5}\frac{(\be'-\al')(\be-\be')}{\be'-\al} \right]+O(1) .
\end{aligned}
\end{equation}
Notice also, from the relations (\ref{eval7}), that the functions $\th_{11}^2+\th_{12}^2$ and $\th_{11}\th_{12}$ have no jumps in neighborhoods of $\al$ or $\be$, and take finite values at $z=\al$ and $z=\be$, thus are analytic in neighborhoods of $\al$ and $\be$. 
Using (\ref{eval14}) and (\ref{eval15}), we see that these functions have Taylor expansions about $z=\be$,
\begin{equation}\label{eval21}
\begin{aligned}
 \th_{11}^2(z)+\th_{12}^2(z)&=2\frac{\th_3^2(0)}{\th_3^2(\frac{\Omega_n}{2})}
\frac{\th_3^2(\tilde{u}_\infty-\frac{\Omega_n}{2})}{\th_3^2(\tilde{u}_\infty)}
+\left.\frac{\pi^2}{2A_\be^2}\frac{d^2}{d\tilde{u}^2}\bigg(\th_{11}^2+\th_{12}^2\bigg)\right|_{z=\be}(z-\be)+\cdots,  \\
\th_{11}(z)\th_{12}(z)&=\frac{\th_3^2(0)}{\th_3^2(\frac{\Omega_n}{2})}\frac{\th_3^2(\tilde{u}_\infty-\frac{\Omega_n}{2})}{\th_3^2(\tilde{u}_\infty)}+\left.\frac{\pi^2}{2A_\be^2}\frac{d^2}{d\tilde{u}^2}\bigg(\th_{11}\th_{12}\bigg)\right|_{z=\be}(z-\be)+\cdots, \\
\end{aligned}
\end{equation}
and about $z=\al$,
\begin{equation}\label{eval22}
\begin{aligned}
 \th_{11}^2(z)+\th_{12}^2(z)&=2\frac{\th_3^2(0)}{\th_3^2(\frac{\Omega_n}{2})}\frac{\th_4^2(\tilde{u}_\infty-\frac{\Omega_n}{2})}{\th_4^2(\tilde{u}_\infty)}+\left.\frac{\pi^2}{2A_\al^2}\frac{d^2}{d\tilde{u}^2}\bigg(\th_{11}^2+\th_{12}^2\bigg)\right|_{z=\al}(\al-z)+\cdots,  \\
\th_{11}(z)\th_{12}(z)&=\frac{\th_3^2(0)}{\th_3^2(\frac{\Omega_n}{2})}\frac{\th_4^2(\tilde{u}_\infty-\frac{\Omega_n}{2})}{\th_4^2(\tilde{u}_\infty)}+\left.\frac{\pi^2}{2A_\al^2}\frac{d^2}{d\tilde{u}^2}\bigg(\th_{11}\th_{12}\bigg)\right|_{z=\al}(\al-z)+\cdots. \\
\end{aligned}
\end{equation}

By a similar argument, $J(z)$ and $\th_{11}\th_{12}$ are also analytic in neighborhoods of $\al'$ and $\be'$ and using (\ref{eval16}) and (\ref{eval17}) we can write their Taylor expansions about $z=\al'$,
 \begin{equation}\label{eval23}
 \begin{aligned}
 J(z)&=2\frac{\th_3^2(0)}{\th_3^2(\frac{\Omega_n}{2})}\frac{\th_1^2(\tilde{u}_\infty-\frac{\Omega_n}{2})}{\th_1^2(\tilde{u}_\infty)} +\left.\frac{\pi^2}{2A_{\al'}^2}\frac{d^2}{d\tilde{u}^2}J(z)\right|_{z=\al'}(z-\al')+\cdots, \\
 \th_{11}(z)\th_{12}(z)&=\frac{\th_3^2(0)}{\th_3^2(\frac{\Omega_n}{2})}\frac{\th_1^2(\tilde{u}_\infty-\frac{\Omega_n}{2})}{\th_1^2(\tilde{u}_\infty)} + \left.\frac{\pi^2}{2A_{\al'}^2}\frac{d^2}{d\tilde{u}^2}\bigg(\th_{11}\th_{12}\bigg)\right|_{z=\al'}(z-\al')+\cdots ,\\
 \end{aligned}
 \end{equation}
 and about $z=\be'$,
  \begin{equation}\label{eval24}
 \begin{aligned}
 J(z)&=2\frac{\th_3^2(0)}{\th_3^2(\frac{\Omega_n}{2})}\frac{\th_2^2(\tilde{u}_\infty-\frac{\Omega_n}{2})}{\th_2^2(\tilde{u}_\infty)} +\left.\frac{\pi^2}{2A_{\be'}^2}\frac{d^2}{d\tilde{u}^2}J(z)\right|_{z=\be'}(\be'-z)+\cdots, \\
 \th_{11}(z)\th_{12}(z)&=\frac{\th_3^2(0)}{\th_3^2(\frac{\Omega_n}{2})}\frac{\th_2^2(\tilde{u}_\infty-\frac{\Omega_n}{2})}{\th_2^2(\tilde{u}_\infty)} + \left.\frac{\pi^2}{2A_{\be'}^2}\frac{d^2}{d\tilde{u}^2}\bigg(\th_{11}\th_{12}\bigg)\right|_{z=\be'}(\be'-z)+\cdots. \\
 \end{aligned}
 \end{equation}
Finally, notice that the function $\th_{11}^2-\th_{12}^2$ is an odd function of $\tilde{u}$, and using (\ref{eval14}) and (\ref{eval15}), we can write, for $x\in (\be,\be+\ep)$,
\begin{equation}\label{eval27}
\begin{aligned}
\th_{11}^2(x)&-\th_{12}^2(x)=\sqrt{x-\be}\\
&\times\left[\frac{4\pi}{A_\be} \frac{\th_3^2(0)}{\th_3^2(\frac{\Omega_n}{2})}
\frac{\th_3^2(\tilde{u}_\infty-\frac{\Omega_n}{2})}{\th_3^2(\tilde{u}_\infty)}
\left(\frac{\th_3'(\tilde{u}_\infty)}{\th_3(\tilde{u}_\infty)}
-\frac{\th_3'(\tilde{u}_\infty-\frac{\Omega_n}{2})}{\th_3(\tilde{u}_\infty-\frac{\Omega_n}{2})}\right)
+O(x-\be)\right]\,,
\end{aligned}
\end{equation}
and for $x\in (\al-\ep,\al)$,
\begin{equation}\label{eval28}
\begin{aligned}
\th_{11}^2(x)&-\th_{12}^2(x)=-\sqrt{\al-x}\\
&\times \left[\frac{4\pi}{A_\al}\frac{\th_3^2(0)}{\th_3^2(\frac{\Omega_n}{2})}
\frac{\th_4^2(\tilde{u}_\infty-\frac{\Omega_n}{2})}{\th_4^2(\tilde{u}_\infty)}
\left(\frac{\th_4'(\tilde{u}_\infty)}{\th_4(\tilde{u}_\infty)}
-\frac{\th_4'(\tilde{u}_\infty-\frac{\Omega_n}{2})}{\th_4(\tilde{u}_\infty-\frac{\Omega_n}{2})}\right)
+O(\al-x)\right]\,. 
\end{aligned}
\end{equation}
Similarly, using (\ref{eval16}) and (\ref{eval17}), we can write,  for $x\in (\al',\al'+\ep)$,
\begin{equation}\label{eval31}
\begin{aligned}
K(x)&=-\sqrt{x-\al'}\\
&\times\left[\frac{4\pi}{A_{\al'}} \frac{\th_3^2(0)}{\th_3^2(\frac{\Omega_n}{2})}
\frac{\th_1^2(\tilde{u}_\infty-\frac{\Omega_n}{2})}{\th_1^2(\tilde{u}_\infty)}
\left(\frac{\th_1'(\tilde{u}_\infty)}{\th_1(\tilde{u}_\infty)}
-\frac{\th_1'(\tilde{u}_\infty-\frac{\Omega_n}{2})}{\th_1'(\tilde{u}_\infty-\frac{\Omega_n}{2})}\right) 
 +O(x-\al'))\right], 
\end{aligned}
\end{equation}
and for $x\in (\be'-\ep,\be')$,
\begin{equation}\label{eval32}
\begin{aligned}
K(x)&=\sqrt{\be'-x}\\
&\times\left[\frac{4\pi}{A_{\be'}} \frac{\th_3^2(0)}{\th_3^2(\frac{\Omega_n}{2})}
\frac{\th_2^2(\tilde{u}_\infty-\frac{\Omega_n}{2})}{\th_2^2(\tilde{u}_\infty)}
\left(\frac{\th_2'(\tilde{u}_\infty)}{\th_2(\tilde{u}_\infty)}
-\frac{\th_2'(\tilde{u}_\infty-\frac{\Omega_n}{2})}{\th_2'(\tilde{u}_\infty-\frac{\Omega_n}{2})}\right)
 +O(\be'-x) \right] \\
\end{aligned}
\end{equation}

From equations (\ref{eval19}), (\ref{eval27}), (\ref{eval28}), (\ref{eval31}), and (\ref{eval32}), it follows that the functions $\left(\th_{11}^2(z)-\th_{12}^2(z)\right)\psi_\al^{-3/2}(z)$ and $\left(\th_{11}^2(z)-\th_{12}^2(z)\right)\psi_\be^{-3/2}(z)$ are meromorphic in neighborhoods of $\al$ and $\be$, respectively, and have simple poles at $z=\al$ and $z=\be$, respectively, and that  the functions $K(z)\psi_{\al'}^{-3/2}(z)$ and $K(z)\psi_{\be'}^{-3/2}(z)$ are meromorphic in neighborhoods of $\al'$ and $\be'$, respectively, and have simple poles at $z=\al'$ and $z=\be'$, respectively.

Let us compute the residues of functions that appear in (\ref{eval11}).
Observe that
\begin{equation}\label{pb3}
\frac{\Om_n}{2}=n\om+\frac{\pi}{2},\quad \tilde u_\infty=-\frac{\om}{2}+\frac{\pi}{2},
\quad \frac{\Om_n}{2}-\tilde u_\infty=n\om+\frac{\om}{2}.
\end{equation}
From (\ref{eval19a}), \ref{eval19b}), (\ref{eval21}), and (\ref{eval22}), we obtain that 
\begin{equation}\label{eval33}
\begin{aligned}
&\underset{z=\al}{\Res}3\big(\th_{11}^2(z)+\th_{12}^2(z)\big)\left(\ga^2(z) + \ga^{-2}(z)\right)\psi_\al^{-3/2}(z) \\
&=\frac{\th_3^2(0)\th_4^2(n\om+\frac{\om}{2})}{\th_3^2(\frac{\om}{2})\th_4^2(n\om)}
\left[-3(\be'-\al)-\frac{9}{5}(\be-\al)-\frac{9}{5}(\al'-\al)+\frac{6}{5}\frac{(\al'-\al)(\be-\al)}{(\be'-\al)}\right]  \\
&\qquad -\frac{3\pi^2}{4(\be'-\al)}\left.\frac{d^2}{d\tilde{u}^2}\bigg(\th_{11}^2+\th_{12}^2\bigg)\right|_{z=\al} \\
\end{aligned}
\end{equation}
and
\begin{equation}\label{eval34}
\begin{aligned}
&\underset{z=\be}{\Res}3\big(\th_{11}^2(z)+\th_{12}^2(z)\big)\left(\ga^2(z) + \ga^{-2}(z)\right)\psi_\be^{-3/2}(z) \\
&=\frac{\th_3^2(0)\th_3^2(n\om+\frac{\om}{2})}{\th_4^2(\frac{\om}{2})\th_4^2(n\om)}
\left[3(\be-\al')+\frac{9}{5}(\be-\be')+\frac{9}{5}(\be-\al)-\frac{6}{5}\frac{(\be-\al)(\be-\be')}{(\be-\al')}\right]  \\
&\qquad +\frac{3\pi^2}{4(\be-\al')}\left.\frac{d^2}{d\tilde{u}^2}\bigg(\th_{11}^2+\th_{12}^2\bigg)\right|_{z=\be}.
\end{aligned}
\end{equation}
Also,
\begin{equation}\label{eval35}
\begin{aligned}
&\underset{z=\al}{\Res}\big(\th_{11}(z)\th_{12}(z)\big)\left(\ga^2(z) - \ga^{-2}(z)\right)\psi_\al^{-3/2}(z) \\
&=\frac{\th_3^2(0)\th_4^2(n\om+\frac{\om}{2})}{\th_3^2(\frac{\om}{2})\th_4^2(n\om)}
\left[-\frac{(\be'-\al)}{2}+\frac{3}{10}(\be-\al)+\frac{3}{10}(\al'-\al)
-\frac{1}{5}\frac{(\al'-\al)(\be-\al)}{(\be'-\al)}\right]  \\
&\qquad +\frac{\pi^2}{4(\be'-\al)}\left.\frac{d^2}{d\tilde{u}^2}\bigg(\th_{11}\th_{12}\bigg)\right|_{z=\al}, 
\end{aligned}
\end{equation}
and
\begin{equation}\label{eval36}
\begin{aligned}
&\underset{z=\be}{\Res}\big(\th_{11}(z)\th_{12}(z)\big)\left(\ga^2(z) - \ga^{-2}(z)\right)\psi_\be^{-3/2}(z) \\
&=\frac{\th_3^2(0)\th_3^2(n\om+\frac{\om}{2})}{\th_4^2(\frac{\om}{2})\th_4^2(n\om)}
\left[-\frac{(\be-\al')}{2}+\frac{3}{10}(\be-\be')+\frac{3}{10}(\be-\al)
-\frac{1}{5}\frac{(\be-\al)(\be-\be')}{(\be-\al')}\right]  \\
&\qquad +\frac{\pi^2}{4(\be-\al')}\left.\frac{d^2}{d\tilde{u}^2}\bigg(\th_{11}\th_{12}\bigg)\right|_{z=\be}. 
\end{aligned}
\end{equation}

From (\ref{eval19}), (\ref{eval27}), and (\ref{eval28}), we obtain that
\begin{equation}\label{eval37}
\underset{z=\al}{\Res} 6\big(\th_{11}^2(z)-\th_{12}^2(z)\big)\psi_\al^{-3/2}(z)
=-12\pi\frac{\th_3^2(0)\th_4^2(n\om+\frac{\om}{2})}{\th_3^2(\frac{\om}{2})\th_4^2(n\om)} 
\left[\frac{\th_3'(\frac{\om}{2})}{\th_3(\frac{\om}{2})}
-\frac{\th_4'(n\om+\frac{\om}{2})}{\th_4(n\om+\frac{\om}{2})}\right]  
\end{equation}
and
\begin{equation}\label{eval38}
\underset{z=\be}{\Res} 6\big(\th_{11}^2(z)-\th_{12}^2(z)\big)\psi_\be^{-3/2}(z)
=-12\pi\frac{\th_3^2(0)\th_3^2(n\om+\frac{\om}{2})}{\th_4^2(\frac{\om}{2})\th_4^2(n\om)} 
 \left[\frac{\th_4'(\frac{\om}{2})}{\th_4(\frac{\om}{2})}
-\frac{\th_3'(n\om+\frac{\om}{2})}{\th_3'(n\om+\frac{\om}{2})}\right].
\end{equation}
We now turn our attention to the residues at the inner turning points.  From (\ref{eval19c}), (\ref{eval19d}), 
(\ref{eval23}), and (\ref{eval24}),  we have
\begin{equation}\label{eval41}
\begin{aligned}
&\underset{z=\al'}{\Res} 3\psi_{\al'}^{-3/2}(z)J(z)\big(\ga^2(z)+\ga^{-2}(z)\big)\\
&=\frac{\th_3^2(0)\th_1^2(n\om+\frac{\om}{2})}{\th_2^2(\frac{\om}{2})\th_4^2(n\om)}
\left[3(\be-\al')+\frac{9}{5}(\be'-\al')
 -\frac{9}{5}(\al'-\al)+\frac{6}{5}\frac{(\al'-\al)(\be'-\al')}{\be-\al'}\right] \\
&\qquad +\frac{3\pi^2}{4(\be-\al')}\left.\frac{d^2}{d\tilde{u}^2}J(z)\right|_{z=\al'}
\end{aligned}
\end{equation}
and
\begin{equation}\label{eval42}
\begin{aligned}
&\underset{z=\be'}{\Res} 3\psi_{\be'}^{-3/2}(z)J(z)\big(\ga^2(z)+\ga^{-2}(z)\big)\\
&=\frac{\th_3^2(0)\th_2^2(n\om+\frac{\om}{2})}{\th_1^2(\frac{\om}{2})\th_4^2(n\om)}
\left[-3(\be'-\al)-\frac{9}{5}(\be'-\al') +\frac{9}{5}(\be-\be')-\frac{6}{5}\frac{(\be'-\al')(\be-\be')}{\be'-\al}\right] \\
&\qquad -\frac{3\pi^2}{4(\be'-\al)}\left.\frac{d^2}{d\tilde{u}^2}J(z)\right|_{z=\be'}.
\end{aligned}
\end{equation}
Also,
\begin{equation}\label{eval43}
\begin{aligned}
&\underset{z=\al'}{\Res}\psi_{\al'}^{-3/2}(z)\big(\ga^2(z)-\ga^{-2}(z)\big)\th_{11}(z)\th_{12}(z)\\
&=\frac{\th_3^2(0)\th_1^2(n\om+\frac{\om}{2})}{\th_2^2(\frac{\om}{2})\th_4^2(n\om)}
\left[-\frac{(\be-\al')}{2}+\frac{3}{10}(\be'-\al')
-\frac{3}{10}(\al'-\al)+\frac{1}{5}\frac{(\al'-\al)(\be'-\al')}{\be-\al'}\right] \\
&\qquad +\frac{\pi^2}{4(\be-\al')}\left.\frac{d^2}{d\tilde{u}^2}\bigg(\th_{11}\th_{12}\bigg)\right|_{z=\al'}
\end{aligned}
\end{equation}
and
\begin{equation}\label{eval44}
\begin{aligned}
&\underset{z=\be'}{\Res}\psi_{\be'}^{-3/2}(z)\big(\ga^2(z)-\ga^{-2}(z)\big)\th_{11}(z)\th_{12}(z)\\
&=\frac{\th_3^2(0)\th_2^2(n\om+\frac{\om}{2})}{\th_1^2(\frac{\om}{2})\th_4^2(n\om)}
\left[-\frac{(\be'-\al)}{2}+\frac{3}{10}(\be'-\al')
-\frac{3}{10}(\be-\be')+\frac{1}{5}\frac{(\be'-\al')(\be-\be')}{\be'-\al}\right] \\
&\qquad +\frac{\pi^2}{4(\be'-\al)}\left.\frac{d^2}{d\tilde{u}^2}\bigg(\th_{11}\th_{12}\bigg)\right|_{z=\be'}.
\end{aligned}
\end{equation}

From (\ref{eval19}), (\ref{eval31}), (\ref{eval32}), we have
\begin{equation}\label{eval45}
\underset{z=\al'}{\Res}6\psi_{\al'}^{-3/2}(z)K(z)
=12\pi\frac{\th_3^2(0)\th_1^2(n\om+\frac{\om}{2})}{\th_2^2(\frac{\om}{2})\th_4^2(n\om)}
\left[\frac{\th_2'(\frac{\om}{2})}{\th_2(\frac{\om}{2})}
-\frac{\th_1'(n\om+\frac{\om}{2})}{\th_1(n\om+\frac{\om}{2})}\right],
\end{equation}
and
\begin{equation}\label{eval46}
\underset{z=\be'}{\Res}6\psi_{\be'}^{-3/2}(z)K(z)
=12\pi\frac{\th_3^2(0)\th_2^2(n\om+\frac{\om}{2})}{\th_1^2(\frac{\om}{2})\th_4^2(n\om)}
\left[\frac{\th_1'(\frac{\om}{2})}{\th_1(\frac{\om}{2})}
-\frac{\th_2'(n\om+\frac{\om}{2})}{\th_2(n\om+\frac{\om}{2})}\right].
\end{equation}
Combining (\ref{eval10}),(\ref{eval11}),(\ref{eval33}),(\ref{eval35}), and (\ref{eval37}), we get that
\begin{equation}\label{eval46a}
\begin{aligned}
\underset{z=\al}{\Res}[j_X^0(z)]=&\frac{i}{96}\bigg[\frac{\th_3^2(0)\th_4^2(n\om+\frac{\om}{2})}{\th_3^2(\frac{\om}{2})\th_4^2(n\om)}\left(-\frac{7}{2}(\be'-\al)-\frac{3}{2}(\be-\al)-\frac{3}{2}(\al'-\al)+\frac{(\al'-\al)(\be-\al)}{\be'-\al}\right. \\
&\left. -12\pi\left(\frac{\th_3'(\frac{\om}{2})}{\th_3(\frac{\om}{2})}-\frac{\th_4'(n\om+\frac{\om}{2})}{\th_4(n\om+\frac{\om}{2})}\right)\right)-\frac{\pi^2}{4(\be'-\al)}\frac{d^2}{d\tilde{u}^2}\bigg[3(\th_{11}^2+\th_{12}^2)-\th_{11}\th_{12}\bigg]_{z=\al}\bigg].
\end{aligned}
\end{equation}
Similarly, we have
\begin{equation}\label{eval46ap}
\begin{aligned}
\underset{z=\al'}{\Res}[j_X^0(z)]=\frac{i}{96}&\left[\frac{\th_3^2(0)\th_1^2(n\om+\frac{\om}{2})}{\th_2^2(\frac{\om}{2})\th_4^2(n\om)}\left(\frac{7}{2}(\be-\al')+\frac{3}{2}(\be'-\al')-\frac{3}{2}(\al'-\al)+\frac{(\al'-\al)(\be'-\al')}{\be-\al'}\right.\right. \\
&\left.\left. +12\pi\left(\frac{\th_2'(\frac{\om}{2})}{\th_2(\frac{\om}{2})}-\frac{\th_1'(n\om+\frac{\om}{2})}{\th_1(n\om+\frac{\om}{2})}\right)\right)+\frac{\pi^2}{4(\be-\al')}\frac{d^2}{d\tilde{u}^2}\bigg[3J(z)-\th_{11}\th_{12}\bigg]_{z=\al'}\right],
\end{aligned}
\end{equation}
\begin{equation}\label{eval46bp}
\begin{aligned}
\underset{z=\be'}{\Res}[j_X^0(z)]=\frac{i}{96}&\bigg[\frac{\th_3^2(0)\th_2^2(n\om+\frac{\om}{2})}{\th_1^2(\frac{\om}{2})\th_4^2(n\om)}\left[\frac{7}{2}(\be'-\al)+\frac{3}{2}(\be'-\al')-\frac{3}{2}(\be-\be')+\frac{(\be-\be')(\be'-\al')}{\be'-\al}\right. \\
&\left. -12\pi\left(\frac{\th_1'(\frac{\om}{2})}{\th_1(\frac{\om}{2})}-\frac{\th_2'(n\om+\frac{\om}{2})}{\th_2(n\om+\frac{\om}{2})}\right)\right]+\frac{\pi^2}{4(\be'-\al)}\frac{d^2}{d\tilde{u}^2}\bigg[3J(z)-\th_{11}\th_{12}\bigg]_{z=\be'}\bigg],
\end{aligned}
\end{equation}
and
\begin{equation}\label{eval46b}
\begin{aligned}
\underset{z=\be}{\Res}[j_X^0(z)]=&\frac{i}{96}\bigg[\frac{\th_3^2(0)\th_3^2(n\om+\frac{\om}{2})}{\th_4^2(\frac{\om}{2})\th_4^2(n\om)}\left(-\frac{7}{2}(\be-\al')-\frac{3}{2}(\be-\al)-\frac{3}{2}(\be-\be')+\frac{(\be-\be')(\be-\al)}{\be-\al'}\right. \\
&\left. +12\pi\left(\frac{\th_4'(\frac{\om}{2})}{\th_4(\frac{\om}{2})}-\frac{\th_3'(n\om+\frac{\om}{2})}{\th_3(n\om+\frac{\om}{2})}\right)\right)-\frac{\pi^2}{4(\be-\al')}\frac{d^2}{d\tilde{u}^2}\bigg[3(\th_{11}^2+\th_{12}^2)-\th_{11}\th_{12}\bigg]_{z=\be}\bigg].
\end{aligned}
\end{equation}
Using MAPLE for calculations, we get
\begin{equation}\label{eval49}
\begin{aligned}
&\left.\frac{d^2}{d\tilde{u}^2}[3(\th_{11}^2+\th_{12}^2)-\th_{11}\th_{12}]\right|_{z=\al}
=2\frac{\th_3^2(0)\th_4^2(n\om+\frac{\om}{2})}{\th_3^2(\frac{\om}{2})\th_4^2(n\om)}
\left[5\frac{\th_4^{''}(n\om+\frac{\om}{2})}{\th_4(n\om+\frac{\om}{2})}
-5\frac{\th_3^{''}(\frac{\om}{2})}{\th_3(\frac{\om}{2})}\right. \\
&\hskip 3cm +7\left(\frac{\th_4'(n\om+\frac{\om}{2})}{\th_4(n\om+\frac{\om}{2})}\right)^2
+17\left(\frac{\th_3'(\frac{\om}{2})}{\th_3(\frac{\om}{2})}\right)^2 
\left. -24\frac{\th_4'(n\om+\frac{\om}{2})\th_3'(\frac{\om}{2})}{\th_4(n\om+\frac{\om}{2})\th_3(\frac{\om}{2})}\right],
\end{aligned}
\end{equation}
\begin{equation}\label{eval50}
\begin{aligned}
&\left.\frac{d^2}{d\tilde{u}^2}[3(\th_{11}^2+\th_{12}^2)-\th_{11}\th_{12}]\right|_{z=\be}
=2\frac{\th_3^2(0)\th_3^2(n\om+\frac{\om}{2})}{\th_4^2(\frac{\om}{2})\th_4^2(n\om)}
\left[5\frac{\th_3^{''}(n\om+\frac{\om}{2})}{\th_3(n\om+\frac{\om}{2})}
-5\frac{\th_4^{''}(\frac{\om}{2})}{\th_4(\frac{\om}{2})}\right. \\
&\hskip 3cm +7\left(\frac{\th_3'(n\om+\frac{\om}{2})}{\th_3(n\om+\frac{\om}{2})}\right)^2
+17\left(\frac{\th_4'(\frac{\om}{2})}{\th_4(\frac{\om}{2})}\right)^2 
\left. -24\frac{\th_3'(n\om+\frac{\om}{2})\th_4'(\frac{\om}{2})}{\th_3(n\om+\frac{\om}{2})\th_4(\frac{\om}{2})}\right],
\end{aligned}
\end{equation}
\begin{equation}\label{eval51}
\begin{aligned}
&\left.\frac{d^2}{d\tilde{u}^2}[3J(z)-\th_{11}\th_{12}]\right|_{z=\al'}
=2\frac{\th_3^2(0)\th_1^2(n\om+\frac{\om}{2})}{\th_2^2(\frac{\om}{2})\th_4^2(n\om)}
\left[5\frac{\th_1^{''}(n\om+\frac{\om}{2})}{\th_1(n\om+\frac{\om}{2})}
-5\frac{\th_2^{''}(\frac{\om}{2})}{\th_2(\frac{\om}{2})}\right. \\
&\hskip 3cm +7\left(\frac{\th_1'(n\om+\frac{\om}{2})}{\th_1(n\om+\frac{\om}{2})}\right)^2
+17\left(\frac{\th_2'(\frac{\om}{2})}{\th_2(\frac{\om}{2})}\right)^2 
\left. -24\frac{\th_1'(n\om+\frac{\om}{2})\th_2'(\frac{\om}{2})}{\th_1(n\om+\frac{\om}{2})\th_2(\frac{\om}{2})}\right],
\end{aligned}
\end{equation}
\begin{equation}\label{eval52}
\begin{aligned}
&\left.\frac{d^2}{d\tilde{u}^2}[3J(z)-\th_{11}\th_{12}]\right|_{z=\be'}
=2\frac{\th_3^2(0)\th_2^2(n\om+\frac{\om}{2})}{\th_1^2(\frac{\om}{2})\th_4^2(n\om)}
\left[5\frac{\th_2^{''}(n\om+\frac{\om}{2})}{\th_2(n\om+\frac{\om}{2})}
-5\frac{\th_1^{''}(\frac{\om}{2})}{\th_1(\frac{\om}{2})}\right. \\
&\hskip 3cm +7\left(\frac{\th_2'(n\om+\frac{\om}{2})}{\th_2(n\om+\frac{\om}{2})}\right)^2
+17\left(\frac{\th_1'(\frac{\om}{2})}{\th_1(\frac{\om}{2})}\right)^2 
\left. -24\frac{\th_2'(n\om+\frac{\om}{2})\th_1'(\frac{\om}{2})}{\th_2(n\om+\frac{\om}{2})\th_1(\frac{\om}{2})}\right].
\end{aligned}
\end{equation}

These formulae, combined with (\ref{eval12a}), prove Lemma \ref{X1}.

\section{Large $n$ asymptotic formula for $h_n$}

We evaluate the large $n$ asymptotic behavior of
$h_{nn}$ and then we use formula (\ref{rw5}). By (\ref{IP6}), $h_{nn}=[\bold P_1]_{12}$, and by (\ref{red3}),
\begin{equation}\label{zn1}
[\bold P_1]_{12}=[\bold R_1]_{12}\left(-\frac{n\pi i}{\ga}\right)\,,
\end{equation}
hence
\begin{equation}\label{zn2}
h_{nn}=[\bold R_1]_{12}\left(-\frac{n\pi i}{\ga}\right)\,.
\end{equation}
Furthermore, from (\ref{ft3b}) we obtain that
\begin{equation}\label{zn3}
h_{nn}=e^{nl}[\bold T_1]_{12}\left(-\frac{n\pi i}{\ga}\right)\,,
\end{equation}
and from (\ref{st2}), that
\begin{equation}\label{zn4}
h_{nn}=e^{nl}[\bold S_1]_{12}\left(-\frac{n\pi i}{\ga}\right)\,.
\end{equation}
It follows from (\ref{tt1}) that
\begin{equation}\label{zn5}
\bold S_1=\bold M_1+\bold X_1.
\end{equation}
By (\ref{m36}),
\begin{equation}\label{zn6}
\left[\bold M_1\right]_{12}=\frac{iA\th_4\big((n+1)\om\big)}{\th_4(n\om)}\,,\qquad
\om=\frac{\pi(1+\z)}{2}\,,\qquad  A=\frac{\pi \th_1'(0)}{2\th_1(\om)}\,,
 \end{equation}
and by (\ref{x_1}), 
\begin{equation}\label{zn7}
\left[\bold X_1\right]_{12}=\frac{c(n)}{n}+O(n^{-2}),
 \end{equation}
where
\begin{equation}\label{zn7a}
c(n)=X_{\al}+X_{\al'}+X_{\be'}+X_{\be}
\end{equation}
is an explicit quasi-periodic function of $n$.
Therefore,
\begin{equation}\label{zn8}
h_{nn}=e^{nl}\left[\frac{iA\th_4\big((n+1)\om\big)}{\th_4(n\om)}
+\frac{c(n)}{n}+O(n^{-2})\right]\left(-\frac{n\pi i}{\ga}\right)\,.
\end{equation}
By (\ref{lm18}),
\begin{equation}\label{zn9}
e^{\frac{l}{2}}= \frac{\pi \th'_1(0)}{2e\th_1(\om)}=\frac{A}{e}\,,
\end{equation}
hence
\begin{equation}\label{zn10}
\begin{aligned}
h_{nn}&=\left(\frac{A}{e}\right)^{2n}\left[iA\,\frac{\th_4\big((n+1)\om\big)}{\th_4(n\om)}
+\frac{c(n)}{n}+O(n^{-2})\right]\left(-\frac{n\pi i}{\ga}\right)\\
&=\frac{n\pi A^{2n+1}\th_4\big((n+1)\om\big)}{\ga e^{2n}\th_4(n\om)}
\left(1+\frac{c_1(n)}{n}+O(n^{-2})\right),
\end{aligned}
\end{equation}
where
\begin{equation}\label{zn11}
c_1(n)=\frac{c(n)\th_4(n\om)}{iA\th_4\big((n+1)\om\big)}\,.
\end{equation}
From (\ref{rw5}) and the Stirling formula we obtain that
\begin{equation} \label{zn12}
\frac{h_n}{(n!)^2}=\frac{n^{2n}h_{nn}}{(n!)^2(2\ga)^{2n}}=
\left(\frac{e}{2\ga}\right)^{2n}\frac{h_{nn}}{2\pi n}\left(1-\frac{1}{6n}+O(n^{-2})\right)\,,
\end{equation}
hence by (\ref{zn10}),
\begin{equation} \label{zn13}
\begin{aligned}
\frac{h_n}{(n!)^2}&=
\left(\frac{e}{2\ga}\right)^{2n}\frac{1}{2\pi n}\,\frac{n\pi A^{2n+1}\th_4\big((n+1)\om\big)}{\ga e^{2n}\th_4(n\om)}
\left(1+\frac{c_1(n)}{n}-\frac{1}{6n}+O(n^{-2})\right)\\
&=G^{2n+1}\,\frac{\th_4\big((n+1)\om\big)}{\th_4(n\om)}
\left(1+\frac{c_2(n)}{n}+O(n^{-2})\right)\,,
\end{aligned}
\end{equation}
where
\begin{equation}\label{zn14}
G=\frac{A}{2\ga}=  \frac{\pi \th'_1(0)}{4\ga\th_1(\om)},\qquad c_2(n)=c_1(n)-\frac{1}{6}\,.
\end{equation}
Observe that $c_1(n)$ has the form,
\begin{equation}\label{c1n_1}
c_1(n)=f(n\om,\om),
\end{equation}
where $f(x,\om)$ is a real analytic function, periodic with respect to both $x$ and $\om$,
of periods $\pi$ and $2\pi$, respectively,
so that
\begin{equation}\label{c1n_2}
f(x+\pi,\om)=f(x,\om),\qquad f(x,\om+2\pi)=f(x,\om).
\end{equation}
We can now summarize now the asymptotic formula for $h_n/(n!)^2$.

\begin{prop}\label{h_n1} As $n\to\infty$, 
\begin{equation} \label{hn1}
\frac{h_n}{(n!)^2}= G^{2n+1}\,\frac{\th_4\big((n+1)\om\big)}{\th_4(n\om)}
\left[1+\frac{f_0(n\om,\om)}{n}+O(n^{-2})\right]\,,
\end{equation}
where 
\begin{equation} \label{hn2}
 G=\frac{\pi \th'_1(0)}{4\ga\th_1(\om)}
\end{equation}
and 
\begin{equation} \label{hn3}
 f_0(x,\om)=f(x,\om)-\frac{1}{6}
\end{equation}
is a real analytic function which satisfies the periodicity conditions
\begin{equation}\label{hn4}
f_0(x+\pi,\om)=f_0(x,\om),\qquad f_0(x,\om+2\pi)=f_0(x,\om).
\end{equation}
\end{prop}

By (\ref{zn11}),
\begin{equation}\label{hn5}
f(n\om,\om)=\frac{X(n\om,\om)\th_4(n\om)}{iA\th_4\big((n+1)\om\big)}\,.
\end{equation}
where
\begin{equation}\label{hn6}
X\equiv X_\al+X_{\al'}+X_{\be'}+X_\be
\end{equation}
and explicit expressions for $X_\al,X_{\al'},X_{\be'},X_\be$ are given in Lemma \ref{X1}.

In the subsequent sections, we carry out a concrete evaluation of $f(n\om,\om)$, obtaining that
\begin{equation}\label{hn7}
f_0(x,\om)\equiv 0,
\end{equation}
thus we can improve Proposition \ref{h_n1} to the following.
\begin{prop}\label{h_n2}
As $n\to\infty$, 
\begin{equation} \label{hn8}
\frac{h_n}{(n!)^2}= G^{2n+1}\,\frac{\th_4\big((n+1)\om\big)}{\th_4(n\om)}
\bigg(1+O(n^{-2})\bigg)\,,
\end{equation}
where $G$ is defined in (\ref{hn2}).
\end{prop}
To this end, we will first show that $f(n\om,\om)$ does not depend on $n$.

\medskip

\section{$f(x,\om)$ is constant in $x$}
Denote
\begin{equation}\label{Q1}
z=n\om+\frac{\om}{2}\,, \quad \tilde{f}(x,\om) \equiv f(x-\frac{\om}{2},\om)\,,
\end{equation}
so that 
\begin{equation}\label{Q1a}
\tilde{f}(z,\om)=f(n\om,\om).
\end{equation}
To prove that $f(x,\om)$ is constant in $x$ we will prove the following lemmas.
\begin{lem}\label{doubleper}
The function $\tilde{f}(z,\om)$ is doubly periodic in $z$.
\end{lem}
\begin{lem}\label{residue}
The function $\tilde{f}(z,\om)$ is analytic throughout the $z$-plane.
\end{lem}
From these two lemmas, it follows immediately that $\tilde{f}(z,\om)$ is constant in $z$, as it is a doubly periodic entire function of $z$, and it thus follows that $f(x,\om)$ is constant in $x$.

To prove Lemma \ref{doubleper}, we will check that $\tilde{f}(z+\pi\tau,\om)=\tilde{f}(z,\om)$.
By (\ref{zn11}), (\ref{zn7a}),
\begin{equation}\label{aver1}
c_1(n)\equiv f(n\om,\om)=\tilde{f}(z,\om)=Y_1+Y_2+Y_3+Y_4,
\end{equation}
where 
\begin{equation}\label{aver2}
\begin{aligned}
Y_1&=\frac{X_\al \th_4(n\om)}{iA\th_4(n\om+\om)}= 
\frac{\th_1(\om)\th_3^2(0)\th_4^2(z)}{48\pi \th'_1(0)\th_3^2(\frac{\om}{2})\th_4(z-\frac{\om}{2})\th_4(z+\frac{\om}{2})}
\left(C_\al+12\pi \xi_\al+\frac{\pi^2\eta_\al}{2(\be'-\al)}\right), \\
Y_2&=\frac{X_{\al'} \th_4(n\om)}{iA\th_4(n\om+\om)}= 
\frac{\th_1(\om)\th_3^2(0)\th_1^2(z)}{48\pi \th'_1(0)\th_2^2(\frac{\om}{2})\th_4(z-\frac{\om}{2})\th_4(z+\frac{\om}{2})}
\left(C_{\al'}+12\pi \xi_{\al'}+\frac{\pi^2\eta_{\al'}}{2(\be-\al')}\right), \\
Y_3&=\frac{X_{\be'} \th_4(n\om)}{iA\th_4(n\om+\om)}= 
\frac{\th_1(\om)\th_3^2(0)\th_2^2(z)}{48\pi \th'_1(0)\th_1^2(\frac{\om}{2})\th_4(z-\frac{\om}{2})\th_4(z+\frac{\om}{2})}
\left(C_{\be'}+12\pi \xi_{\be'}+\frac{\pi^2\eta_{\be'}}{2(\be'-\al)}\right), \\
Y_4&=\frac{X_{\be} \th_4(n\om)}{iA\th_4(n\om+\om)}= 
\frac{\th_1(\om)\th_3^2(0)\th_3^2(z)}{48\pi \th'_1(0)\th_4^2(\frac{\om}{2})\th_4(z-\frac{\om}{2})\th_4(z+\frac{\om}{2})}
\left(C_{\be}+12\pi \xi_{\be}+\frac{\pi^2\eta_{\be}}{2(\be-\al')}\right).
\end{aligned}
\end{equation}
Observe that $Y_1,Y_2,Y_3,Y_4$  can be written in the form
\begin{equation}\label{aver3a}
\begin{aligned}
Y_1&=Q_{11}h_{11}(z)+Q_{12}h_{12}(z)+Q_{13}h_{13}(z)+Q_{14}h_{14}(z), \\
Y_2&=Q_{21}h_{21}(z)+Q_{22}h_{22}(z)+Q_{23}h_{23}(z)+Q_{24}h_{24}(z), \\
Y_3&=Q_{31}h_{31}(z)+Q_{32}h_{32}(z)+Q_{33}h_{33}(z)+Q_{34}h_{34}(z), \\
Y_4&=Q_{41}h_{41}(z)+Q_{42}h_{42}(z)+Q_{43}h_{43}(z)+Q_{44}h_{44}(z), \\
\end{aligned}
\end{equation}
where 
\begin{equation}\label{aver3e}
\begin{aligned}
h_{11}(z)&=\frac{\th_4^2(z)}{\th_4(z-\frac{\om}{2})\th_4(z+\frac{\om}{2})}, \
h_{12}(z)=\frac{\th'_4(z)\th_4(z)}{\th_4(z-\frac{\om}{2})\th_4(z+\frac{\om}{2})} , \
h_{13}(z)=\frac{\th'_4(z)^2}{\th_4(z-\frac{\om}{2})\th_4(z+\frac{\om}{2})}, \\
h_{14}(z)&=\frac{\th''_4(z)\th_4(z)}{\th_4(z-\frac{\om}{2})\th_4(z+\frac{\om}{2})} , \
h_{21}(z)=\frac{\th_1^2(z)}{\th_4(z-\frac{\om}{2})\th_4(z+\frac{\om}{2})}, \
h_{22}(z)=\frac{\th'_1(z)\th_1(z)}{\th_4(z-\frac{\om}{2})\th_4(z+\frac{\om}{2})} , \\
h_{23}(z)&=\frac{\th'_1(z)^2}{\th_4(z-\frac{\om}{2})\th_4(z+\frac{\om}{2})} , \
h_{24}(z)=\frac{\th''_1(z)\th_1(z)}{\th_4(z-\frac{\om}{2})\th_4(z+\frac{\om}{2})} , \
h_{31}(z)=\frac{\th_2^2(z)}{\th_4(z-\frac{\om}{2})\th_4(z+\frac{\om}{2})}, \\
h_{32}(z)&=\frac{\th'_2(z)\th_2(z)}{\th_4(z-\frac{\om}{2})\th_4(z+\frac{\om}{2})} , \
h_{33}(z)=\frac{\th'_2(z)^2}{\th_4(z-\frac{\om}{2})\th_4(z+\frac{\om}{2})} , \
h_{34}(z)=\frac{\th''_2(z)\th_2(z)}{\th_4(z-\frac{\om}{2})\th_4(z+\frac{\om}{2})} , \\
h_{41}(z)&=\frac{\th_3^2(z)}{\th_4(z-\frac{\om}{2})\th_4(z+\frac{\om}{2})} , \
h_{42}(z)=\frac{\th'_3(z)\th_3(z)}{\th_4(z-\frac{\om}{2})\th_4(z+\frac{\om}{2})} , \
h_{43}(z)=\frac{\th'_3(z)^2}{\th_4(z-\frac{\om}{2})\th_4(z+\frac{\om}{2})} , \\
h_{44}(z)&=\frac{\th''_3(z)\th_3(z)}{\th_4(z-\frac{\om}{2})\th_4(z+\frac{\om}{2})},
\end{aligned}
\end{equation}
and the numbers $Q_{ij}$ do not depend on $z$. More specifically,
\begin{equation}\label{aver4}
\begin{aligned}
Q_{11}&=
\frac{\th_1(\om)\th_3^2(0)}{48\pi \th'_1(0)\th_3^2(\frac{\om}{2})}
\left[C_\al+12\pi \frac{\th_3'(\frac{\om}{2})}{\th_3(\frac{\om}{2})}+\frac{\pi^2}{2(\be'-\al)}
\left(-5\frac{\th_3^{''}(\frac{\om}{2})}{\th_3(\frac{\om}{2})}
+17\left(\frac{\th_3^{'}(\frac{\om}{2})}{\th_3(\frac{\om}{2})}\right)^2
\right)\right],\\
Q_{12}&=-\frac{\th_1(\om)\th_3^2(0)}{4 \th'_1(0)\th_3^2(\frac{\om}{2})}
-\frac{\pi\th_1(\om)\th_3^2(0)\th_3^{'}(\frac{\om}{2})}{4(\be'-\al)\th'_1(0)\th_3^3(\frac{\om}{2})}\,,
\qquad Q_{13}=\frac{7\pi\th_1(\om)\th_3^2(0)}{96(\be'-\al) \th'_1(0)\th_3^2(\frac{\om}{2})}\,,\\
Q_{14}&=\frac{5\pi\th_1(\om)\th_3^2(0)}{96(\be'-\al) \th'_1(0)\th_3^2(\frac{\om}{2})}\,,
\end{aligned}
\end{equation}
and similar formulae hold for other $Q_{ij}$. 
In particular, notice that 
\begin{equation}\label{Q1aa}
\begin{aligned}
Q_{13}+Q_{14}&=\frac{\pi\th_1(\om)\th_3^2(0)}{8(\be'-\al) \th'_1(0)\th_3^2(\frac{\om}{2})}=
\frac{\th_1(\om)\th_1(\frac{\om}{2})}{8 \th'_1(0)\th_2(\frac{\om}{2})\th_3(\frac{\om}{2})\th_4(\frac{\om}{2})}\,,\\
Q_{23}+Q_{24}&=-\frac{\pi\th_1(\om)\th_3^2(0)}{8(\be-\al') \th'_1(0)\th_2^2(\frac{\om}{2})}
=-\frac{\th_1(\om)\th_4(\frac{\om}{2})}{8 \th'_1(0)\th_1(\frac{\om}{2})\th_2(\frac{\om}{2})\th_3(\frac{\om}{2})}\,,\\
Q_{33}+Q_{34}&=-\frac{\pi\th_1(\om)\th_3^2(0)}{8(\be'-\al) \th'_1(0)\th_1^2(\frac{\om}{2})}
=-\frac{\th_1(\om)\th_3(\frac{\om}{2})}{8 \th'_1(0)\th_1(\frac{\om}{2})\th_2(\frac{\om}{2})\th_4(\frac{\om}{2})}\,,\\
Q_{43}+Q_{44}&=\frac{\pi\th_1(\om)\th_3^2(0)}{8(\be-\al') \th'_1(0)\th_4^2(\frac{\om}{2})}
=\frac{\th_1(\om)\th_2(\frac{\om}{2})}{8 \th'_1(0)\th_1(\frac{\om}{2})\th_3(\frac{\om}{2})\th_4(\frac{\om}{2})}\,.
\end{aligned}
\end{equation}
Observe that all $Q_{ij}=Q_{ij}(\om)$ are periodic functions
of $\om$ of period $2\pi$. 
It follows from equations (\ref{main9}) that the functions
\begin{equation}\label{Q3}
h_{j1}(z),\qquad j=1,2,3,4,
\end{equation}
are doubly periodic,
\begin{equation}\label{Q4}
h_{j1}(z+\pi)=h_{j1}(z)\,,\qquad h_{j1}(z+\pi\tau)=h_{j1}(z),
\end{equation}
while the functions,
\begin{equation}\label{Q5}
\begin{aligned}
h_{j2}(z), \ h_{j3}(z), \ h_{j4}(z), \quad j=1,2,3,4
\end{aligned}
\end{equation}
satisfy the equations,
\begin{equation}\label{Q6}
\begin{aligned}
&h_{jk}(z+\pi)=h_{jk}(z),\qquad k=2,3,4;\\
&h_{j2}(z+\pi\tau)=h_{j2}(z)-2ih_{j1}(z),\qquad h_{j3}(z+\pi\tau)=h_{j3}(z)-4ih_{j2}(z)-4h_{j1}(z),\\
&h_4(z+\pi\tau)=h_{j4}(z)-4ih_{j2}(z)-4h_{j1}(z).
\end{aligned}
\end{equation}
This implies that the functions $Y_j=Y_j(z)$ for $j=1,2,3,4$, satisfy the equations
\begin{equation}\label{Q6a}
\begin{aligned}
Y_j(z+\pi)&=Y_j(z), \\
Y_j(z+\pi\tau)&=Y_j(z)+(-2iQ_{j2}-4Q_{j3}-4Q_{j4})h_{j1}(z)\\
&+(-4iQ_{j3}-4iQ_{j4})h_{j2}(z).
\end{aligned}
\end{equation}
The proof of Lemma \ref{doubleper} then follows immediately from (\ref{Q6a}) and the following identities:
\begin{equation}\label{Q8}
(Q_{13}+Q_{14})h_{11}(z)+(Q_{23}+Q_{24})h_{21}(z)+(Q_{33}+Q_{34})h_{31}(z)+(Q_{43}+Q_{44})h_{41}(z) \equiv 0
\end{equation}
\begin{equation}\label{Q9}
(Q_{13}+Q_{14})h_{12}(z)+(Q_{23}+Q_{24})h_{22}(z)+(Q_{33}+Q_{34})h_{32}(z)+(Q_{43}+Q_{44})h_{42}(z) \equiv 0
\end{equation}
\begin{equation}\label{Q7}
Q_{12}h_{11}(z)+Q_{22}h_{21}(z)+Q_{32}h_{31}(z)+Q_{42}h_{41}(z) \equiv 0,
\end{equation}
which are proven below.
Introduce here the notation
\begin{equation}\label{Q7a}
\th_j \equiv \th_j(\frac{\om}{2}), \ \th_j' \equiv \th_j'(\frac{\om}{2}), \ \th_j'' \equiv \th_j''(\frac{\om}{2})  \quad \textrm{for} \ j=1,2,3,4.
\end{equation}

The sum in (\ref{Q8}) can be written as
\begin{equation}\label{Q10}
\frac{\th_1(\om)\big[\th_1^2\th_4^2(z)-\th_4^2\th_1^2(z)+\th_2^2\th_3^2(z)-\th_3^2\th_2^2(z)\big]}{8\th_1'(0)\th_1\th_2\th_3\th_4\th_4(z-\frac{\om}{2})\th_4(z+\frac{\om}{2})},
\end{equation}
which is zero by the Jacobi identity (\ref{Q11}).

The sum in (\ref{Q9}) can be written as
\begin{equation}\label{Q12}
\frac{\th_1(\om)\big[\th_1^2\th_4'(z)\th_4(z)-\th_4^2\th_1'(z)\th_1(z)+\th_2^2\th_3'(z)\th_3(z)-\th_3^2\th_2'(z)\th_2(z)\big]}{8\th_1'(0)\th_1\th_2\th_3\th_4\th_4(z-\frac{\om}{2})\th_4(z+\frac{\om}{2})}.
\end{equation}
Using the identities (\ref{4}), we can write the expression in brackets in the numerator of (\ref{Q12}) as
\begin{equation}\label{Q13}
\begin{aligned}
\frac{\th_1'(z)}{\th_1(z)}&\bigg[\th_1^2\th_4^2(z)-\th_4^2\th_1^2(z)-\th_3^2\th_2^2(z)+\th_2^2\th_3^2(z)\bigg] \\
&+\frac{\th_2(z)\th_3(z)\th_4(z)}{\th_1(z)}\bigg[\th_3^2\th_2^2(0)-\th_1^2\th_4^2(0)-\th_2^2\th_3^2(0)\bigg].
\end{aligned}
\end{equation}
The first term in (\ref{Q13}) vanishes by (\ref{Q11}) and the second term vanishes by (\ref{evf6a}).  Thus (\ref{Q9}) is proven.

Finally, we can expand the sum in (\ref{Q7}) and make the substitutions from identities (\ref{4}), to obtain
\begin{equation}\label{Q14}
\begin{aligned}
\frac{\th_1(w)}{4\th_1'(0)\th_4(z-\frac{\om}{2})\th_4(z+\frac{\om}{2})}\left[\frac{\th_1'\big(\th_4^2\th_1^2(z)-\th_1^2\th_4^2(z)-\th_2^2\th_3^2(z)+\th_3^2\th_2^2(z)\big)}{\th_1^2\th_2\th_3\th_4} \right.\\
\left. +\frac{\th_3^2(z)}{\th_4^2}\left(\th_3^2(0)+\frac{\th_2^2\th_4^2(0)}{\th_1^2}\right)+\frac{\th_1^2(z)}{\th_2^2}\left(\th_3^2(0)-\frac{\th_4^2\th_2^2(0)}{\th_1^2}\right)-\frac{\th_3^2(0)}{\th_1^2}\th_2^2(z)\right].
\end{aligned}
\end{equation}
Once again, the first term vanishes by (\ref{Q11}).  If we substitute the identity (\ref{evf6b})
in the second and third terms, (\ref{Q14}) becomes simply
\begin{equation}\label{Q16}
\frac{\th_1(w)}{4\th_1'(0)\th_1\th_4(z-\frac{\om}{2})\th_4(z+\frac{\om}{2})}\bigg[\th_3^2(z)\th_2^2(0)-\th_2^2(z)\th_3^2(0)-\th_1^2(z)\th_4^2(0)\bigg],
\end{equation}
which is zero by (\ref{evf6a}).  This proves (\ref{Q7}) and thus Lemma \ref{doubleper}.  

We now turn to the proof of Lemma \ref{residue}.  Notice that in the fundamental rectangle
\begin{equation}\label{evf12a}
\mathcal R
=\{z\in\C:\; -\frac{\pi}{2}\le \Re z\le \frac{\pi}{2}\,,\;0\le\Im z\le \pi\tau|\},
\end{equation}
the function $\tilde{f}(z,\om)$ has the two simple poles,
\begin{equation}\label{evf13}
z_{1,2}=\frac{\pi\tau}{2}\pm \frac{\om}{2}\,.
\end{equation}
Because of the double-periodicity of $\tilde{f}$, we must have 
\begin{equation}\label{Q17}
\underset{z=z_1}{\Res} \tilde{f}(z,\om)=-\underset{z=z_2}{\Res} \tilde{f}(z,\om).
\end{equation}
Denote
\begin{equation}\label{evf14}
R_{ij}(\om)=\underset{z=z_1}{\Res}h_{ij}(z,\om).
\end{equation}
Then we have that
\begin{equation}\label{Q19}
\begin{aligned}
R_{11}&=\frac{\th_1^2}{\th_1'(0)\th_1(\om)}, \ R_{12}=\frac{(\th_1'-i\th_1)\th_1}{\th_1'(0)\th_1(\om)}, \ R_{13}=\frac{(\th_1'-i\th_1)^2}{\th_1'(0)\th_1(\om)}, \ R_{14}=\frac{(\th_1''-2i\th_1'-\th_1)\th_1}{\th_1'(0)\th_1(\om)},\\
R_{21}&=\frac{\th_4^2}{\th_1'(0)\th_1(\om)}, \ R_{22}=\frac{(\th_4'-i\th_4)\th_4}{\th_1'(0)\th_1(\om)}, \ R_{23}=\frac{(\th_4'-i\th_4)^2}{\th_1'(0)\th_1(\om)}, \ R_{24}=\frac{(\th_4''-2i\th_4'-\th_4)\th_4}{\th_1'(0)\th_1(\om)}, \\
R_{31}&=\frac{-\th_3^2}{\th_1'(0)\th_1(\om)}, \, R_{32}=\frac{-(\th_3'-i\th_3)\th_3}{\th_1'(0)\th_1(\om)}, \, R_{33}=\frac{-(\th_3'-i\th_3)^2}{\th_1'(0)\th_1(\om)}, \, R_{34}=\frac{-(\th_3''-2i\th_3'-\th_3)\th_3}{\th_1'(0)\th_1(\om)}, \\
R_{41}&=\frac{-\th_2^2}{\th_1'(0)\th_1(\om)}, \, R_{42}=\frac{-(\th_2'-i\th_2)\th_2}{\th_1'(0)\th_1(\om)}, \, R_{43}=\frac{-(\th_2'-i\th_2)^2}{\th_1'(0)\th_1(\om)}, \, R_{44}=\frac{-(\th_2''-2i\th_2'-\th_2)\th_2}{\th_1'(0)\th_1(\om)},\\
\end{aligned}
\end{equation}
and 
\begin{equation}\label{Q20}
\underset{z=z_1}{\Res} \tilde{f}(z,\om)=\sum_{j,k=1}^4 Q_{jk}R_{jk}.
\end{equation}
A priori, the sum in (\ref{Q20}) is quite complicated, so we evaluate first the imaginary part .  Multiplying out (\ref{Q20}) and again making the substitutions from (\ref{4}), we get
\begin{equation}\label{Q21}
\Im \underset{z=z_1}{\Res} \tilde{f}(z,\om)=\frac{(\th_2^4-\th_4^4)\big(\th_1^2\th_3^2(0)+\th_2^2\th_4^2(0)-\th_4^2\th_2^2(0)\big)}{4\th_1'(0)^2\th_1^2\th_2^2\th_4^2},
\end{equation}
which is zero by (\ref{evf6b}).

Substituting the identities (\ref{Q23}), along with (\ref{4}), into the sum (\ref{Q20}) gives
\begin{equation}\label{Q24}
\underset{z=z_1}{\Res} \tilde{f}(z,\om)=\frac{24A+17B+12C+10D+7E+5F}{96\th_1'(0)^2\th_1^2\th_2^3\th_3\th_4^3},
\end{equation}
where 
\begin{equation}\label{Q25}
\begin{aligned}
A&=\th_3^2\big(\th_2^4+\th_4^4\big)\bigg[\th_2^2\th_4^2\th_2^2(0)\th_4^2(0)-\th_1^2\th_3^2(0)\left(\th_2^2\th_4^2(0)-\th_4^2\th_2^2(0)\right)\bigg], \\
B&=-\th_3^2\left(\th_2^8\th_4^4(0)+\th_4^8\th_2^4(0)\right), \\
C&=\th_1^2\th_2^2\th_4^2\bigg(\th_2^4+\th_4^4-\th_1^4-\th_3^4\bigg), \\
D&=-\th_2^4\th_3^2\th_4^4\th_3^4(0), \\
E&=-\th_3^2\bigg[\th_1^4\th_3^4(0)\left(\th_2^4+\th_4^4\right)+\th_2^4\th_4^4\big(\th_2^4(0)+\th_4^4(0)\big)\bigg], \\
F&=\th_2^2\th_4^2\th_3^2(0)\bigg[\th_2^2\th_4^2(0)\left(\th_3^4-\th_1^4-\th_4^4+\th_2^4\right)+\th_4^2\th_2^2(0)\left(\th_4^4-\th_2^4+\th_3^4-\th_1^4\right)\bigg].
\end{aligned}
\end{equation}
Note that none of these terms involve derivatives of theta functions.  We immediately have $C=0$ by the identity (\ref{Q26}).  We can also use this identity to write
\begin{equation}\label{Q27}
F=2\th_2^2\th_4^2\th_3^2(0)\bigg[\th_2^2\th_4^2(0)\left(\th_3^4-\th_4^4\right)+\th_4^2\th_2^2(0)\left(\th_3^4-\th_2^4\right)\bigg],
\end{equation}
and (\ref{evf6b}) to write
\begin{equation}\label{Q28}
A=\th_3^2\big(\th_2^4+\th_4^4\big)\bigg[\th_2^2\th_4^2\th_2^2(0)\th_4^2(0)+\th_1^4\th_3^2(0)\bigg].
\end{equation}
We now combine the terms $A,B$, and $E$ to obtain
\begin{equation}\label{Q29}
\begin{aligned}
24A+17B+7E=&24\th_2^2\th_3^2\th_4^2\big(\th_4^2\th_2^2(0)-\th_2^2\th_4^2(0)\big)\big(\th_2^2\th_2^2(0)-\th_4^2\th_4^2(0)\big) \\
&\quad -21\th_2^4\th_3^2\th_4^4\th_3^4(0)+7\th_3^6\big(\th_2^4+\th_4^4\big)\big(\th_2^4(0)+\th_4^4(0)\big) \\
&\quad -7\th_3^2\big(\th_4^8\th_4^4(0)+\th_2^8\th_2^4(0)\big).
\end{aligned}
\end{equation}
By (\ref{evf6b}) and (\ref{evf6d})
we can write this sum as
\begin{equation}\label{Q31}
\begin{aligned}
24A+17B+7E&=24\th_1^2\th_2^2\th_3^2\th_4^2\th_3^2(0)\big(\th_2^2\th_2^2(0)-\th_4^2\th_4^2(0)\big)-21\th_2^4\th_3^2\th_4^4\th_3^4(0) \\
&\qquad +7\th_3^6\th_3^4(0)\big(\th_2^4+\th_4^4\big)-7\th_3^2\big(\th_4^8\th_4^4(0)+\th_2^8\th_2^4(0)\big).
\end{aligned}
\end{equation}
We now combine all terms and use (\ref{evf6b}) and (\ref{evf6d})
to write all terms solely in terms of $\th_2,\th_4$, and factors which are constant with respect to $\om$, yielding
\begin{equation}\label{Q32}
\begin{aligned}
24A+17B+10D+7E+5F&=\left(41\frac{\th_2^2\th_4^2\th_4^2(0)}{\th_3^2(0)}+24\frac{\th_4^4\th_4^4(0)}{\th_2^2(0)\th_3^2(0)}+17\frac{\th_2^4\th_2^4(0)}{\th_3^2(0)}\right) \\
&\qquad \times \bigg(\th_2^4(0)+\th_4^4(0)-\th_3^4(0)\bigg),
\end{aligned}
\end{equation}
which is zero by (\ref{evf6d}).  Lemma \ref{residue} is thus proven, and it follows that $\tilde{f}(z,\om)$ is constant in $z$.  To evaluate the constant, we can take $z=0$.

\section{Evaluation of $\tilde{f}(0,\om)$} \label{eval_f}
We will evaluate $\tilde{f}(z,\om)$ at $z=0$.  Notice that many of the functions $h_{jk}(z)$ vanish at $z=0$.  In fact we have
\begin{equation}\label{evf1}
\tilde{f}(0,\om)=Q_{11}h_{11}(0)+Q_{14}h_{14}(0)+Q_{23}h_{23}(0)+Q_{31}h_{31}(0)+Q_{34}h_{34}(0)+Q_{41}h_{41}(0)+Q_{44}h_{44}(0).
\end{equation}
The identities (\ref{evf2}) and (\ref{evf3}) allow us eliminate all derivatives of theta functions from the sum (\ref{evf1}) except $\th_2', \th_2''$.  Making these substitutions and simplifying, (\ref{evf1}) becomes
\begin{equation}\label{evf4}
\begin{aligned}
\tilde{f}(0,\om)=&\frac{\th_1(\om)\big(\th_1^2\th_4^2(0)+\th_2^2\th_3^2(0)-\th_3^2\th_2^2(0)\big)}{96\th_1'(0)}\left(\frac{5\th_2''(0)}{\th_1\th_2\th_3\th_4^3}+\frac{17(\th_2')^2}{\th_1\th_2^3\th_3\th_4^3}-\frac{5\th_2''}{\th_1\th_2^2\th_3\th_4^3}\right. \\
 &+ \left. \frac{24\th_2'\big(\th_3^2\th_2^2(0)+\th_1^2\th_4^2(0)\big)}{\th_1^2\th_2^3\th_3^2\th_4^2}\right)+\frac{\th_1(\om)(24A+17B+7C+5D+3E+2F)}{96\th_1'(0)\th_1^3\th_2^3\th_3^3\th_4^3\th_2^2(0)}
\end{aligned}
\end{equation}
where
\begin{equation}\label{evf5}
\begin{aligned}
A&=\th_2^2\th_4^2\th_2^2(0)\th_3^2(0)\big(\th_1^4\th_4^4(0)+\th_3^4\th_2^4(0)\big), \\ 
B&=\th_4^2\th_2^2(0)\big(\th_1^2\th_4^2(0)-\th_3^2\th_2^2(0)\big)\big((\th_3^2\th_2^2(0)+ \th_1^2\th_4^2(0))^2-\th_1^2\th_3^2\th_2^2(0)\th_4^2(0)\big), \\
C&=\th_2^4\th_4^2\th_2^2(0)\th_3^4(0)\big(\th_1^2\th_4^2(0)-\th_3^2\th_2^2(0)\big), \\
D&=\th_1^2\th_3^2\bigg[\th_4^2\th_2^4(0)\th_4^2(0)\big(\th_3^2\th_2^2(0)-\th_1^2\th_4^2(0)\big)+\th_3^2(0)\bigg(\th_2^2(0)\big(\th_3^4\th_2^4(0)-\th_1^4\th_4^4(0)\big) \\
& \qquad +\th_2^2\th_1^2(0)\th_4^2(0)\big(\th_1^2(\th_4^2(0)+\th_2^2\th_3^2(0)\big)\bigg)\bigg], \\ 
E&=\th_1^2\th_2^2\th_3^2\th_2^2(0)\th_3^2(0)\bigg[\th_4^2(0)\big(\th_1^2\th_3^2(0)-\th_4^2\th_2^2(0)\big)-\th_2^2\th_2^4(0)\bigg], \\
F&=\th_1^2\th_2^2\th_3^2\th_2^2(0)\bigg[4\th_2^2\th_3^2(0)\th_4^4(0)-\th_2^2(0)\big(\th_1^2\th_2^2(0)\th_4^2(0)+\th_3^2\th_4^4(0)+\th_3^2\th_3^4(0)\big)\bigg].
\end{aligned}
\end{equation}
The identity (\ref{evf6a}) implies that all terms in (\ref{evf4}) involving derivatives of theta functions in the numerator vanish.  Additionally, (\ref{evf6a})-(\ref{evf6d}) allow us to simplify the numbers $B, C, D$ and $E$.  Namely, we have
\begin{equation}\label{evf7}
\begin{aligned}
B&=-\th_2^2\th_4^2\th_2^2(0)\th_3^2(0)\big((\th_3^2\th_2^2(0)+ \th_1^2\th_4^2(0))^2-\th_1^2\th_3^2\th_2^2(0)\th_4^2(0)\big),  \\
C&=-\th_2^6\th_4^2\th_2^2(0)\th_3^6(0), \ D=\th_1^2\th_2^2\th_3^2\th_2^2(0)\th_3^2(0)\bigg(\th_1^2\th_3^2(0)\th_4^2(0)+\th_3^2\th_2^2(0)\th_3^2(0)+\th_4^2\th_2^2(0)\th_4^2(0)\bigg), \\
E&=-\th_1^2\th_2^4\th_3^2\th_2^2(0)\th_3^6(0).
\end{aligned}
\end{equation}
Combining these terms gives us 
\begin{equation}\label{evf8}
\begin{aligned}
&24A+17B+7C+5D+3E+2F=\\
&\th_2^2\th_2^2(0)\bigg[-12\th_1^2\th_3^2\th_4^2\th_2^2(0)\th_3^2(0)\th_4^2(0)+7\th_4^2\th_3^2(0)\big(\th_3^4\th_2^4(0)+\th_1^4\th_4^4(0)-\th_2^4\th_3^4(0)\big) \\
&+8\th_1^2\th_2^2\th_3^2\th_3^2(0)\th_4^4(0)+5\th_1^4\th_3^2\th_3^4(0)\th_4^2(0)+3\th_1^2\th_3^2\th_3^4(0)\big(\th_3^2\th_2^2(0)-\th_2^2\th_3^2(0)\big) \\
&-2\th_1^2\th_3^2\th_2^2(0)\th_4^2(0)\big(\th_1^2\th_2^2(0)+\th_3^2\th_4^2(0)\big)\bigg].
\end{aligned}
\end{equation}
Again using (\ref{evf6a})-(\ref{evf6d}), this expression simplifies to
\begin{equation}\label{evf9}
24A+17B+7C+5D+3E+2F=8\th_1^2\th_2^2\th_3^2\th_4^2\th_2^2(0)\th_3^2(0)\th_4^2(0).
\end{equation}
Inserting this into (\ref{evf4}), we get
\begin{equation}\label{evf10}
\tilde{f}(0,\om) \equiv \tilde{f}(z,\om) \equiv f(n\om,\om)=\frac{\th_1(\om)\th_2^2(0)\th_3^2(0)\th_4^2(0)}{12\th_1'(0)\th_1\th_2\th_3\th_4}=\frac{1}{6}
\end{equation}
by (\ref{main13a}) and (\ref{2}).
It then follows that
\begin{equation}\label{evf11}
f_0(n\om,\om)=f(n\om,\om)-\frac{1}{6}=0.
\end{equation}
This proves (\ref{hn7}) and therefore Proposition \ref{h_n2}.

\section{Large $n$ asymptotics of $Z_n$}

By substituting (\ref{hn8}) into (\ref{dph15}) we obtain that
\begin{equation}\label{Z_1}
\begin{aligned}
\frac{\tau_n}{\prod_{k=0}^{n-1} (k!)^2}&=2^{n^2}\prod_{k=0}^{n-1} \frac{h_k}{(k!)^2}\\
&=2^{n^2}h_0\prod_{k=1}^{n-1}\left[G^{2k+1}\,\frac{\th_4((k+1)\om)}{\th_4(k\om)}
\left(1+O(k^{-2})\right)\right]\\
& =C\th_4(n\om)(2G)^{n^2}\bigg(1+O(n^{-1})\bigg),
\end{aligned}
\end{equation}
where $C>0$ does not depend on $n$. Thus, by (\ref{pf7}),
\begin{equation}\label{Z_2}
Z_n=\frac{[\sinh(\ga-t)\sinh(\ga+t)]^{n^2}\tau_n}{\left(
\prod_{k=0}^{n-1}k!\right)^2}=C\th_4(n\om)F^{n^2}\bigg(1+O(n^{-1})\bigg),
\end{equation}
where
\begin{equation}\label{Z_3}
F=2G\sinh(\ga-t)\sinh(\ga+t)=\frac{\pi \sinh(\ga-t)\sinh(\ga+t)\th'_1(0)}{2\ga\th_1(\om)}\,.
\end{equation}
 Theorem \ref{thmain1} is proved.

\medskip

\appendix

\section{Proof of formula (\ref{dph7}) for $\tau_n$} \label{proof_tau_n}
  
We have 
\begin{equation} \label{ap1}
\tau_n=\det\begin{pmatrix}\phi^{(i+j-2)}(t)\end{pmatrix}_{1 \leq i,j \leq n}
\end{equation}    
where
\begin{equation}\label{ap2}
\begin{aligned}
\phi(t)&=\frac{\sinh(2\ga)}{\sinh(\ga+t)\sinh(\ga-t)}
=\frac{2(1-e^{4\ga})}{(1-e^{-2(\ga-t)})(1-e^{-2(\ga+t)})}\\
&=\frac{2}{1-e^{-2(\ga-t)}}+\frac{2}{1-e^{-2(\ga+t)}}-2
=2\sum_{l=0}^\infty e^{-2l(\ga-t)}+2\sum_{l=0}^\infty e^{-2l(\ga+t)}-2\\
&=2\sum_{l=0}^\infty e^{-2l(\ga-t)}+2\sum_{l=1}^\infty e^{-2l(\ga+t)}
=2\sum_{l=-\infty}^\infty e^{2tl}e^{-2\ga |l|}.
\end{aligned}
\end{equation}   
Thus, 
\begin{equation} \label{ap5}
\phi^{(k)}(t)=2\sum_{l=-\infty}^\infty (2l)^k e^{2tl}e^{-2\ga |l|}.
\end{equation} 
From equation (\ref{ap5}) and multi-linearity of the determinant function, we have that $\tau_n$ is equal to
\begin{equation} \label{ap6}
\begin{aligned}
2^n\sum_{l_1,l_2,\dots,l_n=-\infty}^{\infty}
&\det
\begin{pmatrix}
e^{2tl_1-2\ga |l_1|} & (2l_1) e^{2tl_1-2\ga |l_1|} & \cdots & (2l_1)^{n-1} e^{2tl_1-2\ga |l_1|} \\
(2l_2) e^{2tl_2-2\ga |l_2|} & (2l_2)^2 e^{2tl_2-2\ga |l_2|} & \cdots & (2l_2)^{n} e^{2tl_2-2\ga |l_2|} \\
(2l_3)^2 e^{2tl_3-2\ga |l_3|} & (2l_3)^3 e^{2tl_3-2\ga |l_3|} & \cdots & (2l_3)^{n+1} e^{2tl_3-2\ga |l_3|} \\
\vdots & \vdots & \ddots & \vdots \\
(2l_n)^{n-1} e^{2tl_n-2\ga |l_n|} & (2l_n)^{n} e^{2tl_n-2\ga |l_n|} & \cdots & (2l_n)^{2n-2} e^{2tl_n-2\ga |l_n|} \\
\end{pmatrix}\\
&=2^n\sum_{l_1,l_2,\dots,l_n=-\infty}^{\infty} 
\Delta (2l)\prod_{i=1}^{n}e^{2tl_1-2\ga |l_i|} \prod_{i=1}^n(2l_i)^{i-1}\\
&=2^{n^2}\sum_{l_1,l_2,\dots,l_n=-\infty}^{\infty} \prod_{i=1}^{n}
\Delta (l)e^{2tl_1-2\ga |l_i|} \prod_{i=1}^n(l_i)^{i-1},
\end{aligned}
\end{equation}
where $\Delta (l)=\prod_{i<j}(l_i-l_j)$ is the Vandermonde determinant.
Note that, up to sign, this expression for $\tau_n$
is invariant with respect to any permutation of $l_i's$.  So, multiplying by their signs 
and then summing over all permutations, we get
\begin{equation} \label{ap9}
n!\tau_n=2^{n^2}\sum_{l_1,l_2,\dots,l_n=-\infty}^{\infty} \prod_{i=1}^{n}e^{2tl_i-2\ga |l_i|} 
\Delta (l_i)\sum_{\pi\in S_n}(-1)^\pi \prod_{i=1}^{n}(l_i)^{\pi(i)-1},
\end{equation}
thus
\begin{equation} \label{ap10}
\tau_n=\frac{2^{n^2}}{n!}\sum_{l_1,l_2,\dots,l_n=-\infty}^{\infty}  \Delta (l_i)^2\prod_{i=1}^{n}e^{2tl_i-2\ga |l_i|}.
\end{equation}

\section{Derivation of formula (\ref{dph15})} \label{proof_tau_n_2}

Multilinearity of the determinant function, combined with the form of the Vandermonde matrix, allows us 
to replace $\Delta(l)$ with  
\begin{equation} \label{app1}
\det\begin{pmatrix}
1 & 1 & 1 & \cdots & 1 \\
P_{1}(l_{1}) & P_{1}(l_{2}) & P_{1}(l_{3}) & \cdots & P_{1}(l_{n}) \\
P_{2}(l_{1}) & P_{2}(l_{2}) & P_{2}(l_{3}) & \cdots & P_{2}(l_{n}) \\
\vdots&\vdots&\vdots&\vdots&\vdots \\
P_{n-1}(l_{1}) & P_{n-1}(l_{2}) & P_{n-1}(l_{3}) & \cdots & P_{n-1}(l_{n})
\end{pmatrix},
\end{equation}
where $\{P_{j}(x)\}_{j=0}^\infty$ is the system of monic polynomials 
orthogonal  with respect to the weight $w(l)$.
 Then (\ref{dph7}) becomes 
\begin{equation} \label{app2}
\tau_n=\frac{2^{n^2}}{n!}\sum_{l_1,\ldots,l_n=1}^\infty
\left(\sum_{\pi\in S_{n}}(-1)^\pi\prod_{k=1}^{n}P_{\pi(k)-1}(l_{k})\right)^{2}
\prod_{k=1}^{n}w(l_{k}).
\end{equation}
The orthogonality condition ensures that, after summing, only diagonal terms are non-zero, 
so we get 
\begin{equation} \label{app3}
\tau_{n}=\frac{2^{n^2}}{n!}\sum_{l_1,\ldots,l_n=1}^\infty\left(\sum_{\pi\in S_{n}}
\prod_{k=1}^{n}P_{\pi(k)-1}^{2}(l_{k})\right)
\prod_{k=1}^{n}w(l_{k})=2^{n^{2}}\prod_{k=0}^{n-1}h_{k}.
\end{equation}

\section{Proof of (\ref{pmi4})} \label{proof_of_jump}

From, (\ref{pmi3}), (\ref{st1}), (\ref{st2}), and (\ref{g3}), we have that the 
jump $j_{\tilde{Q}}$ on $(\al'-\ep,\al')$ is given by
\begin{equation}\label{j1}
\begin{aligned}
j_{\tilde{Q}}&=e^{\frac{in \pi z}{2\ga}\sigma_3}e^{-n(g_-(z)-\frac{V(z)}{2}-\frac{l}{2})\sigma_3} 
j_S e^{n(g_+(z)-\frac{V(z)}{2}-\frac{l}{2})\sigma_3}e^{\frac{in \pi z}{2\ga}\sigma_3} \\
&=e^{\frac{in \pi z}{2\ga}\sigma_3}e^{-n(g_-(z)-\frac{V(z)}{2}-\frac{l}{2})\sigma_3} 
j_-^{-1} j_T j_+^{-1} e^{n(g_+(z)-\frac{V(z)}{2}-\frac{l}{2})\sigma_3}e^{\frac{in \pi z}{2\ga}\sigma_3} \\
&=e^{\frac{in \pi z}{2\ga}\sigma_3}e^{-n(g_-(z)-\frac{V(z)}{2}-\frac{l}{2})\sigma_3} 
\begin{pmatrix} 
0 & 1 \\ -1 & 0 
\end{pmatrix} 
e^{n(g_+(z)-\frac{V(z)}{2}-\frac{l}{2})\sigma_3}e^{\frac{in \pi z}{2\ga}\sigma_3} \\
&=e^{\frac{in \pi z}{2\ga}\sigma_3}
\begin{pmatrix} 
0 & e^{-n(g_+(z)+g_-(z)-V(z)-l)} \\ -e^{n(g_+(z)+g_-(z)-V(z)-l)} & 0 
\end{pmatrix} 
e^{\frac{in \pi z}{2\ga}\sigma_3} \\
&=\begin{pmatrix} 
0 & 1\\ -1& 0 
\end{pmatrix}. 
\end{aligned}
\end{equation}
From, (\ref{pmi3}), (\ref{st2}), (\ref{ft2}), and (\ref{red5}), we have that 
the jump $j_{\tilde{Q}}$ on $(\al',\al'+\ep)$ is given by
\begin{equation}\label{j2}
\begin{aligned}
j_{\tilde{Q}}&=e^{\frac{in \pi z}{2\ga}\sigma_3}e^{-n(g_-(z)-\frac{V(z)}{2}-\frac{l}{2})\sigma_3} 
j_S e^{n(g_+(z)-\frac{V(z)}{2}-\frac{l}{2})\sigma_3}e^{\frac{in \pi z}{2\ga}\sigma_3} \\
&=e^{\frac{in \pi z}{2\ga}\sigma_3}e^{-n(g_-(z)-\frac{V(z)}{2}-\frac{l}{2})\sigma_3} 
\begin{pmatrix} 
\frac{\ga}{n\pi i}e^{\frac{in\pi z}{2\ga}} & 0 \\ 0 &  \frac{n\pi i}{\ga}e^{-\frac{in\pi z}{2\ga}}
\end{pmatrix}^{-1} 
j_T 
\begin{pmatrix} 
-\frac{\ga}{n\pi i}e^{-\frac{in\pi z}{2\ga}} & 0 \\ 0 &  -\frac{n\pi i}{\ga}e^{\frac{in\pi z}{2\ga}}
\end{pmatrix} \\
& \quad \times e^{n(g_+(z)-\frac{V(z)}{2}-\frac{l}{2})\sigma_3}e^{\frac{in \pi z}{2\ga}\sigma_3} \\
&=e^{-n(g_-(z)-\frac{V(z)}{2}-\frac{l}{2})\sigma_3} 
\begin{pmatrix} 
\frac{n\pi i}{\ga}& 0 \\ 0 &  \frac{\ga}{n\pi i}
\end{pmatrix} 
e^{n(g_-(z)-\frac{l}{2})\sigma_3} j_R e^{-n(g_+(z)-\frac{l}{2})\sigma_3}
\begin{pmatrix} 
-\frac{\ga}{n\pi i} & 0 \\ 0 &  -\frac{n\pi i}{\ga} 
\end{pmatrix} \\ 
&\quad \times e^{n(g_+(z)-\frac{V(z)}{2}-\frac{l}{2})\sigma_3} \\
&=e^{n\frac{V(z)}{2}\sigma_3} 
\begin{pmatrix} 
\frac{n\pi i}{\ga}& 0 \\ 0 &  \frac{\ga}{n\pi i}
\end{pmatrix}  
\begin{pmatrix} 
1 & 0 \\ (\frac{n\pi i}{\ga})^2e^{nV(z)} & 1 
\end{pmatrix} 
\begin{pmatrix} 
-\frac{\ga}{n\pi i} & 0 \\ 0 &  -\frac{n\pi i}{\ga} 
\end{pmatrix}  
e^{-n\frac{V(z)}{2}\sigma_3} \\
&=e^{n\frac{V(z)}{2}\sigma_3} 
\begin{pmatrix} 
-1 & 0 \\ -e^{nV(z)} & -1 
\end{pmatrix}  
e^{-n\frac{V(z)}{2}\sigma_3} =
\begin{pmatrix} 
-1 & 0 \\ -1 & -1 
\end{pmatrix}.  
\end{aligned}
\end{equation}
From (\ref{pmi3}), (\ref{st2}), (\ref{ft2}), (\ref{red5}), and analytic continuation of 
(\ref{g3}) into a neighborhood of $[\al,\al']$, we have that the jump $j_{\tilde{Q}}$ on 
$(\al',\al'+i\ep)$ is given by
\begin{equation}\label{j3}
\begin{aligned}
j_{\tilde{Q}}&=e^{-n(g_+(z)-\frac{V(z)}{2}-\frac{l}{2})\sigma_3} 
\begin{pmatrix} 
-\frac{n\pi i}{\ga} & 0 \\ 0 &  -\frac{\ga}{n\pi i} 
\end{pmatrix} 
e^{n(g_+(z)-\frac{l}{2})\sigma_3} j_R e^{-n(g_+(z)-\frac{l}{2})\sigma_3} 
j_+(z)^{-1} \\
& \quad \times e^{n(g_+(z)-\frac{V(z)}{2}-\frac{l}{2})\sigma_3} e^{\frac{in \pi z}{2\ga}\sigma_3} \\
&=e^{n\frac{V(z)}{2}\sigma_3} 
\begin{pmatrix} 
-\frac{n\pi i}{\ga} & 0 \\ 0 &  -\frac{\ga}{n\pi i} 
\end{pmatrix}  
j_R e^{-n(g_+(z)-\frac{l}{2})\sigma_3} j_+(z)^{-1} e^{n(g_+(z)-\frac{V(z)}{2}-\frac{l}{2})\sigma_3} 
e^{\frac{in \pi z}{2\ga}\sigma_3}\\
&=e^{n\frac{V(z)}{2}\sigma_3} 
\begin{pmatrix} 
-\frac{n\pi i}{\ga} & 0 \\ 0 &  -\frac{\ga}{n\pi i} 
\end{pmatrix}  
\begin{pmatrix} 
\Pi(z) & \frac{\ga}{n\pi i}e^{-nV(z)}e^{\frac{in\pi z}{2\ga}} \\ 
-\frac{n\pi i}{\ga}e^{nV(z)}e^{\frac{in\pi z}{2\ga}} & -\frac{n\pi i}{\ga}e^{\frac{in\pi z}{2\ga}} 
\end{pmatrix} \\
& \quad \times e^{-n(g_+(z)-\frac{l}{2})\sigma_3} 
\begin{pmatrix} 1 & 0 \\ -e^{-n(g_+(z)-g_-(z))} & 1\end{pmatrix} 
e^{n(g_+(z)-\frac{V(z)}{2}-\frac{l}{2})\sigma_3} e^{\frac{in \pi z}{2\ga}\sigma_3} \\
&=e^{n\frac{V(z)}{2}\sigma_3} 
\begin{pmatrix} -2i\sin(\frac{n\pi z}{2\ga}) & -e^{-nV(z)}e^{\frac{in\pi z}{2\ga}} \\ 
e^{nV(z)}e^{\frac{in\pi z}{2\ga}} & e^{\frac{in\pi z}{2\ga}} 
\end{pmatrix} 
\begin{pmatrix} 
e^{-n\frac{V(z)}{2}} & 0 \\ -e^{n(g_+(z)+g_-(z)-l-\frac{V(z)}{2})} & e^{n\frac{V(z)}{2}}
\end{pmatrix}  
e^{\frac{in \pi z}{2\ga}\sigma_3}\\
&= \begin{pmatrix} 
-2i\sin(\frac{n\pi z}{2\ga})e^{n\frac{V(z)}{2}} & -e^{-n\frac{V(z)}{2}}e^{\frac{in\pi z}{2\ga}} \\ 
e^{n\frac{V(z)}{2}}e^{\frac{in\pi z}{2\ga}} & e^{-n\frac{V(z)}{2}}e^{\frac{in\pi z}{2\ga}} 
\end{pmatrix} 
\begin{pmatrix} 
e^{-n\frac{V(z)}{2}} & 0 \\ -e^{n(g_+(z)+g_-(z)-l-\frac{V(z)}{2})} & e^{n\frac{V(z)}{2}}
\end{pmatrix}  e^{\frac{in \pi z}{2\ga}\sigma_3}\\
&= \begin{pmatrix} 
-2i\sin(\frac{n\pi z}{2\ga})+e^{\frac{in\pi z}{2\ga}}e^{n(g_+(z)+g_-(z)-l-V(z))}  & -e^{\frac{in\pi z}{2\ga}} \\ 
e^{\frac{in\pi z}{2\ga}}-e^{\frac{in\pi z}{2\ga}}e^{n(g_+(z)+g_-(z)-l-V(z))}  & e^{\frac{in\pi z}{2\ga}} 
\end{pmatrix} 
e^{\frac{in \pi z}{2\ga}\sigma_3}\\
&= \begin{pmatrix} 
e^{-\frac{in\pi z}{2\ga}}  & -e^{\frac{in\pi z}{2\ga}} \\ 0 & e^{\frac{in\pi z}{2\ga}} 
\end{pmatrix} 
e^{\frac{in \pi z}{2\ga}\sigma_3}=
 \begin{pmatrix} 1  & -1 \\ 0 & 1 \end{pmatrix}.
\end{aligned}
\end{equation}
Similarly,  we have that the jump $j_{\tilde{Q}}$ on $(\al'-i\ep,\al')$ is given by
\begin{equation}\label{j4}
\begin{aligned}
j_{\tilde{Q}}&=e^{\frac{in \pi z}{2\ga}\sigma_3}e^{-n(g_-(z)-\frac{V(z)}{2}-\frac{l}{2})\sigma_3} 
j_S e^{n(g_-(z)-\frac{V(z)}{2}-\frac{l}{2})\sigma_3}e^{-\frac{in \pi z}{2\ga}\sigma_3} \\
&=e^{\frac{in \pi z}{2\ga}\sigma_3}e^{-n(g_-(z)-\frac{V(z)}{2}-\frac{l}{2})\sigma_3} 
\begin{pmatrix}
\frac{\ga}{n\pi i} e^{\frac{in\pi z}{2\ga}} & 0 \\ 0 & \frac{n\pi i}{\ga} e^{-\frac{in\pi z}{2\ga}}
\end{pmatrix}^{-1}  
j_T j_-(z)e^{n(g_-(z)-\frac{V(z)}{2}-\frac{l}{2})\sigma_3}e^{-\frac{in \pi z}{2\ga}\sigma_3} \\
&=e^{n\frac{V(z)}{2}\sigma_3} 
\begin{pmatrix}\frac{n\pi i}{\ga}  & 0 \\ 0 & \frac{\ga}{n\pi i} \end{pmatrix} 
j_R e^{-n(g_-(z)-\frac{l}{2})\sigma_3} j_-(z)e^{n(g_-(z)-\frac{V(z)}{2}-\frac{l}{2})\sigma_3}
e^{-\frac{in \pi z}{2\ga}\sigma_3}\\ 
&=e^{n\frac{V(z)}{2}\sigma_3} 
\begin{pmatrix}\frac{n\pi i}{\ga}  & 0 \\ 0 & \frac{\ga}{n\pi i} \end{pmatrix}
\begin{pmatrix} \Pi(z) & \frac{\ga}{n\pi i}e^{-nV(z)}e^{-\frac{in\pi z}{2\ga}} \\ 
-\frac{n\pi i}{\ga}e^{nV(z)}e^{-\frac{in\pi z}{2\ga}} & \frac{n\pi i}{\ga}e^{-\frac{in\pi z}{2\ga}} 
\end{pmatrix} \\
&\quad \times\begin{pmatrix} 
e^{-n\frac{V(z)}{2}} & 0 \\ e^{n(g_+(z)+g_-(z)-l-\frac{V(z)}{2})} & e^{n\frac{V(z)}{2}} 
\end{pmatrix} 
e^{-\frac{in \pi z}{2\ga}\sigma_3} \\
&=\begin{pmatrix} 
2i\sin(\frac{n\pi z}{2\ga})e^{n\frac{V(z)}{2}} & e^{-n\frac{V(z)}{2}}e^{-\frac{in\pi z}{2\ga}} \\ 
-e^{n\frac{V(z)}{2}}e^{-\frac{in\pi z}{2\ga}} & e^{-n\frac{V(z)}{2}}e^{-\frac{in\pi z}{2\ga}} 
\end{pmatrix} 
\begin{pmatrix} 
e^{-n\frac{V(z)}{2}} & 0 \\ e^{n(g_+(z)+g_-(z)-l-\frac{V(z)}{2})} & e^{n\frac{V(z)}{2}} 
\end{pmatrix} 
e^{-\frac{in \pi z}{2\ga}\sigma_3} \\
&=\begin{pmatrix} 
2i\sin(\frac{n\pi z}{2\ga})+e^{-\frac{in\pi z}{2\ga}}e^{n(g_+(z)+g_-(z)-l-V(z))} & e^{-\frac{in\pi z}{2\ga}} \\ 
-e^{-\frac{in\pi z}{2\ga}} +e^{-\frac{in\pi z}{2\ga}}e^{n(g_+(z)+g_-(z)-l-V(z))} & e^{-\frac{in\pi z}{2\ga}} 
\end{pmatrix} 
e^{-\frac{in \pi z}{2\ga}\sigma_3} \\
&=\begin{pmatrix} 
e^{\frac{in\pi z}{2\ga}} & e^{-\frac{in\pi z}{2\ga}} \\ 0 & e^{-\frac{in\pi z}{2\ga}} 
\end{pmatrix}
 e^{-\frac{in \pi z}{2\ga}\sigma_3} 
=\begin{pmatrix} 1 &1 \\ 0 & 1 \end{pmatrix}.
\end{aligned}
\end{equation}

\section{Proof of Proposition \ref{proposition_model}}  \label{proof_model}
From (\ref{eq18c}) and (\ref{m5}), we have 
\begin{equation}\label{me1}
\tilde{u}_\infty = \frac{\pi}{4}(1-\z)
\end{equation}
This, combined with formula (\ref{st4}) for $\Om_n$ immediately gives
\begin{equation}\label{me2}
\frac{\th_3(\tilde{u}_\infty+d)\th_3(-\tilde{u}_\infty+d+\frac{\Om_n}{2})}{\th_3(-\tilde{u}_\infty+d)\th_3(\tilde{u}_\infty+d+\frac{\Om_n}{2})}=\frac{\th_3(0)\th_4((n+1)\om)}{\th_3(\om)\th_4(n\om)}.
\end{equation}
Formulae (\ref{eq8b}) give that
\begin{equation}\label{me3}
\frac{(\be-\be')+(\al'-\al)}{4i}=\frac{\pi\th_4^2(0)}{4i}\left[\frac{\th_1^2(\frac{\om}{2})\th_4^2(\frac{\om}{2})+\th_2^2(\frac{\om}{2})\th_3^2(\frac{\om}{2})}{\th_1(\frac{\om}{2})\th_2(\frac{\om}{2})\th_3(\frac{\om}{2})\th_4(\frac{\om}{2})}\right].
\end{equation}
Plugging the duplication formulae (\ref{2}) and (\ref{me4})
into (\ref{me3}) yields
\begin{equation}\label{me5}
\frac{(\be-\be')+(\al'-\al)}{4i}=\frac{\pi}{2i}\frac{\th_2(0)\th_4(0)\th_3(\om)}{\th_1(\om)}.
\end{equation}
 Combining (\ref{me2}) and (\ref{me5}), we can write the $[12]$ entry of (\ref{m24}) as
 \begin{equation}\label{me6}
 \begin{aligned}
\left[\bold M_1\right]_{12}
&=\frac{i\pi}{2}\frac{\th_4\big((n+1)\om\big)}{\th_4(n\om)}
\frac{\th_3(0)}{\th_3(\om)}\frac{\th_2(0)\th_4(0)\th_3(\om)}{\th_1(\om)} \\
 &=\frac{i\pi}{2}\frac{\th_4\big((n+1)\om\big)}{\th_4(n\om)}
\frac{\th_1'(0)}{\th_2(\frac{\pi\z}{2})}
=\frac{i A \th_4\big((n+1)\om\big)}{\th_4(n\om)}\,.
 \end{aligned}
 \end{equation}
Similarly, we can write the $[21]$ entry of (\ref{m24}) as
 \begin{equation}\label{me7}
 \left[\bold M_1\right]_{21}=\frac{A \th_4(n\om)}{i\th_4\big((n-1)\om\big)}\,.
 \end{equation}
 
 \section{Theta function identities}\label{identities}
The following identities (see \cite{WW}) are used in this paper.  There are the identities involving derivatives of theta functions:
 \begin{equation}\label{main13a}
\th_1'(0)=\th_2(0)\th_3(0)\th_4(0),
\end{equation}
\begin{equation}\label{4}
\begin{aligned}
\th_4'(z)&=\frac{\th_1'(z)\th_4(z)-\th_4(0)^2\th_2(z)\th_3(z)}{\th_1(z)}, \\
\th_2'(z)&=\frac{\th_1'(z)\th_2(z)-\th_2(0)^3\th_3(z)\th_4(z)}{\th_1(z)}, \\
\th_3'(z)&=\frac{\th_1'(z)\th_3(z)-\th_3(0)^2\th_2(z)\th_4(z)}{\th_1(z)},
\end{aligned}
\end{equation}

\begin{equation}\label{evf3}
\begin{aligned}
\th_4'(z)&=\frac{\th_2'(z)\th_4(z)+\th_1(z)\th_3(z)\th_3^2(0)}{\th_2(z)}, \\
\th_1'(z)&=\frac{\th_2'(z)\th_1(z)+\th_3(z)\th_4(z)\th_2^2(0)}{\th_2(z)}, \\
\th_3'(z)&=\frac{\th_2'(z)\th_3(z)+\th_1(z)\th_4(z)\th_4^2(0)}{\th_2(z)}, \\
\end{aligned}
\end{equation}

\begin{equation}\label{evf2}
\begin{aligned}
\th_4''(z)&=\frac{\th_2''(z)\th_4(z)}{\th_2(z)}+\frac{2\th_2'(z)\th_1(z)\th_3(z)\th_3^2(0)}{\th_2^2(z)}+\frac{\th_4(z)\th_3^2(0)}{\th_2^2(z)}\bigg(\th_3^2(z)\th_2^2(0)+\th_1^2(z)\th_4^2(0)\bigg), \\
\th_1''(z)&=\frac{\th_2''(z)\th_1(z)}{\th_2(z)}+\frac{2\th_2'(z)\th_3(z)\th_4(z)\th_2^2(0)}{\th_2^2(z)}+\frac{\th_1(z)\th_2^2(0)}{\th_2^2(z)}\bigg(\th_3^2(z)\th_3^2(0)+\th_4^2(z)\th_4^2(0)\bigg), \\
\th_3''(z)&=\frac{\th_2''(z)\th_3(z)}{\th_2(z)}+\frac{2\th_2'(z)\th_1(z)\th_4(z)\th_4^2(0)}{\th_2^2(z)}+\frac{\th_3(z)\th_3^2(0)}{\th_2^2(z)}\bigg(\th_4^2(z)\th_2^2(0)+\th_1^2(z)\th_4^2(0)\bigg),
\end{aligned}
\end{equation}
\begin{equation}\label{Q23}
\begin{aligned}
\th_4''(z)&=\frac{\th_1''(z)\th_4(z)}{\th_1(z)}-\frac{2\th_1'(z)\th_2(z)\th_3(z)\th_4^2(0)}{\th_1^2(z)}+\frac{\th_4^2(0)\th_4^2(z)}{\th_1^2(z)}\bigg(\th_3^2(z)\th_2^2(0)+\th_2^2(z)\th_3^2(0)\bigg), \\
\th_2''(z)&=\frac{\th_1''(z)\th_2(z)}{\th_1(z)}-\frac{2\th_1'(z)\th_3(z)\th_4(z)\th_2^2(0)}{\th_1^2(z)}+\frac{\th_2^2(0)\th_2^2(z)}{\th_1^2(z)}\bigg(\th_4^2(z)\th_3^2(0)+\th_3^2(z)\th_4^2(0)\bigg), \\
\th_3''(z)&=\frac{\th_1''(z)\th_3(z)}{\th_1(z)}-\frac{2\th_1'(z)\th_2(z)\th_4(z)\th_3^2(0)}{\th_1^2(z)}+\frac{\th_3^2(0)\th_3^2(z)}{\th_1^2(z)}\bigg(\th_4^2(z)\th_2^2(0)+\th_2^2(z)\th_4^2(0)\bigg),
\end{aligned}
\end{equation}
the duplication formulae
\begin{equation}\label{2}
\th_1(2z)=\frac{2\th_1(z)\th_2(z)\th_3(z)\th_4(z)}{\th_1'(0)},
\end{equation}
\begin{equation}\label{me4}
\th_3(2z)\th_3(0)\th_2^2(0)=\th_1(z)^2\th_4(z)^2+\th_2(z)^2\th_3(z)^2 
\end{equation}
\begin{equation}\label{Q26}
\th_4(2z)\th_4^3(0)=\th_3^4(z)-\th_2^4(z)=\th_4^4(z)-\th_1^4(z),
\end{equation}
the addition formula
\begin{equation}\label{Q11}
\begin{aligned}
\th_3(y+z)\th_3(y-z)\th_2^2(0)&=\th_3^2(y)\th_2^2(z)+\th_4^2(y)\th_1^2(z) \\
&=\th_1^2(y)\th_4^2(z)+\th_2^2(y)\th_3^2(z),
\end{aligned}
\end{equation}
and the identities relating squares of theta functions

\begin{equation}\label{evf6a}
\th_1^2(z)\th_4^2(0)=\th_3^2(z)\th_2^2(0)-\th_2^2(z)\th_3^2(0),
\end{equation}
\begin{equation}\label{evf6b}
\th_2^2(z)\th_4^2(0)=\th_4^2(z)\th_2^2(0)-\th_1^2(z)\th_3^2(0),
\end{equation}
\begin{equation}\label{evf6c}
\th_3^2(z)\th_4^2(0)=\th_4^2(z)\th_3^2(0)-\th_1^2(z)\th_2^2(0),
\end{equation}
\begin{equation}\label{evf6d}
\th_4^2(z)\th_4^2(0)=\th_3^2(z)\th_3^2(0)-\th_2^2(z)\th_2^2(0).
\end{equation}

\end{document}